\documentclass{article}
\usepackage{amsfonts}
\usepackage{amssymb}
\usepackage{graphicx}
\usepackage{amsmath}
\usepackage{amsthm}
\usepackage{pifont}
\usepackage[colorlinks=true,linkcolor=blue,citecolor= red,bookmarksopen=true,urlcolor=blue]{hyperref}
\nonstopmode
\newtheorem{theorem}{Theorem}[section]
\newtheorem{proposition}[theorem]{Proposition}
\newtheorem{lemma}[theorem]{Lemma}
\newtheorem*{lemma*}{Lemma}

\newtheorem{corollary}[theorem]{Corollary}

\newtheorem{definition}[theorem]{Definition}

\newtheorem*{standingassumptionI*}{Standing Assumption I}
\newtheorem*{standingassumptionII*}{Standing Assumption II}
\newtheorem{assumption}[theorem]{Assumption}

\theoremstyle{remark}

\newtheorem{remark}[theorem]{Remark}
\newcommand{\probp}{{\mathbb P}}
\newcommand{\probq}{{\mathbb Q}}
\newcommand{\R}{{\mathbb R}}

\newcommand{\cF}{\mathcal{F}}

\newcommand{\norm}[1]{\left\lVert#1\right\rVert}

\newcommand\Ep[1]{\mathbb{E}_\mathbb{P} \left[#1\right]}

\newcommand{\rn}[2]{\frac{\mathrm{d}#1}{\mathrm{d}#2}}

\DeclareMathOperator*{\argmax}{arg\,max}
\begin{document}

\title{Multivariate Systemic Optimal Risk Transfer Equilibrium }
\date{\today}
\author{Alessandro Doldi\thanks{%
Dipartimento di Matematica, Universit\`{a} degli Studi di Milano, Via
Saldini 50, 20133 Milano, Italy, $\,\,$\emph{alessandro.doldi@unimi.it}. }
\and Marco Frittelli\thanks{%
Dipartimento di Matematica, Universit\`a degli Studi di Milano, Via Saldini
50, 20133 Milano, Italy, \emph{marco.frittelli@unimi.it}. }}
\maketitle

\begin{abstract}
\noindent A Systemic Optimal Risk Transfer Equilibrium (SORTE) was
introduced in: \textquotedblleft Systemic optimal risk transfer equilibrium\textquotedblright, Mathematics and
Financial Economics (2021), for the analysis of the equilibrium among
financial institutions or in insurance-reinsurance markets. A SORTE
conjugates the classical B\"{u}hlmann's notion of a risk exchange equilibrium with a capital allocation principle based on systemic expected
utility optimization. In this paper we extend such a notion to the case when
the value function to be optimized is multivariate in a general sense, and
it is not simply given by the sum of univariate utility functions. This
takes into account the fact that preferences of single agents might depend
on the actions of other participants in the game. Technically, the extension
of SORTE to the new setup requires developing a theory for multivariate
utility functions and selecting at the same time a suitable framework for
the duality theory. Conceptually, this more general framework allows us to
introduce and study a Nash Equilibrium property of the optimizer. \noindent
We prove existence, uniqueness, and the Nash Equilibrium property of the
newly defined Multivariate Systemic Optimal Risk Transfer Equilibrium.
\end{abstract}


\parindent0em \noindent 
\textbf{Acknowledgment:} the authors would like to thank Achraf Tamtalini
for pointing out the necessity of closedness under truncation for the
arguments in Section \ref{mSORTEworkonlinfty}.

\section{Introduction}

A Systemic Optimal Risk Transfer Equilibrium, denoted with SORTE, was
introduced and analyzed in Biagini et al. (2021) \cite{BDFFM}. The SORTE
concept was inspired by B\"{u}hlmann's notion of a Risk Transfer Equilibrium
in insurance-reinsurance markets. However, in B\"{u}hlmann's definition the
vector assigning the budget constraints was given a priori. On the contrary,
in the SORTE, such a vector is endogenously determined by solving a systemic
utility maximization problem. As remarked in \cite{BDFFM}, \textquotedblleft 
\textit{SORTE gives priority to the systemic aspects of the problem, in
order to optimize the overall systemic performance, rather than to
individual rationality\textquotedblright }. For the precise definition of a
SORTE, its existence, uniqueness and Pareto optimality we refer to \cite%
{BDFFM}. In Section \ref{mSORTESORTE} we will only very briefly recall its
motivation and heuristic definition in order to compare it with the results
of the present paper. We will not address any integrability issues in this
Introduction.

\bigskip

The capital allocation and risk sharing equilibrium that we consider in this
new work, similarly to the one introduced in \cite{BDFFM}, can be applied to
many contexts, such as: equilibrium among financial institutions, agents, or
countries; insurance and reinsurance markets; capital allocation among
business units of a single firm; wealth allocation among investors.

\textit{The key novelty in this work is that we consider preferences of agents which
depend on other agents' choices. This is modeled using multivariate utility
functions}.

Let $(\Omega ,\mathcal{F},{\mathbb{P}}\mathbf{)}$ be a probability space and 
$L^{0}(\Omega ,\mathcal{F},{\mathbb{P}})$ be the vector space of
(equivalence classes of) real valued $\mathcal{F}$-measurable random
variables. The sigma algebra $\mathcal{F}$ represents all possible
measurable events at the terminal time. $\mathbb{E}_{\mathbb{Q}}\left[ \cdot %
\right] $ denotes the expectation under a probability $\mathbb{Q}$. For the
sake of simplicity, we are assuming zero interest rate.

In a one period setup we consider $N$ agents. The individual risk (or the
random endowment or the future profit and loss) plus the initial wealth of
each agent is represented by the random variable $X^{j}\in L^{0}(\Omega ,%
\mathcal{F},{\mathbb{P}})$. Thus the risk vector $X=[X^{1},...,X^{N}]\in
(L^{0}(\Omega ,\mathcal{F},{\mathbb{P}}))^{N}$ denotes the original
configuration of the system.

We assume that the system has at disposal a total amount of capital $A\in 
\mathbb{R}$ to be used in case of necessity. This amount could have been
assigned by the Central Bank, or could have been the result of the previous
trading in the system, or could have been collected ad hoc by the agents.
The amount A could represent an insurance pot or a fund collected (as
guarantee for future investments) in a community of homeowners. For further
interpretation of $A$, see also the related discussion in Section 5.2 of
Biagini et al. (2020) \cite{bffm}. 
In any case, \textit{we consider the quantity $A$
as exogenously determined and our aim is to allocate this amount among the
agents in order to optimize the overall systemic satisfaction}.

In this paper we work with a multivariate utility function $U:{\mathbb{R}}%
^{N}\rightarrow {\mathbb{R}}$ that is strictly concave and strictly
increasing with respect to the partial componentwise order. However, some results (see Corollary \ref{corollarystrict}) hold also without the \textit{strict} concavity or the \textit{strict} monotonicity. We develop a condition on the multivariate utility $U$, see Definition \ref{mSORTEwellcontrolled}, that will play the same role as the Inada conditions in the one dimensional case. Details and
precise assumptions are deferred to Section \ref{mSORTEsubsecmSORTE} and Section \ref{mSORTEsecsetupassumptions}.
Examples of utility functions $U$
satisfying our assumptions are collected in Section \ref{mSORTEsecexamples}.

\bigskip


\textbf{Systemic optimal (deterministic) initial-time allocation}

If we denote with $a^{j}\in \mathbb{R}$ the cash received (if positive) or
provided (if negative) by agent $j$ at initial time, then the time $T$
wealth at disposal of agent $j$ will be $X^{j}+a^{j}$. The optimal
allocation ${a_{{X}}\in }\mathbb{R}^{N}$ could then be determined as the
solution to the following aggregate criterion 
\begin{equation}
\Pi _{A}^{\det }(X):=\sup \left\{ \mathbb{E}\left[ U(X+a)\right] \mid {a\in }%
\mathbb{R}^{N}\text{ s.t. }\sum_{j=1}^{N}a^{j}=A\right\} .  \notag
\end{equation}%
As the vector ${a_{{X}}\in }\mathbb{R}^{N}$ is deterministic, it is known at
the initial time and therefore the allocation is distributed (only) at such
initial time and this is to the advantage of each agent. Indeed, if the
agent $j$ receives the fund ${a_{{X}}^{j}}$ at initial time, the agent may
use this amount to prevent financial ruin of future default.

By the monotonicity of $U$, we may formalize the budget constraints set in
the utility maximization problem (here and below) using equivalently the
equality $\sum_{j=1}^{N}a^{j}=A$ or the inequality $\sum_{j=1}^{N}a^{j}\leq
A $.

\bigskip

\textbf{Systemic optimal (random) final-time allocation}

We are now going to replace in $\Pi _{A}^{\det }$ the constant vectors ${%
a\in }\mathbb{R}^{N}$ with random vectors $Y=[Y^{1},...,Y^{N}]\in
(L^{0}(\Omega ,\mathcal{F},{\mathbb{P}}))^{N}$ representing final-time
random allocations. Set\ 
\begin{equation*}
\mathcal{C}_{\mathbb{R}}:=\left\{ Y\in (L^{0}(\Omega ,\mathcal{F},\mathbb{P}%
))^{N}\mid \sum_{j=1}^{N}Y^{j}\in \mathbb{R}\right\}
\end{equation*}%
and note that each component $Y^{j}$ of the vector $Y\in \mathcal{C}_{%
\mathbb{R}}$ is a random variable (measurable with respect to the sigma
algebra at the terminal time), but the sum of the components is $\,{\mathbb{P%
}}$-a.s. equal to some constant in $\mathbb{R}$. 
 We may
impose additional constraints on the vectors $Y$ of random allocations by
requiring that they belongs further to a prescribed set $%
\mathcal{B}$ of \textit{feasible allocations}. It will be assumed that
\begin{equation}
\mathbb{R}^{N}\subseteq \mathcal{B}\subseteq \mathcal{C}_{\mathbb{R}}
\end{equation}
and
that $\mathcal{B}$ is translation invariant: $\mathcal{B}+{\mathbb{R}}^{N}=%
\mathcal{B}$.
We consider  a family of probability vectors $\mathbb{Q}:=[\mathbb{Q}^{1},...,\mathbb{Q}^{N}]\in \mathcal{Q}_{\mathcal{B}%
,V}$ (see \eqref{mSORTEQV}) associated to $\mathcal{B}$ and to the convex
conjugate $V$ of $U$, and we take  $\mathcal{L}:=\bigcap_{{\mathbb{Q}}\in \mathcal{Q}_{\mathcal{B}%
,V}}L^{1}({\mathbb{Q}})$ for $L^{1}({{\mathbb{Q}}):=}L^{1}(\Omega,\cF,{\mathbb{Q}}^{1}{%
)\times ...\times }L^{1}(\Omega,\cF,{\mathbb{Q}}^{N})$.

With these notations, a different possibility to allocate the amount $A$
among the agents is to consider the following criterion 
\begin{equation}
\Pi _{A}^{\text{ran}}(X):=\sup \left\{ \mathbb{E}\left[ U(X+Y)\right] \mid
Y\in \mathcal{B}\cap \mathcal{L}\text{ s.t. }\sum_{j=1}^{N}Y^{j}=A,\text{ }%
\mathbb{P}-{a.s.}\right\} .  \label{mSORTEproblem random}
\end{equation}%
It is clear that $\Pi _{A}^{\det }(X)\leq \Pi _{A}^{\text{ran}}(X)$, thus
random allocations realizes, as obvious, a greater systemic expected
utility, as the dependence among the components $X^{j}$ of the original risk
can be taken into account by the terms $X^{j}+Y^{j}$. The condition $%
\sum_{j=1}^{N}Y^{j}=A$ is instrumental for the allocation of the amount $A$.

The
optimization problem in \eqref{mSORTEproblem random} can be seen as
the maximization of systemic utility for the allocation of the amount $A$
over feasible allocations $Y\in \mathcal{B}$, in a regulatory approach.
Indeed, only the utility of the whole system is taken into account in %
\eqref{mSORTEproblem random}, while optimality for single agents is
not required. The problem \eqref{mSORTEproblem random} is
similar in spirit to classical risk sharing problems (see Barrieu and El
Karoui (2005) \cite{BarrieuElKaroui05}). Unlike in the classical risk
sharing problems, we have a multivariate value function in place of the
classical sum of univariate ones. 

We observe that the \textquotedblleft budget\textquotedblright\ constraints
in $\Pi _{A}^{\text{ran}}$ are not expressed in the classical way using
expectation under some (or many) probability measures, but are instead
formalized as $\mathbb{P}-{a.s.}$ equality. Only in case $N=1$ the problem
becomes trivial, as $\Pi _{A}^{\text{ran}}(X)=\mathbb{E}\left[ U(X+A)\right]
.$

On a technical level our first main result of this paper in the detailed
study of the problem $\Pi _{A}^{\text{ran}}.$ We first show in Theorem \ref%
{mSORTEthmmaingeneral1} that $\Pi _{A}^{\text{ran}}(X)$ can be rewritten
with the budget contraint assigned by the family of probability vectors $\mathcal{Q}_{\mathcal{B}%
,V}$, namely%
\begin{equation}
\Pi _{A}^{\text{ran}}(X)=\sup \left\{ \mathbb{E}_{\mathbb{P}}\left[ U(X+Y)%
\right] \mid Y\in \mathcal{L},\,\sum_{j=1}^{N}\mathbb{E}_{\mathbb{Q}^{j}}%
\left[ Y^{j}\right] \leq A\,\forall {\mathbb{Q}}\in \mathcal{Q}_{\mathcal{B}%
,V}\right\}\,. \label{PiPi}
\end{equation}%
We prove (Theorem \ref%
{mSORTEthmmaingeneral1} and Corollary \ref{mSORTEcoroptsumA}): (i) the existence of the optimizer $\widehat{Y}\in 
\mathcal{L}$ of the problem in (\ref{PiPi}); (ii) its dual formulation as a
minimization problem over $\mathcal{Q}_{\mathcal{B},V}$; (iii) the existence
of the optimizer $\widehat{{\mathbb{Q}}}\in \mathcal{Q}_{\mathcal{B},V}$ of
such dual formulation; (iv) that \eqref{mSORTEproblem random} and \eqref{PiPi} have the same optimizer. Additionally, in \eqref{mSORTEeq666} we prove that for such an optimizer $\widehat{{\mathbb{Q}}}$ we have
\begin{equation}
\Pi _{A}^{\text{ran}}(X)=\sup \left\{ \mathbb{E}_{\mathbb{P}}\left[ U(X+Y)%
\right] \mid Y\in \mathcal{L},\,\sum_{j=1}^{N}\mathbb{E}_{\widehat{{\mathbb{Q%
}}}^{j}}[Y^{j}]\leq A\right\} .
\label{PiQ}
\end{equation}

We now present some more conceptual motivation for the analysis of this
problem.

As stated above, $\Pi _{A}^{\text{ran}}(X)$ is greater than $\Pi _{A}^{\det
}(X)$ and the random variables $Y=[Y^{1},...,Y^{N}]$ in $\Pi _{A}^{\text{ran}%
}(X)$ are terminal time allocations, as they are $\mathcal{F}$-measurable.
However, obviously one may split $Y$ in two components 
\begin{equation}
Y=a+(Y-a):=a+\widetilde{Y}\text{, }  \label{split}
\end{equation}%
for some $a\in \mathbb{R}^{N}$ such that $\sum_{j=1}^{N}a^{j}=A$, which then
represents an \textit{initial capital allocation} $a=(a^1,...,a^N)$ of $A$ and a \textit{terminal
time risk exchange} $\widetilde{Y}$ satisfying $\sum_{j=1}^{N}\widetilde{Y}%
^{j}=0,$ as $A=\sum_{j=1}^{N}Y^{j}=\sum_{j=1}^{N}a^{j}+\sum_{j=1}^{N}%
\widetilde{Y}^{j}=A+\sum_{j=1}^{N}\widetilde{Y}^{j}.$ 
 We pose two natural questions:

\begin{enumerate}
\item Is there an \textquotedblleft optimal\textquotedblright\ way to select
such initial capital $a\in \mathbb{R}^{N}$ ?

\item Could we discover a risk exchange equilibrium among the agents that is
embedded in the problem $\Pi _{A}^{\text{ran}}?$
\end{enumerate}
From the
formulation in (\ref{PiQ}), one could conjecture that the amount $\widehat{a}%
^{j}$ assigned as the expectation of the optimizer of $\Pi _{A}^{\text{ran}%
}(X)$ under the probability $\widehat{{\mathbb{Q}}}^{j},$ namely $\widehat{a}%
^{j}:=\mathbb{E}_{\widehat{{\mathbb{Q}}}^{j}}[\widehat{Y}^{j}]$, could have
a special relevance.
We will show indeed that the optimal solution to the above problem $\Pi
_{A}^{\text{ran}}(X)$ coincides with a multivariate version of the Systemic
Optimal Risk Transfer Equilibrium (SORTE) introduced in Biagini et al.
(2020)\ \cite{BDFFM}.

In order to answer these questions more precisely we need to recall the
notion of a risk exchange equilibrium, as proposed by B\"{u}hlmann
(1980) \cite{Buhlmann1} and (1984) \cite{Buhlmann}.

\bigskip

\subsection{Risk exchange equilibrium and Systemic Optimal Risk Transfer Equilibrium \label{mSORTESORTE}}

\textbf{}

We recall that in this paper we work with a multivariate utility function
but, in order to illustrate the risk exchange equilibrium and the SORTE concepts, in this subsection
we assume that the preferences of each agent $j$ are given via expected
utility, by a strictly concave, increasing utility function $u_{j}:\mathbb{%
R\rightarrow R}$, $j=1,...,N.$ In this case the corresponding multivariate utility function would be  $U(x):=\sum_{j=1}^{N}u_{j}(x^j),\quad x\in {\mathbb{R}}%
^{N}$. The vector $X=(X^{1},...,X^{N})\in
(L^{0}(\Omega ,\mathcal{F},{\mathbb{P}}))^{N}$ denotes the original risk
configuration of the system (the individual risk plus the initial wealth)
and each agent is an expected utility maximizer. At terminal time a
reinsurance mechanism is allowed to happen, in that each agent $j$ agrees to
receive (if positive) or to provide (if negative) the amount $\widetilde{Y}%
^{j}(\omega )$ at the final time in exchange of the amount $\mathbb{E}_{%
\mathbb{Q}^{j}}\left[ \widetilde{Y}^{j}\right] $ paid (if positive) or
received (if negative) at the initial time, where $\mathbb{Q}:=[\mathbb{Q}%
^{1},...,\mathbb{Q}^{N}]$ is some pricing probability vector (the
equilibrium price system). The reinsurance nature of this reallocation comes
from the fact that the clearing condition
\begin{equation}
\label{mSORTEclearing1}
\sum_{j=1}^{N}\widetilde{Y}^{j}=0\,\,\,{\mathbb{P}}-a.s.,
\end{equation}%
is required to hold, which models a terminal-time risk transfer mechanism. Integrability or
boundedness conditions on $\widetilde{Y}^{j}$ will be added later when we
rigorously formalize the setting. We use the $\sim$ in the notation $\widetilde{Y}$ for the sake of consistency with the previous work \cite{BDFFM}. 

\bigskip
As defined in B\"{u}hlmann (1980) \cite{Buhlmann1} and (1984) \cite{Buhlmann}%
, a pair ($\widetilde{Y}_{{X}},{\mathbb{Q}}_{{X}})$ is a \textbf{risk
exchange equilibrium} \textbf{w.r.to the vector} $X$ if:

($\alpha $) for each $j$, $\widetilde{Y}_{{X}}^{j}$ maximizes: $\mathbb{E}_{{%
\mathbb{P}}}\left[ u_{j}(X^{j}+\widetilde{Y}^{j}-\mathbb{E}_{{\mathbb{Q}}_{{X%
}}^{j}}[\widetilde{Y}^{j}])\right] $ among all r.v. $\widetilde{Y}^{j}$;

($\beta $) $\sum_{j=1}^{N}\widetilde{Y}_{{X}}^{j}=0$ $\mathbb{P}$-a.s.;

($\gamma $) $\mathbb{Q}_{{X}}^{1}=...=\mathbb{Q}_{{X}}^{N}.$

\bigskip
The optimal value in ($\alpha $) is denoted by 
\begin{equation}
V^{{\mathbb{Q}}_{{X}}^{j}}(X^{j})=\sup_{\widetilde{Y}^{j}}\left\{ \mathbb{E}%
_{{\mathbb{P}}}\left[ u_{j}(X^{j}+\widetilde{Y}^{j}-\mathbb{E}_{{\mathbb{Q}}%
_{{X}}^{j}}[\widetilde{Y}^{j}])\right] \right\} .  \label{VQ}
\end{equation}

\begin{remark}
\label{mSORTEremclearingnotinopt copy(1)} A key point of B\"{u}hlmann's risk
exchange equilibrium, which carries over to SORTE and mSORTE, is that in ($%
\alpha $) the single agent $j$ is optimizing over \textit{all possible
random variables} $\widetilde{Y}^{j}$ and not over only those that satisfies
the clearing condition ($\beta $). Indeed for the single myopic agent the
clearing condition is irrelevant. Observe that if in ($\alpha $) we consider
a generic probability vector $\mathbb{Q}$, the solutions $\widetilde{Y}_{{X}%
}^{j}$ of the single $N$ problems in ($\alpha$) will typically not satisfy the
clearing condition ($\beta $). It is only the selection of the equilibrium
pricing vector ${\mathbb{Q}}_{{X}}$ in ($\alpha $) that permits to comply
with the clearing condition ($\beta $).
\end{remark}

\begin{remark}
Differently from B\"{u}hlmann's notion we will also impose that the exchange
vector $\widetilde{Y}$ belongs to a prescribed set $\mathcal{B}$ of \textit{%
feasible allocations}, as already mentioned. If there are no further
constraints, i.e. $\mathcal{B}=\mathcal{C}_{\mathbb{R}}$, then $\mathbb{Q}_{{%
X}}^{1}=...=\mathbb{Q}_{{X}}^{N}$ \ holds in B\"{u}hlmann's risk exchange equilibrium. But the
presence of the constraints on $\widetilde{Y},$ represented by $\mathcal{B}$%
, forces to abandon the condition ($\gamma $) in B\"{u}hlmann's risk exchange equilibrium and to
allow for a generic vector $\mathbb{Q}_{X}:=[\mathbb{Q}_{X}^{1},...,\mathbb{Q%
}_{X}^{N}]$. A detailed discussion on several potentially different pricing
measures in B\"{u}hlmann's risk exchange equilibrium with feasible allocation set $\mathcal{B}\neq \mathcal{C}_\R$, as well as in SORTE, can be found in the
Introduction of \cite{BDFFM}.
\end{remark}

\bigskip
This prompt us to define a $\mathcal{B}$-\textbf{risk exchange equilibrium} 
\textbf{w.r.to the vector} $X$ as a pair ($\widetilde{Y}_{{X}},{\mathbb{Q}}_{%
{X}})$ satisfying:

($\alpha $) for each $j$, $\widetilde{Y}_{{X}}^{j}$ maximizes: $\mathbb{E}_{{%
\mathbb{P}}}\left[ u_{j}(X^{j}+\widetilde{Y}^{j}-\mathbb{E}_{{\mathbb{Q}}_{{X%
}}^{j}}[\widetilde{Y}^{j}])\right] $ among all r.v. $\widetilde{Y}^{j}$;

($\beta $) $\sum_{j=1}^{N}\widetilde{Y}_{{X}}^{j}=0$ $\mathbb{P}$-a.s. and $%
\widetilde{Y}\in \mathcal{B}$.

\bigskip
After this review of the concept of risk exchange equilibrium, we now  return to our problem of allocating the amount $A\in {\mathbb{R}}$.
Observe that if $a\mathbf{\in }\mathbb{R}^{N}$ is allocated at initial time
among the agents and $\sum_{j=1}^{N}a^{j}=A$ then the initial risk
configuration of each agent becomes $(a^{j}+X^{j})$ and they may enter in a
risk exchange equilibrium w.r. to such modified vector $(a_X+X).$
\bigskip

According to \cite{BDFFM}, a triple $(\widetilde{Y}_{{X}},{\mathbb{Q}}_{{X}%
},a_{{X}})$ is a \textbf{Systemic Optimal Risk Transfer Equilibrium} (SORTE)
with budget $A\in {\mathbb{R}}$ if:

(a) the pair ($\widetilde{Y}_{{X}},{\mathbb{Q}}_{{X}})$ is a $\mathcal{B}$%
-risk exchange equilibrium w.r.to the vector $(a_X+X)$,

(b) $a_{{X}}\mathbf{\in }\mathbb{R}^{N}$ maximizes $%
\sum_{j=1}^{N}V^{{\mathbb{Q}}_{{X}}^{j}}(a^{j}+X^{j})$ among all $a\mathbf{\in }\mathbb{R}^{N}$ s.t. $%
\sum_{j=1}^{N}a^{j}=A$.

Thus the optimal value in (b) equals
\begin{equation}
\sum_{j=1}^{N}V^{{\mathbb{Q}}_{{X}}^{j}}(a_{{X}}^{j}+X^{j})=\sup \left\{ \sum_{j=1}^{N}\sup_{%
\widetilde{Y}^{j}}\left\{ \mathbb{E}_{{\mathbb{P}}}\left[ u_{j}(a^{j}+X^{j}+%
\widetilde{Y}^{j}-\mathbb{E}_{{\mathbb{Q}}_{{X}}^{j}}[\widetilde{Y}^{j}])%
\right] \right\} \mid a\mathbf{\in }\mathbb{R}^{N}\text{ s.t.}%
\sum_{j=1}^{N}a^{j}=A\right\} . \label{VVV}
\end{equation}

Observe that SORTE explains how optima can be
realized conjugating optimality for the system as a whole (optimization over 
$a\in {\mathbb{R}}^{N}$ in \eqref{VVV}) and convenience for
single agents (the inner supremum in \eqref{VVV}). Under fairly general assumptions, existence, uniqueness and Pareto
Optimality of a SORTE were proven in \cite{BDFFM}.

In this paper we propose a multivariate version of this concept, which we
label with mSORTE, and prove that the optimizer $\widehat{Y}$ of $\Pi _{A}^{%
\text{ran}}(X)$, the optimizer $\widehat{{\mathbb{Q}}}$ of the dual
formulation of $\Pi _{A}^{\text{ran}}(X)$, the selection $\widehat{a}^{j}:=%
\mathbb{E}_{^{\widehat{{\mathbb{Q}}}^{j}}}[\widehat{Y}^{j}]$ in the
splitting (\ref{split}) determine the (unique) mSORTE. 
\subsection{Multivariate Systemic Optimal Risk Transfer Equilibrium}

\label{mSORTEsubsecmSORTE} Essentially, in this paper we answer to the following
question: what can we say about a concept of equilibrium similar to SORTE, with the
same underlying exchange dynamics, when preferences of each agent depend on
the actions of the other agents in the system? In our analysis, we consider
a multivariate utility function $U:{\mathbb{R}}^{N}\rightarrow {\mathbb{R}}$
for the system. {Mildly speaking, $U$ is a utility associated to the system
as a whole. $U$ determines the preferences of single agents in the system,
who take into account the actions and choices of the others: if the random
vector $Z^{[-j]}=[Z^{1},\dots ,Z^{j-1},Z^{j+1},\dots ,Z_{N}]$ models the
positions of agents $1,\dots ,j-1,j+1,\dots ,N$, we suppose that each agent $%
j$ is an expected utility maximizer, in the sense that he/she seeks $%
W\mapsto \max \mathbb{E}_{\mathbb{P}}\left[ U([Z^{[-j]};W])\right] $ where $%
[Z^{[-j]};W]:=\left[ Z^{1},\dots ,Z^{j-1},W,Z^{j+1},\dots ,Z^{N}\right] $. $%
U $ can be thought as nontrivial aggregation of preferences of single
agents: if each agent has preferences given by univariate utility functions $%
u_{1},\dots ,u_{N}:{\mathbb{R}}\rightarrow {\mathbb{R}}$, then given an
aggregation function $\Gamma :{\mathbb{R}}^{N}\rightarrow {\mathbb{R}}$
which is concave and nondecreasing we can consider $U(x)=\Gamma
(u_{1}(x^{1}),\dots ,u_{N}(x^{N}))$, in the spirit of Liebrich and Svindland
(2019) \cite{LiebrichSvindland19} Definition 4, as a natural candidate for $%
U $. Alternatively, $U$ is a counterpart to the multivariate loss function $%
\ell $ considered e.g. in \cite{Drapeau} in the framework of Systemic Risk
Measures. Just setting $U(x)=-\ell (-x),$ $x\in {\mathbb{R}}^{N},$ produces
a natural candidate for our $U$ starting from a loss function $\ell $. The
difference between loss functions and utility is not conceptual here, being
instead just an effect of considering as positive amounts losses (as in the
case of $\ell $) or gains (as in the case of standard utility functions).} 
Multivariate utility functions could be employed also to describe the case of a single firm having investments in $N$ units, where the interconnections among the $N$ desks are relevant.

We will impose on our multivariate utility function $U$ conditions which play the same role of Inada conditions in the univariate case. Examples of utility functions
satisfying our assumptions are collected in Section \ref{mSORTEsecexamples}.
Notice that the setup and results in \cite{BDFFM} can be recovered from the
ones in this paper by setting $U(x)=\sum_{j=1}^Nu_j(x^j),x\in{\mathbb{R}}^N$%
, as described in Section \ref{mSORTEseccomparison}. 
\medskip

In this paper we introduce and analyze the following concept. A mSORTE is a triple $(%
\widetilde{Y}_{X},a_{X},{\mathbb{Q}}_{X})$ where:

\begin{itemize}
\item $(\widetilde{Y}_{X},a_{X})$ solve the double optimization problem 
\begin{equation}
\sup \left\{ \sup_{\widetilde{Y}}\mathbb{E}_{\mathbb{P}}\left[ U\left( a+X+%
\widetilde{Y}-{\mathbb{E}_{{{\mathbb{Q}}}_{X}}[\widetilde{Y}]}\right) \right]
\mid a\in {\mathbb{R}}^{N},\,\sum_{j=1}^{N}a^{j}=A\right\} ;
\label{mSORTEeqintrostrong}
\end{equation}

\item {the pricing vector ${\mathbb{Q}}_{X}$ is selected in such a way
that the optimal solution $\widetilde{Y}_{X}$ belongs to the set of feasible
allocations $\mathcal{B}$ and verifies the clearing condition \eqref{mSORTEclearing1}}.
\end{itemize}

We prove existence and uniqueness of a mSORTE.
Quite remarkably, this generalization of a SORTE allows us to introduce and
to study a Nash Equilibrium property for a mSORTE, as shown in Section \ref%
{mSORTESECMAIN} (see Theorem \ref{mSORTEthmmsorteexists}). 
In Section \ref{secexplicitformulas} we provide an example of a class of exponential multivariate utility functions with the explicit computations of the mSORTE.
\bigskip

From a technical perspective, our results can be considered as consequences
of Theorem \ref{mSORTEthmmaingeneral1} and Theorem \ref%
{mSORTEthmmaingeneral2}. The proof of Theorem \ref{mSORTEthmmaingeneral1},
which is the most lengthy and complex, is split in several steps and
collected in Section \ref{mSORTESecProof}. We use a Koml\'{o}s- type
argument, in contrast with the gradient approach in \cite{BDFFM}. This
allows us to obtain existence of optimizers for both the primal and the dual
problems without requiring differentiability of $U(\cdot )$, which is a
rather unusual result in the literature. We also remark that, differently
from \cite{BDFFM}, we need to construct the dual system $(M^{\Phi },K_{\Phi
})$, where $M^{\Phi }$ is a multivariate Orlicz Heart having as topological
dual space the K\"{o}the dual $K_{\Phi }$. Here, we denote with $\Phi :({%
\mathbb{R}}_{+})^{N}\rightarrow {\mathbb{R}}$ the multivariate Orlicz
function $\Phi (x):=U(0)-U(-x)$ associated to the multivariate utility
function $U$. Details of this construction are provided in Section \ref%
{mSORTESecorlicz}.

\bigskip

As already mentioned, this paper is a somehow natural prosecution of \cite%
{BDFFM}. Thus, as far as the conceptual aspects are concerned, we refer to
the literature review in \cite{BDFFM} for extended comments. Here, we limit
ourselves to mentioning that \cite{BDFFM}, and so indirectly this work,
originated from the systemic risk approach developed in Biagini et al.
(2019) \cite{BFFMB} and (2020) \cite{bffm}. For an exhaustive overview on
the literature on systemic risk, see Fouque and Langsam (2013) \cite%
{JP_Langsam} and Hurd (2016) \cite{Hurd}.\newline
Risk sharing equilibria have been studied in Borch (1962) \cite{Borch}, B%
\"{u}hlmann ((1980) \cite{Buhlmann1} and (1984) \cite{Buhlmann}) and B\"{u}%
hlmann and Jewell (1979) \cite{BJ79}. In Barrieu and El Karoui (2005) \cite%
{BE05} inf-convolution of convex risk measures has been introduced as a
fundamental tool for studying risk sharing. Further developments in this
direction have been obtained in Acciaio (2007) \cite{Acciaio}, Filipovi\'{c}
and Svindland (2008) \cite{Filipovic_Svindland}, Jouini et al. (2008) \cite%
{JST07}, Mastrogiacomo and Rosazza Gianin (2015) \cite{MRG14}. Among other
works on risk sharing are also Dana and Le Van (2010) \cite{Dana_LeVan},
Embrechts et al. (2020) \cite{ELW182}, Embrechts et al. (2018) \cite{ELW18},
Filipovi\'{c} and Kupper (2008) \cite{Kupper}, Heath and Ku (2004) \cite%
{Heath_Ku}, Tsanakas (2009) \cite{Tsanakas}, Weber (2018) \cite{Weber17}.
Recent further extensions have been obtained in Liebrich and Svindland
(2019) \cite{LS18}. We refer to Carlier and Dana, (2013) \cite{Carlier1} and
(2012) \cite{Carlier2}, for Risk sharing procedures under multivariate
risks. Regarding multivariate utility functions, which have been widely
exploited in the study of optimal investment under transaction costs, we
cite Campi and Owen (2011) \cite{Campi}, Deelstra et al. (2001) \cite%
{Deelestra}, Kamizono (2004) \cite{Kamizono1}, Bouchard and Pham (2005) \cite%
{Bouchard} and references therein.

\bigskip

The paper is organized as follows. Section \ref{mSORTESecorlicz} is a short
account on multivariate Orlicz spaces and on the relevant properties needed
in the sequel of the paper. The multivariate utility functions used in this
paper are introduced is Section \ref{mSORTEsecsetupassumptions}, together
with our setup and assumptions. The core of the paper is Section \ref%
{mSORTEsecmsorte}, where we formally present the key concepts and provide
our main results. Section \ref{mSORTEsecutmaxanddual} collects some
preliminary results on duality and utility maximization. Most of the proofs
are deferred to Section \ref{mSORTESecProof}. The Appendix collects some
additional technical results and some of the proofs related to Section \ref%
{mSORTESecorlicz}.

\section{Preliminary notations and multivariate Orlicz spaces}

\label{mSORTESecmultiut} Let $(\Omega ,\mathcal{F},{\mathbb{P}}{)}$ be a
probability space and consider the following set of probability vectors on $%
(\Omega ,\mathcal{F})$ 
\begin{equation*}
\mathcal{P}^{N}:=\left\{ {{\mathbb{Q}}=[}{\mathbb{Q}}^{1},...,{\mathbb{Q}}%
^{N}{]}\mid \text{ such that }{\mathbb{Q}}^{j}\ll {\mathbb{P}}\text{ for all 
}j=1,...,N\right\} ,\quad N\in \mathbb{N},\text{ }N\geq 1.
\end{equation*}

For a vector of probability measures ${{\mathbb{Q}}}$ we write ${{\mathbb{Q}}%
}\ll {\mathbb{P}}$ to denote ${\mathbb{Q}}^{1}{\ll {\mathbb{P}}},\dots ,{%
\mathbb{Q}}^{N}\ll {\mathbb{P}}$. Similarly for ${{\mathbb{Q}}}\sim {\mathbb{%
P}}$. For ${\mathbb{Q}}\in \mathcal{P}^{1}$ let 
\begin{equation*}
L^{0}({\mathbb{Q}}{):=}L^{0}(\Omega ,\mathcal{F},{\mathbb{Q}};\mathbb{R}{)}%
\,\,\,\,\,\,\,\,L^{1}({\mathbb{Q}}{):=}L^{1}(\Omega ,\mathcal{F},{\mathbb{Q}}%
;\mathbb{R}{)}\,\,\,\,\,\,\,\,L^{\infty }({\mathbb{Q}}):=L^{\infty }(\Omega ,%
\mathcal{F},{\mathbb{Q}};\mathbb{R}{)}
\end{equation*}%
be the vector spaces of (equivalence classes of) ${\mathbb{Q}}$-a.s. finite, 
${\mathbb{Q}}$-integrable and ${\mathbb{Q}}$-essentially bounded random
variables respectively, and set $L_{+}^{p}({\mathbb{Q}})=\left\{ Z\in L^{p}({%
\mathbb{Q}})\middle|Z\geq 0\,{\mathbb{Q}}-\text{a.s.}\right\} $ and $%
L^{p}(\Omega ,\mathcal{F},{\mathbb{Q}};\mathbb{R}^{N}{)}=(L^{p}({\mathbb{Q}}%
))^{N}$, for $p\in \{0,1,\infty \}$. For ${{\mathbb{Q}}}=[{\mathbb{Q}}%
^{1},\dots ,{\mathbb{Q}}^{N}]\in \mathcal{P}^{N}$ and $p\in \{0,1,\infty \}$
define%
\begin{equation*}
L^{p}({{\mathbb{Q}}):=}L^{p}({\mathbb{Q}}^{1}{)\times ...\times }L^{p}({%
\mathbb{Q}}^{N})\,,\,\,\,\,\,\,\,\,L_{+}^{p}({{\mathbb{Q}}):=}L_{+}^{p}({%
\mathbb{Q}}^{1}{)\times ...\times }L_{+}^{p}({\mathbb{Q}}^{N})\,
\end{equation*}%
%
%
%
%
and write $\mathbb{E}_\probq[Z]=\left[\mathbb{E}_{\probq^1}[Z^1],\dots,\mathbb{E}_{\probq^N}[Z^N]\right]$ for  $Z\in L^1(\probq)$.
Given a vector $y\in {\mathbb{R}}^{N}$ and $n\in \{1,\dots ,N\}$ we will
denote by $y^{[-n]}$ the vector in $\mathbb{R}^{N-1}$ obtained suppressing
the $n$-th component of $y$ for $N\geq 2$ (and $y^{[-n]}=\emptyset$ if $N=1$%
) and we set 
\begin{equation}
\lbrack y^{[-n]};z]:=\left[ y^{1},\dots ,y^{n-1},z,y^{n+1},\dots ,y^{N}%
\right] \,\in \,\mathbb{R}^{N}\text{,\quad for }z\in \mathbb{R}.
\label{mSORTEVectorY}
\end{equation}
We will write ${\mathbb{R}}_{+}:=[0,+\infty)$ and ${\mathbb{R}}%
_{++}:=(0,+\infty)$, $\langle x,y\rangle=\sum_{j=1}^Nx^jy^j$ for the usual
inner product of vectors $x,y\in{\mathbb{R}}^N$. For a vector $x\in {\mathbb{%
R}}^N$, $(x)^\pm$ denote the vectors of positive, negative parts
respectively of the components of $x$. Same applies to $\left|x\right|$.

\subsection{Multivariate Orlicz spaces}

\label{mSORTESecorlicz}

Given a univariate Young function $\phi :{\mathbb{R}}_{+}\rightarrow {%
\mathbb{R}}$ we can associate its conjugate function $\phi ^{\ast
}(y):=\sup_{x\in {\mathbb{R}_{+}}}\left( xy-\phi (x)\right) $ for $y\in {%
\mathbb{R}}_{+}$. As in \cite{RaoRen}, we can associate to both $\phi $ and $%
\phi ^{\ast }$ the Orlicz spaces and Hearts $L^{\phi },M^{\phi },L^{\phi
^{\ast }},M^{\phi ^{\ast }}$. For univariate utility functions, the economic
motivation and the mathematical convenience of using Orlicz spaces theory in
utility maximization problems were shown in \cite{BF08AAP}. We now introduce
multivariate Orlicz functions and spaces induced by multivariate utility
functions. The following definition is a slight modification of the one in
Appendix B of \cite{Drapeau}.

\begin{definition}
\label{mSORTEdeforliczfunct} A function $\Phi :({\mathbb{R}}%
_{+})^{N}\rightarrow {\mathbb{R}}$ is said to be a multivariate Orlicz
function if it is null in $0$, convex, continuous, increasing in the usual
partial order and satisfies: there exist $A>0,B$ constants such that $\Phi
(x)\geq A\sum_{j=1}^{N}x^{j}-B\,\,\forall x\in ({\mathbb{R}}_{+})^{N}$.
\end{definition}


For a given multivariate Orlicz function $\Phi $ we define, as in \cite%
{Drapeau}, the Orlicz space and the Orlicz Heart respectively: 
\begin{equation*}
L^{\Phi }:=\left\{ X\in \left( L^{0}\left( (\Omega ,\mathcal{F},{\mathbb{P}}%
);[-\infty ,+\infty ]\right) \right) ^{N}\mid \,\exists \,\lambda \in
(0,+\infty ),\mathbb{E}_{\mathbb{P}}\left[ \Phi (\lambda \left\vert
X\right\vert )\right] <+\infty \right\} ,
\end{equation*}%
\begin{equation}
M^{\Phi }:=\left\{ X\in \left( L^{0}\left( (\Omega ,\mathcal{F},{\mathbb{P}}%
);[-\infty ,+\infty ]\right) \right) ^{N}\mid \,\forall \,\lambda \in
(0,+\infty ),\mathbb{E}_{\mathbb{P}}\left[ \Phi (\lambda \left\vert
X\right\vert )\right] <+\infty \right\} ,  \label{MPhi}
\end{equation}%
where $\left\vert X\right\vert :=\left[ \left\vert X^{j}\right\vert \right]
_{j=1}^{N}$ is the componentwise absolute value, and $L^{0}\left( (\Omega ,%
\mathcal{F},{\mathbb{P}});[-\infty ,+\infty ]\right) $ is the set of
equivalence classes of $[-\infty ,+\infty ]$-valued $\mathcal{F}$-measurable
functions. We introduce the Luxemburg norm as the functional 
\begin{equation*}
\left\Vert X\right\Vert _{\Phi }:=\inf \left\{ \lambda >0\mid \mathbb{E}_{%
\mathbb{P}}\left[ \Phi \left( \frac{1}{\lambda }\left\vert X\right\vert
\right) \right] \leq 1\right\}
\end{equation*}%
defined on $\left( L^{0}\left( (\Omega ,\mathcal{F},{\mathbb{P}});[-\infty
,+\infty ]\right) \right) ^{N}$ and taking values in $[0,+\infty ]$.

\begin{lemma}
\label{mSORTElemmasummary} Let $\Phi$ be a multivariate Orlicz function.
Then:

\begin{enumerate}
\item {the Luxemburg norm is finite on $X$ if and only if $X\in L^\Phi$; }

\item {the Luxemburg norm is in fact a norm on $L^\Phi$, which makes it a
Banach space; }

\item {$M^\Phi$ is a vector subspace of $L^\Phi$, closed under Luxemburg
norm, and is a Banach space itself if endowed with the Luxemburg norm; }

\item {$L^\Phi$ is continuously embedded in $(L^1({\mathbb{P}}))^N$; }

\item {convergence in Luxemburg norm implies convergence in probability; }

\item {$X\in L^\Phi$, $\left|Y^j\right|\leq \left|X^j\right|\,\forall
j=1,\dots, N$ implies $Y\in L^\Phi$, and the same holds for the Orlicz
Heart. In particular $X\in L^\Phi$ implies $X^\pm\in L^\Phi$ and the same
holds for the Orlicz Heart;}

\item {the topology of $\left\Vert \cdot \right\Vert _{\Phi }$ on $M^{\Phi }$
is order continuous (see \cite{Edgar} Definition 2.1.13 for the definition)
with respect to the componentwise ${\mathbb{P}}$-a.s. order and $M^{\Phi }$
is the closure of $(L^{\infty }({\mathbb{P}}))^{N}$ in Luxemburg norm; }

\item {$M^{\Phi }$ and $L^{\Phi }$ are Banach lattices if endowed with the
topology induced by $\left\Vert \cdot \right\Vert _{\Phi }$ and with the
componentwise ${\mathbb{P}}$-a.s. order.}
\end{enumerate}
\end{lemma}

\begin{proof}
Claims (1)-(5) follow as in \cite{Drapeau}. (6) is trivial from the
definitions. As to (7), sequential order continuity is an application of
Dominated Convergence Theorem, and order continuity follows from Theorem
1.1.3 in \cite{Edgar}. (8) is evident from the previous items.
\end{proof}

Now we need to work a bit on duality.

\begin{definition}
The K\"{o}the dual $K_{\Phi }$ of the space $L^{\Phi }$ is defined as 
\begin{equation}
K_{\Phi }:=\left\{ Z\in \left(L^{0}\left( (\Omega ,\mathcal{F},{\mathbb{P}}%
);[-\infty ,+\infty ]\right)\right)^{N} \mid \sum_{j=1}^{N}X^{j}Z^{j}\in
L^{1}({\mathbb{P}}),\,\forall \,X\in L^{\Phi }\right\} \,.  \notag
\end{equation}
\end{definition}

\begin{proposition}
\label{mSORTEpropdual1} $K_{\Phi }$ can be identified with a subspace of the
topological dual of $L^{\Phi }$ and is a subset of $(L^{1}({\mathbb{P}}%
))^{N} $.
\end{proposition}

\begin{proof}
See Section \ref{mSORTEsecmultiorliczproofs}.
\end{proof}

By Proposition \ref{mSORTEpropdual1} $K_{\Phi }$ is a normed space which can
be naturally endowed with the dual norm of continuous linear functionals,
which we will denote by 
\begin{equation*}
\left\Vert Z\right\Vert _{\Phi }^{\ast }:=\sup \left\{\mathbb{E}_\mathbb{P} %
\left[ |\sum_{j=1}^{N}X^{j}Z^{j}|\right]\mid \left\Vert X\right\Vert _{\Phi
}\leq 1\right\} \,.
\end{equation*}
This norm will play here the role of the Orlicz norm, and the relation
between the two norms $\left\| \cdot\right\|_\Phi$ and $\left\|
\cdot\right\|_\Phi^*$ is well understood in the univariate case (see Theorem
2.2.9 in \cite{Edgar}). The following Proposition summarizes useful
properties which show how the K\"{o}the dual can play the role of the Orlicz
space $L^{\Phi ^{\ast }}$ for $M^{\Phi }$ in univariate theory, and are the
counterparts to Corollary 2.2.10 in \cite{Edgar}.

\begin{proposition}
\label{mSORTEthmsummarykoethe} The following hold:

\begin{enumerate}
\item {\label{mSORTElemmaluxenorm} $K_{\Phi }=\left\{ Z\in \left(L^{0}\left(
(\Omega ,\mathcal{F},{\mathbb{P}});[-\infty ,+\infty ]\right)\right)^{N}
\mid \sum_{j=1}^{N}X^{j}Z^{j}\in L^{1}({\mathbb{P}}),\,\forall \,X\in
M^{\Phi }\right\} \,;$}

\item {\label{mSORTEthmdualhearth} the topological dual of $(M^{\Phi
},\left\Vert \cdot \right\Vert _{\Phi })$ is $(K_{\Phi },\left\Vert \cdot
\right\Vert _{\Phi }^{\ast })$;}

\item {\label{mSORTElemmakoetheok} suppose $L^{\Phi }=L^{\Phi _{1}}\times
\dots \times L^{\Phi _{N}}\,.$ Then we have that $K_{\Phi }=L^{\Phi
_{1}^{\ast }}\times \dots \times L^{\Phi _{N}^{\ast }}$ }.
\end{enumerate}
\end{proposition}

\begin{proof}
See Section \ref{mSORTEsecmultiorliczproofs}.
\end{proof}

We now provide an example connecting the multivariate theory to the
univariate classical one.

\begin{remark}
\label{mSORTERemMPHI}Even thought we will not make this assumption in the
rest of the paper, suppose in this Remark that $\Phi (x)=\sum_{j=1}^{N}\Phi
_{j}(x^{j})$ for univariate Orlicz functions, that is each separately
satisfying Definition \ref{mSORTEdeforliczfunct} for $N=1$. Then we could
consider the multivariate spaces $L^{\Phi }$ and $M^{\Phi }$ as above or we
could take $L^{\Phi _{1}}\times \dots \times L^{\Phi _{N}}$ and $M^{\Phi
_{1}}\times \dots \times M^{\Phi _{N}}$.

As shown in Section \ref{mSORTEsecmultiorliczproofs}, the following identity
between sets holds: 
\begin{equation*}
M^{\Phi }=M^{\Phi _{1}}\times \dots \times M^{\Phi _{N}}\quad \text{and}%
\quad L^{\Phi }=L^{\Phi _{1}}\times \dots \times L^{\Phi _{N}}
\end{equation*}%
and furthermore 
\begin{equation}
\frac{1}{N}\sum_{j=1}^{N}\left\Vert X^{j}\right\Vert _{\Phi _{j}}\leq
\left\Vert X\right\Vert _{\Phi }\leq N\sum_{j=1}^{N}\left\Vert
X^{j}\right\Vert _{\Phi _{j}}\,.  \label{mSORTEnormequiv}
\end{equation}%
Observe that in the setup of this Remark, from Proposition \ref%
{mSORTEthmsummarykoethe} Item 3, we have 
\begin{equation*}
K_{\Phi }=L^{\Phi _{1}^{\ast }}\times \dots \times L^{\Phi _{N}^{\ast }}\,.
\end{equation*}
\end{remark}

\section{Setup and assumptions}

\label{mSORTEsecsetupassumptions}

\begin{definition}
\label{mSORTEdefmultiutil} We say that $U:{\mathbb{R}}^{N}\rightarrow {%
\mathbb{R}}$ is a \textbf{multivariate utility function} if it is strictly
concave and strictly increasing with respect to the partial componentwise
order. When $N=1$ we will use the term univariate utility function instead.
For a multivariate utility function $U$ we define the convex conjugate in
the usual way by 
\begin{equation}
V(y):=\sup_{x\in {\mathbb{R}}^{N}}\left( U(x)-\left< x,y\right> \right) . 
\notag
\end{equation}
\end{definition}

Observe that by definition $U(x)\leq \left\langle x,y\right\rangle +V(y)$
for every $x,y\in {\mathbb{R}}^{N}$, and $V(\cdot )\geq U(0)$ that is $V$ is
lower bounded.

\begin{definition}
For a multivariate utility function $U$, we define the function $\Phi $ on $(%
{\mathbb{R}}_{+})^{N}$ by 
\begin{equation}
\Phi (x):=U(0)-U(-x)\,.  \label{mSORTEassocorlicz}
\end{equation}
\end{definition}

\begin{remark}
\label{mSORTEreminada} The well known Inada conditions, for
(one-dimensional) concave increasing utility functions $u:{\mathbb{R}}%
\rightarrow {\mathbb{R}}$, have an evident economic significance and are
very often assumed to hold true in order to solve utility maximization
problems. 
\begin{equation*}
\text{Inada}(+\infty )\text{:\quad }\lim_{x\uparrow +\infty }\frac{u(x)}{x}%
=0\,;\,\,\,\,\text{Inada}(-\infty )\text{:\quad }\lim_{x\downarrow -\infty }%
\frac{u(x)}{x}=+\infty \,.
\end{equation*}%
As it is easy to check, they can be equivalently rewritten as:%
\begin{eqnarray}
&&\text{Inada}(+\infty \text{)}\text{:\quad }\forall \varepsilon >0\text{ \ }%
\exists k_{\varepsilon }\in {\mathbb{R}}:u(x)\leq \varepsilon
x+k_{\varepsilon }\text{ }\forall x\geq 0,\text{ }x\in \mathbb{R},
\label{mSORTEInada++} \\
&&\text{Inada}(-\infty \text{)}\text{:\quad }\forall A>0\text{ }\exists
k_{A}\in {\mathbb{R}}:u(x)\leq Ax+k_{A}\text{ }\forall x\leq 0\text{, }x\in 
\mathbb{R}.  \notag
\end{eqnarray}%
Consider now the condition weaker than Inada$(-\infty )$:%
\begin{equation}
\exists A>0\text{ }\exists k\in {\mathbb{R}}:u(x)\leq Ax+k\text{ }\forall
x\leq 0\text{, }x\in \mathbb{R}.  \label{mSORTEweakInada}
\end{equation}%
One can again easily check that the two conditions \eqref{mSORTEInada++} and %
\eqref{mSORTEweakInada} are equivalent to the following single statement:
there exists an Orlicz function $\widehat{\Phi }:{\mathbb{R}}_{+}\rightarrow 
{\mathbb{R}}$ and a function $f:{\mathbb{R}}_{+}\rightarrow {\mathbb{R}}$
such that 
\begin{equation}
u(x)\leq -\widehat{\Phi }((x)^{-})+\varepsilon |x|+f(\varepsilon )\,\,\,\,%
\text{ for every }\varepsilon >0.  \label{mSORTEOneInada}
\end{equation}
\end{remark}

Motivated by the above remark, we now introduce a condition that will
replace \eqref{mSORTEOneInada} in the multivariate case and will play the
same role as the Inada in the one dimensional case.

\begin{definition}
\label{mSORTEwellcontrolled} We say that a multivariate utility function $U:{%
\mathbb{R}}^{N}\rightarrow {\mathbb{R}}$ is \textbf{well controlled} if
there exist a multivariate Orlicz function $\widehat{\Phi }:{\mathbb{R}}%
_{+}^{N}\rightarrow {\mathbb{R}}$ and a function $f:{\mathbb{R}}%
_{+}\rightarrow {\mathbb{R}}$ such that 
\begin{equation}
U(x)\leq -\widehat{\Phi }((x)^{-})+\varepsilon \sum_{j=1}^{N}\left\vert
x^{j}\right\vert +f(\varepsilon )\,\,\,\,\text{ for every }\varepsilon >0.
\label{mSORTEcontrolwithphihat}
\end{equation}
\end{definition}

\begin{lemma}
\label{mSORTElemmaconswellcontrol} Suppose that the multivariate utility
function $U$ is well controlled. Then:

\begin{itemize}
\item[(i)] the function $\Phi (x)=U(0)-U(-x),$ $x\in {\mathbb{R}}_{+}^{N},$
defines a multivariate Orlicz function;

\item[(ii)] $L^{\Phi }\subseteq L^{\widehat{\Phi }}$;

\item[(iii)] for every $\varepsilon >0$ small enough there exist a constant $%
b_{\varepsilon }$ such that 
\begin{equation}
U(x)\leq \varepsilon \sum_{j=1}^{N}(x^{j})^{+}+b_{\varepsilon
}\,\,\,\,\,\forall x\in {\mathbb{R}}^{N};  \label{mSORTEcontrolwithepsilon}
\end{equation}

\item[(iv)] there exist $a>0$, $b\in {\mathbb{R}}$ such that 
\begin{equation}
U(x)\leq
a\sum_{j=1}^{N}(x^{j})^{+}-2a\sum_{j=1}^{N}(x^{j})^{-}+b\,\,\,\,\,\forall
\,x\in {\mathbb{R}}^{N}.  \label{mSORTElemmacontrolwithline}
\end{equation}
\end{itemize}
\end{lemma}

\begin{proof}
Recall that by Definition \ref{mSORTEdeforliczfunct} there exist $A>0,B\in {%
\mathbb{R}}$ such that $\widehat{\Phi }((x)^{-})\geq
A\sum_{j=1}^{N}(x^{j})^{-}-B$. Therefore, if $U$ is well controlled then
\begin{equation}
U(x)\leq -A\sum_{j=1}^{N}(x^{j})^{-}+B+\varepsilon \sum_{j=1}^{N}\left\vert
x^{j}\right\vert +f(\varepsilon )\,\,\,\,\text{ for every }\varepsilon >0.
\label{mSORTEUAb1}
\end{equation}%
(i): It is enough to show the existence of $A_{\Phi }>0$, $B_{\Phi }\in {%
\mathbb{R}}$ such that $\Phi (x)\geq A_{\Phi }\sum_{j=1}^{N}x^{j}-B_{\Phi
}\,\forall x\in {\mathbb{R}}_{+}^{N}$. For any $x\in {\mathbb{R}}_{+}^{N}$,
using \eqref{mSORTEUAb1} with $\varepsilon =\frac{A}{2},$ we obtain%
\begin{equation*}
U(-x)\leq -A\sum_{j=1}^{N}x^{j}+B+\frac{A}{2}\sum_{j=1}^{N}x^{j}+f\left( 
\frac{A}{2}\right) =-\frac{A}{2}\sum_{j=1}^{N}x^{j}+\left( B+f\left( \frac{A%
}{2}\right) \right) ,
\end{equation*}
and the desired inequality follows letting $A_{\Phi }:=\frac{A}{2}$ and $%
B_{\Phi }:=-\left(U(0)-f\left( \frac{A}{2}\right) -B\right).$

(ii): From \eqref{mSORTEcontrolwithphihat} one can easily see that $%
\Phi(x)\geq \widehat{\Phi}(x)-\varepsilon\sum_{j=1}^Nx^j+\text{constant}$
for all $x\in({\mathbb{R}}_+)^N$. The claim then follows recalling that $%
X\in L^\Phi\Rightarrow X\in(L^1({\mathbb{P}}))^N$ by Lemma \ref%
{mSORTElemmasummary} Item 4.

(iii): Follows from \eqref{mSORTEUAb1} if $\varepsilon <A$ and $%
b_{\varepsilon }:=B+f(\varepsilon ).$

(iv): Follows with $a=\frac{A}{3},$ $b=B+f(\frac{A}{3})$ choosing $%
\varepsilon =\frac{A}{3}$ in \eqref{mSORTEUAb1}.
\end{proof}

\begin{remark}
In the proofs in multivariate setting, the inequalities %
\eqref{mSORTEcontrolwithepsilon} and \eqref{mSORTElemmacontrolwithline} will
play the same role that respectively \eqref{mSORTEInada++} and %
\eqref{mSORTEweakInada} have in the unidimensional case. In Proposition \ref%
{mSORTEPropBLambda} we use the aforementioned univariate Inada conditions to
make sure that \eqref{mSORTEcontrolwithphihat} holds when $U$ has a
particular form.
\end{remark}

We observe that the inclusion $L^{\widehat{\Phi }}\subseteq L^{\Phi }$,
opposite of (ii), is a simple integrability requirement, which can be
rephrased as: if for $X\in (L^{0}\left( (\Omega ,\mathcal{F},{\mathbb{P}}%
);[-\infty ,+\infty ]\right) )^{N}$ there exist $\lambda >0$ such that $%
\mathbb{E}_{\mathbb{P}}\left[ \widehat{\Phi }(-\lambda \left\vert
X\right\vert) \right] >-\infty $, then there exists $\alpha >0$ such that $%
\mathbb{E}_{\mathbb{P}}\left[ U(-\alpha \left\vert X\right\vert )\right]
>-\infty $. This request is rather weak and there are many examples of
choices of $U$ that guarantee this condition is met (see Section \ref%
{mSORTEsecexamples}). Without further mention, the following two standing assumptions hold true throughout the
paper.

\begin{standingassumptionI*}
The function $U:{\mathbb{R}}^{N}\rightarrow {\mathbb{R}}$ is a multivariate
utility function which is well controlled (Definition \ref%
{mSORTEwellcontrolled}) and such that for $\widehat{\Phi}$ in %
\eqref{mSORTEcontrolwithphihat} 
\begin{equation}
L^{\widehat{\Phi }}=L^{\Phi }\,.  \label{mSORTEcontrolwithphihat11}
\end{equation}
\end{standingassumptionI*}

\begin{standingassumptionII*}
$\mathcal{B}\subseteq \mathcal{C}_{{\mathbb{R}}}$ is a convex cone, closed
in probability, $0\in \mathcal{B}$, ${\mathbb{R}}^{N}+\mathcal{B}=\mathcal{B}
$. The vector $X$ belongs to the Orlicz Heart $M^{\Phi }.$ 
\end{standingassumptionII*}

Observe that the Standing Assumption II implies that all constant vectors
belong to $\mathcal{B}$, so that all (deterministic) vector in the form $%
e^{i}-e^{j}$ (differences of elements in the canonical base of ${\mathbb{R}}%
^{N}$) belong to $\mathcal{B}\cap M^{\Phi }$. We recall the following
concept, introduced in \cite{bffm} Definition 5.15, that was already used in 
\cite{BDFFM}.

\begin{definition}
\label{mSORTEdefclosedundertrunc} $\mathcal{B}$ is closed under truncation
if for each $Y\in \mathcal{B}$ there exists $m_{Y}\in \mathbb{N}$ and $%
c_{Y}\in {\mathbb{R}}^{N}$ such that $\sum_{j=1}^{N}Y^{j}=%
\sum_{j=1}^{N}c_{Y}^{j}$ and for all $m\geq m_{Y}$ 
\begin{equation*}
Y_{m}:=Y1_{\{\left\vert Y^{j}\right\vert <m\,\forall j=1,\dots
,N\}}+c_{Y}1_{\Omega \setminus \{\left\vert Y^{j}\right\vert <m\,\forall
j=1,\dots ,N\}}\in \mathcal{B}\,.
\end{equation*}
\end{definition}

\begin{assumption}
\label{mSORTEA1} $\mathcal{B}$ is closed under truncation.
\end{assumption}

As pointed out in \cite{bffm}, $\mathcal{B}=\mathcal{C}_{\mathbb{R}}$ is
closed under truncation. Closedness under truncation property holds true for
a rather wide class of constraints. For a more detailed explanation and
examples, see also \cite{BDFFM} Example 3.17 and Example 4.20. We will also
need the following additional notation.

\begin{enumerate}
\item {For any $A\in {\mathbb{R}}$ we consider the set of feasible random
allocations 
\begin{equation}
\mathcal{B}_{A}:=\mathcal{B}\cap \left\{ Y\in (L^{0}({\mathbb{P}}))^{N}\mid
\sum_{j=1}^{N}Y^{j}\leq A\right\} \subseteq \mathcal{C}_{\mathbb{R}}\,. 
\notag
\end{equation}%
}

\item {$\mathcal{Q}$ is the set of vectors of probability measures ${\mathbb{%
Q}}=[{\mathbb{Q}}^{1},\dots ,{\mathbb{Q}}^{N}]$, with ${\mathbb{Q}}^{j}\ll {%
\mathbb{P}}\,$\ $\forall \,j=1,\dots ,N$, defined by 
\begin{equation}
\mathcal{Q}:=\left\{ {\mathbb{Q}}\mid \left[ \frac{\mathrm{d}{\mathbb{Q}}^{1}%
}{\mathrm{d}{\mathbb{P}}},\dots ,\frac{\mathrm{d}{\mathbb{Q}}^{N}}{\mathrm{d}{%
\mathbb{P}}}\right] \in K_{\Phi },\,\sum_{j=1}^{N}\mathbb{E}_{\mathbb{Q}^{j}}%
\left[ Y^{j}\right] \leq 0\,\text{\ }\forall \,Y\in \mathcal{B}_{0}\cap
M^{\Phi }\right\} \,.  \label{mSORTEdefsetQ}
\end{equation}%
Identifying Radon-Nikodym derivatives and measures in the natural way, this
can be rephrased as: $\mathcal{Q}$ is the set of normalized (i.e. with
componentwise expectations equal to $1$), non negative vectors in the polar
of $\mathcal{B}_{0}\cap M^{\Phi }$, in the dual system $(M^{\Phi },K_{\Phi
}) $. }Observe that $M^{\Phi }\subseteq L^{1}({\mathbb{Q}})$ for all ${%
\mathbb{Q\in }}\mathcal{Q}$ and that $\mathcal{Q}$ depends on the set {$%
\mathcal{B}$.}

\item In the definition of  mSORTE we will adopt the subset 
$\mathcal{Q}_{\mathcal{B},V}\subseteq \mathcal{Q}$ of vectors of probability
measures having \textquotedblleft finite entropy\textquotedblright {\ 
\begin{equation}
\mathcal{Q}_{\mathcal{B},V}:=\left\{ {\mathbb{Q}}\in \mathcal{Q}\mid \mathbb{%
E}_{\mathbb{P}}\left[ V\left( \lambda \frac{\mathrm{d}{\mathbb{Q}}}{\mathrm{d%
}{\mathbb{P}}}\right) \right] <+\infty \text{ for some }\lambda >0\right\}
\label{mSORTEQV}
\end{equation}%
and the set }$\mathcal{L}${\ of the random allocations satisfying the
integrability requirements defined by}%
\begin{equation}
\mathcal{L}:=\bigcap_{{\mathbb{Q}}\in \mathcal{Q}_{\mathcal{B},V}}L^{1}({%
\mathbb{Q}})=\mathcal{L}^{1}\times \dots \times \mathcal{L}^{N}\,.
\label{mSORTEdefL}
\end{equation}%
{Here }$\mathcal{L}^{j}:=\left\{ Y^{j}\in L^{1}({\mathbb{Q}}^{j})\,\forall {%
\mathbb{Q=[{\mathbb{Q}}}}^{1},...,{\mathbb{{\mathbb{Q}}}}^{N}{\mathbb{]\in }}%
\mathcal{Q}_{\mathcal{B},V}\right\} .$
\end{enumerate}

\section{Multivariate Systemic Optimal Risk Transfer Equilibrium}

\label{mSORTEsecmsorte}

\subsection{Main concept}

\label{mSORTEsecmainconcepts}

We now provide the formal definition of the concept already illustrated in
the Introduction (see equation \eqref{mSORTEeqintrostrong}). It is the natural generalization of SORTE as introduced
in \cite{BDFFM} Definition 3.7.

%
%
%

\begin{definition}
\label{mSORTEstrongmsorte} The triple $(\widetilde{Y}_{X},{\mathbb{Q}}_{X},a_X)\in 
\mathcal{L}\mathbf{\times }\mathcal{Q}_{\mathcal{B},V}\mathbf{\times }\mathbb{R}^{N}$ is a 
\textbf{Multivariate Systemic Optimal Risk Transfer Equilibrium} (mSORTE)
with budget $A\in{\mathbb{R}}$ if

\begin{enumerate}
\item $(\widetilde{Y}_{X},a_{X})$ is an optimum for 
\begin{equation*}
\sup\left\{
\sup \left\{ \mathbb{E}_{\mathbb{P}}\left[ U(a+X+\widetilde{Y}-\mathbb{E}_{\probq_X}[\widetilde{Y}])\right] \mid \widetilde{Y}\in \mathcal{L
}\right\} \mid a\in {\mathbb{R}}^{N}   \sum_{j=1}^{N}a_{j}=A\right\}\,;
\end{equation*}

\item $\widetilde{Y}_{X}\in \mathcal{B}$ and $\sum_{j=1}^{N}\widetilde{Y}_{X}^{j}=0\,\, {\mathbb{P}}$%
-a.s..
\end{enumerate}
\end{definition}

\subsection{Main results}

\label{mSORTESECMAIN}We provide sufficient conditions for existence,
uniqueness and the Nash Equilibrium property of a mSORTE, see Theorems \ref%
{mSORTEthmmsorteexists} and \ref{mSORTEthmstrongunique}. Such results are
relatively simple consequences of the following key duality Theorem \ref%
{mSORTEthmmaingeneral1}, whose proof in Section \ref{mSORTESecProof} will
involve several steps. As demonstrated in Section \ref{mSORTEreplacing},
Theorems \ref{mSORTEthmmaingeneral1}, \ref{mSORTEthmmsorteexists} and \ref%
{mSORTEthmadditionalreq} below also hold true if Assumption \ref{mSORTEA1}
is replaced with an assumption on the utility function $U.$

\begin{theorem}
\label{mSORTEthmmaingeneral1} Under Assumption \ref{mSORTEA1} the following
holds: 
\begin{equation}
\sup_{Y\in \mathcal{B}_{A}\cap M^{\Phi }}\mathbb{E}_{\mathbb{P}}\left[ U(X+Y)%
\right] =\sup \left\{ \mathbb{E}_{\mathbb{P}}\left[ U(X+Y)\right] \mid Y\in 
\mathcal{L},\,\sum_{j=1}^{N}\mathbb{E}_{\mathbb{Q}^{j}}\left[ Y^{j}\right]
\leq A\,\forall {\mathbb{Q}}\in \mathcal{Q}_{\mathcal{B},V}\right\}
\label{mSORTEeqonmphiequalLv}
\end{equation}%
\begin{equation}
=\min_{{\mathbb{Q}}\in \mathcal{Q}_{\mathcal{B},V}}\min_{\lambda \geq 0}\left( \lambda
\left( \sum_{j=1}^{N}\mathbb{E}_{\mathbb{Q}^{j}}\left[ X^{j}\right]
+A\right) +\mathbb{E}_{\mathbb{P}}\left[ V\left( \lambda \frac{\mathrm{d}{%
\mathbb{Q}}}{\mathrm{d}{\mathbb{P}}}\right) \right] \right)= \sup_{Y\in 
\mathcal{B}_{A}\cap (L^{\infty }({\mathbb{P}}))^{N}}\mathbb{E}_{\mathbb{P}}%
\left[ U(X+Y)\right] \,.  \label{mSORTEeqminimaxgeneralintro1}
\end{equation}%
Moreover

\begin{enumerate}
\item there exists an optimum $\widehat{Y}\in \mathcal{L}$ to the problem in
RHS of \eqref{mSORTEeqonmphiequalLv}. Such an optimum satisfies $\widehat{Y}%
\in \mathcal{B}_{A}\cap \mathcal{L}$ and for any optimum $(\widehat{\lambda }%
,\widehat{{\mathbb{Q}}})$ of \eqref{mSORTEeqminimaxgeneralintro1} the
following holds: 
\begin{equation}
\label{optsumstoA}
\sum_{j=1}^{N}\mathbb{E}_{\widehat{{\mathbb{Q}}}^{j}}[\widehat{Y}%
^{j}]=A=\sum_{j=1}^{N}\widehat{Y}^{j},\quad {\mathbb{P}}-{a.s.};
\end{equation}

\item any optimum $(\widehat{\lambda },\widehat{{\mathbb{Q}}})$ of %
\eqref{mSORTEeqminimaxgeneralintro1} satisfies $\widehat{\lambda }>0$ and $%
\widehat{{\mathbb{Q}}}\sim {\mathbb{P}}$; 

\item there exists a unique optimum to RHS of \eqref{mSORTEeqonmphiequalLv}.
If $U$ is additionally differentiable, there exists a unique optimum $(%
\widehat{\lambda },\widehat{{\mathbb{Q}}})$ of %
\eqref{mSORTEeqminimaxgeneralintro1}.

\end{enumerate}
\end{theorem}

\begin{proof}
The case $A=0$ is covered in Theorem \ref{mSORTEthmadditionalreq} and
Corollary \ref{mSORTEcorsuponlinftyeqmphi}, the latter only proving the last
equality in \eqref{mSORTEeqminimaxgeneralintro1}. In Section \ref%
{mSORTESecfrom0toA} we then explain how we can apply also to $A\neq 0$ the
same arguments used for $A=0$.
\end{proof}

The following result is the counterpart to Theorem \ref%
{mSORTEthmmaingeneral1}, once a vector ${\mathbb{Q}}\in \mathcal{Q}_{\mathcal{B},V}$ is
fixed, and will be applied in Theorem \ref{mSORTEthmstrongunique}.

\begin{theorem}
\label{mSORTEthmmaingeneral2}Under Assumption \ref{mSORTEA1}, for every ${%
\mathbb{Q}}\in \mathcal{Q}_{\mathcal{B},V}$ and $A\in {\mathbb{R}}$ the following holds: 
\begin{equation}
\sup_{\substack{ Y\in \mathcal{L}  \\ \sum_{j=1}^{N}\mathbb{E}_{\mathbb{Q}%
^{j}}\left[ Y^{j}\right] \leq A}}\mathbb{E}_{\mathbb{P}}\left[ U(X+Y)\right]
=\min_{\substack{ \lambda \geq 0}}\left( \lambda \left( \sum_{j=1}^{N}%
\mathbb{E}_{\mathbb{Q}^{j}}\left[ X^{j}\right] +A\right) +\mathbb{E}_{%
\mathbb{P}}\left[ V\left( \lambda \frac{\mathrm{d}{\mathbb{Q}}}{\mathrm{d}{%
\mathbb{P}}}\right) \right] \right) .  \label{mSORTEeqminimaxgeneralfixedq}
\end{equation}
\end{theorem}

\begin{proof}
Consider first $A=0$. By Equations \eqref{mSORTEeqqfixedsup1} and %
\eqref{mSORTEeqqfixedsup12}

\begin{equation*}
\min_{\substack{ \lambda \geq 0}}\left( \lambda \left( \sum_{j=1}^{N}\mathbb{%
E}_{\mathbb{Q}^{j}}\left[ X^{j}\right] \right) +\mathbb{E}_{\mathbb{P}}\left[
V\left( \lambda \frac{\mathrm{d}{\mathbb{Q}}}{\mathrm{d}{\mathbb{P}}}\right) %
\right] \right) =\sup_{\substack{ Y\in M^{\Phi }  \\ \sum_{j=1}^{N}\mathbb{E}%
_{\mathbb{Q}^{j}}\left[ Y^{j}\right] \leq 0}}\mathbb{E}_{\mathbb{P}}\left[
U(X+Y)\right] .
\end{equation*}%
Observing that $M^{\Phi }\subseteq \mathcal{L}\subseteq L^{1}({\mathbb{Q}})$%
, we have 
\begin{align*}
&\sup_{\substack{ Y\in M^{\Phi } \sum_{j=1}^{N}\mathbb{E}_{\mathbb{Q}^{j}}%
\left[ Y^{j}\right] \leq 0}}\mathbb{E}_{\mathbb{P}}\left[ U(X+Y)\right]
\leq\sup_{\substack{ Y\in \mathcal{L}\cap L^{1}({\mathbb{Q}})  \\ %
\sum_{j=1}^{N}\mathbb{E}_{\mathbb{Q}^{j}}\left[ Y^{j}\right] \leq 0}}\mathbb{%
E}_{\mathbb{P}}\left[ U(X+Y)\right] \\
&\leq \sup_{\substack{ Y\in L^{1}({\mathbb{Q}})  \\ \sum_{j=1}^{N}\mathbb{E}%
_{\mathbb{Q}^{j}}\left[ Y^{j}\right] \leq 0}}\mathbb{E}_{\mathbb{P}}\left[
U(X+Y)\right] \leq \inf_{\lambda \geq 0}\left( \lambda \left( \sum_{j=1}^{N}%
\mathbb{E}_{\mathbb{Q}^{j}}\left[ X^{j}\right] \right) +\mathbb{E}_{\mathbb{P%
}}\left[ V\left( \lambda \frac{\mathrm{d}{\mathbb{Q}}}{\mathrm{d}{\mathbb{P}}%
}\right) \right] \right) ,
\end{align*}%
by Remark \ref{mSORTER1} below. The case $A=0$ is then proved. The case $%
A\neq 0$, instead, follows from Section \ref{mSORTESecfrom0toA}.
\end{proof}

\begin{remark}
\label{mSORTER1}From the definition of $V$ we obtain the Fenchel inequality%
\begin{equation*}
U(X+Y)\leq \langle X+Y,\lambda Z\rangle +V(\lambda Z)\text{\quad }\mathbb{P}%
\text{-a.s. for all }X,Y,Z\in (L^{0}(\mathbb{P}))^{N}\text{, }\lambda \geq 0.
\end{equation*}%
Recall that $M^{\Phi }\subseteq L^{1}({\mathbb{Q}})$ for all ${\mathbb{Q\in }%
}\mathcal{Q}.$ For all $X\in M^{\Phi }$, for all ${\mathbb{Q\in }}\mathcal{Q}
$ and $Y$ such that $\sum_{j=1}^{N}\mathbb{E}_{\mathbb{Q}^{j}}\left[ Y^{j}%
\right] \leq A$ we then have: 
\begin{equation}
\label{mSORTEfenchelineq}
\begin{split}
\mathbb{E}_{\mathbb{P}}\left[ U(X+Y)\right] &\leq \inf_{\lambda \geq
0}\left\{ \lambda \sum_{j=1}^{N}\mathbb{E}_{\mathbb{Q}^{j}}\left[
(X^{j}+Y^{j})\right] +\mathbb{E}_{\mathbb{P}}\left[ V\left( \lambda \frac{%
\mathrm{d}{\mathbb{Q}}}{\mathrm{d}{\mathbb{P}}}\right) \right] \right\} \\
&\leq \inf_{\lambda \geq 0}\left\{ \lambda \left( \sum_{j=1}^{N}\mathbb{E}_{%
\mathbb{Q}^{j}}\left[ X^{j}\right] +A\right) +\mathbb{E}_{\mathbb{P}}\left[
V\left( \lambda \frac{\mathrm{d}{\mathbb{Q}}}{\mathrm{d}{\mathbb{P}}}\right) %
\right] \right\}
\end{split}%
\end{equation}
and the last expression is finite if ${\mathbb{Q\in }}\mathcal{Q}_{\mathcal{B},V}.$ 
\end{remark}

\paragraph{On the existence of a mSORTE and Nash Equilibrium}


\begin{theorem}
\label{mSORTEthmmsorteexists} Under Assumption \ref{mSORTEA1} a Multivariate
Systemic Optimal Risk Transfer Equilibrium $(\widetilde{Y}_X,{\mathbb{Q}_X},%
a_X)\in \mathcal{L}\times \mathcal{Q}_{\mathcal{B},V}\times {\mathbb{R}}^{N}$
exists, with $\mathbb{E}_{\probq^j_X}[\widetilde{Y}^j_X]=0$ for every $j=1,\dots,N$. Furthermore,
the following Nash Equilibrium property holds for any mSORTE $(\widetilde{Y}_X,{\mathbb{Q}_X},%
a_X)$: for every $j=1,\dots,N$

\begin{equation}
\label{eqnashequil}
\widetilde{Y}^j_X\in\argmax\left\{\mathbb{E}_{\mathbb{P}}\left[ U\left( a_{X}+X+[\widetilde{Y}%
_{X}^{[-j]};Z]-{\mathbb{E}_{{{\mathbb{Q}}}_{X}}[\widetilde{Y}_{X}^{[-j]};}%
Z]\right) \right]\mid Z\in\mathcal{L}^j\right\}
\end{equation}%
where $\mathcal{L}^j=\bigcap_{\probq\in\mathcal{Q}_{\mathcal{B},V}}L^1(\probq^j)$,  and $$\mathbb{E}_{\mathbb{Q}_{X}}[\widetilde{Y}_{X}^{[-j]};Z]=\left[\mathbb{E}_{\mathbb{Q}^1_{X}}[\widetilde{Y}^1_{X}],\dots,\mathbb{E}_{\mathbb{Q}^{j-1}_{X}}[\widetilde{Y}^{j-1}_{X}], \mathbb{E}_{\mathbb{Q}^{j}_{X}}[Z],\mathbb{E}_{\mathbb{Q}^{j+1}_{X}}[\widetilde{Y}^{j+1}_{X}],\dots,\mathbb{E}_{\mathbb{Q}^{N}_{X}}[\widetilde{Y}^{N}_{X}] \right]\,.$$

\end{theorem}

\begin{proof}
Take $\widehat{Y}$ as in Theorem \ref{mSORTEthmmaingeneral1}
Item 1, $\widehat{{\mathbb{Q}}}$ an optimizer of %
\eqref{mSORTEeqminimaxgeneralintro1}, and set $\widehat{a}^{j}:=\mathbb{E}_{%
\widehat{{\mathbb{Q}}}^{j}}[\widehat{Y}^{j}]$, $j=1,\dots ,N,$.  Then, from %
\eqref{mSORTEeqonmphiequalLv} and \eqref{mSORTEeqminimaxgeneralintro1}, 
\begin{align}
&\sup\left\{\mathbb{E}_\mathbb{P} \left[U(X+Y)\right]\mid Y\in \mathcal{L}%
,\,\sum_{j=1}^N\mathbb{E}_{\mathbb{Q}^j} \left[Y^j\right]\leq A\,\,\forall{%
\mathbb{Q}}\in\mathcal{Q}_{\mathcal{B},V}\right\}  \notag \\
&=\min_{\substack{ \lambda \geq 0}}\left( \lambda \left( \sum_{j=1}^{N}%
\mathbb{E}_{\widehat{\mathbb{Q}}^{j}}\left[ X^{j}\right] +A\right) +\mathbb{E%
}_{\mathbb{P}}\left[ V\left( \lambda \frac{\mathrm{d}\widehat{{\mathbb{Q}}}}{%
\mathrm{d}{\mathbb{P}}}\right) \right] \right)  \notag \\
&\overset{\eqref{mSORTEeqminimaxgeneralfixedq}}{=}\sup\left\{\mathbb{E}_{%
\mathbb{P}}\left[ U(X+Y)\right]\mid Y\in \mathcal{L},\, \sum_{j=1}^{N}%
\mathbb{E}_{\widehat{{\mathbb{Q}}}^{j}}[Y^{j}]\leq A\right\}
\label{mSORTEeq333} \\
&=\sup\left\{
\sup \left\{ \mathbb{E}_{\mathbb{P}}\left[ U(X+Y)\right] \mid Y\in \mathcal{L%
},\,\mathbb{E}_{\widehat{\mathbb{Q}}^{j}}\left[ Y^{j}\right] \leq
a^{j},\,\forall \,j\mid  a\in {\mathbb{R}}^{N},\, \sum_{j=1}^{N}a_{j}=A\right\} \right\}  \label{mSORTEeq444} 
\\&=:S^{\widehat{\probq}}(A)\notag
\end{align}%
where \eqref{mSORTEeq444} is a simple reformulation of \eqref{mSORTEeq333}.
By Item 1 of Theorem \ref{mSORTEthmmaingeneral1}, the optimizer $\widehat{Y}%
\in \mathcal{L}$ satisfies the constraints of the problem in %
\eqref{mSORTEeq333}, and by \eqref{optsumstoA} $(\widehat{Y},\widehat{a})$  yields  an optimum for the problem in %
\eqref{mSORTEeq444}. 
By Lemma \ref{mSORTESequalsH}, setting $a_X:=\widehat{a}\in\R^N$, $\widetilde{Y}_X=\widehat{Y}-a_X\in\mathcal{L}$, $\probq_X=\widehat{\probq}$ we have that $(\widetilde{Y}_X,\widehat{a})$ satisfies Item 1 in Definition \ref{mSORTEstrongmsorte}.
Observe also that $\widehat{Y}\in\mathcal{B}\Rightarrow\widetilde{Y}_X\in\mathcal{B}$ since $\mathcal{B}+\mathbb{R}^N=\mathcal{B}$, and that $\sum_{j=1}^N\widetilde{Y}_X^j=\sum_{j=1}^N\widehat{Y}^j-\sum_{j=1}^Na^j=A-A=0$, so that also Item 2 in Definition \ref{mSORTEstrongmsorte} is satisfied. Also, by its very definition definition  $\mathbb{E}_{\probq^j_X}[\widetilde{Y}^j_X]=\mathbb{E}_{\widehat{\probq}^j}[\widehat{Y}^j]-\mathbb{E}_{\widehat{\probq}^j}[\widehat{Y}^j]=0$.

Let finally $(Y_{X},{\mathbb{Q}}_{X},a_X)\in \mathcal{L}\mathbf{\times }\mathcal{Q}_{\mathcal{B},V}%
\mathbf{\times }\mathbb{R}^{N}$ be a mSORTE as in Definition \ref%
{mSORTEstrongmsorte}. We prove the Nash Equilibrium property \eqref{eqnashequil}.
  For any $Z\in \mathcal{L}^{j}$, using Item 1 of Definition \ref{mSORTEstrongmsorte}

  \begin{align*}
\mathbb{E}_{\mathbb{P}}\left[ U(a_X+X+\widetilde{Y}_X-\mathbb{E}_{\probq_X}[\widetilde{Y}_X])\right]&=\sup\left\{\mathbb{E}_{\mathbb{P}}\left[ U(a_X+X+\widetilde{Y}-\mathbb{E}_{\probq_X}[\widetilde{Y}])\right]\mid \widetilde{Y}\in\mathcal{L}\right\}\\
&\geq \sup_{Z\in\mathcal{L}^j}\mathbb{E}_{\mathbb{P}}\left[ U\left( a_{X}+X+[\widetilde{Y}%
_{X}^{[-j]};Z]-{\mathbb{E}_{{{\mathbb{Q}}}_{X}}[\widetilde{Y}_{X}^{[-j]};}%
Z]\right) \right]\\
&\geq \mathbb{E}_{\mathbb{P}}\left[ U(a_X+X+\widetilde{Y}_X-\mathbb{E}_{\probq_X}[\widetilde{Y}_X])\right]\,.
\end{align*}
\end{proof}

\paragraph{On uniqueness of a mSORTE \label{mSORTEUNIQUENESSINITIAL}}

\begin{theorem}
\label{mSORTEthmstrongunique} Under Assumption \ref{mSORTEA1}, suppose additionally that $U$ is
differentiable. Then there exists a unique mSORTE  $(\widetilde{Y}_X,{\mathbb{Q}_X},%
a_X)$ with $\mathbb{E}_{\probq^j_X}[\widetilde{Y}^j_X]=0$ for every $j=1,\dots,N$.
\end{theorem}

\begin{proof}
Take a mSORTE $(\widetilde{Y}_X,{\mathbb{Q}_X},%
a_X)$ with $\mathbb{E}_{\probq^j_X}[\widetilde{Y}^j_X]=0$ for every $j=1,\dots,N$.
Set $\widehat{Y}:=\widetilde{Y}_X+a_X-\mathbb{E}_{\probq_X}[\widetilde{Y}_X]=\widetilde{Y}_X+a_X$.
We claim that  $\widehat{Y}$ is an optimizer of RHS of %
\eqref{mSORTEeqonmphiequalLv} and $\probq_X$ is an optimizer
of \eqref{mSORTEeqminimaxgeneralintro1}.
 Observe that $\widehat{Y}\in \mathcal{B}%
_{A}\cap \mathcal{L}$ (using $\mathcal{B}+\R^N=\mathcal{B}$ and $\sum_{j=1}^N\widehat{Y}^j=\sum_{j=1}^N\widetilde{Y}^j_X+\sum_{j=1}^Na^j_X-\sum_{j=1}^N\mathbb{E}_{\probq^j_X}[\widetilde{Y}^j_X]=0+A-0$)  and ${{\mathbb{Q}_X}}\in \mathcal{Q}_{\mathcal{B},V}.$ Since $\widehat{Y}\in \mathcal{B}%
_{A}\cap \mathcal{L}$ and
we are assuming that the set $\mathcal{B}$ is closed under truncation,
 by Lemma \ref{lemmafairnessgeneral} we have that $\sum_{j=1}^N\mathbb{E}_{\mathbb{Q}%
^j} \left[\widehat{Y}^j\right]\leq A$ for all ${\mathbb{Q}}\in\mathcal{Q}_{\mathcal{B},V}$.
 As ${{\mathbb{Q}_X}}\in 
\mathcal{Q}_{\mathcal{B},V}$, we then obtain
\begin{align}
&\,\mathbb{E}_{\mathbb{P}}\left[ U(X+\widehat{Y})\right] \leq \sup\left\{\mathbb{E}_{\mathbb{P}}\left[ U(X+Y)\right]\mid Y\in\mathcal{%
L},\,\sum_{j=1}^N\mathbb{E}_{\mathbb{Q}^j} \left[Y^j\right]\leq A\,\forall{%
\mathbb{Q}}\in\mathcal{Q}_{\mathcal{B},V}\right\} 
\\
&\leq\sup\left\{%
\mathbb{E}_{\mathbb{P}}\left[ U(X+Y)\right]\mid Y\in\mathcal{L}%
,\,\sum_{j=1}^N\mathbb{E}_{\mathbb{Q}_X^j} \left[Y^j\right]\leq A\right\} 
\notag \\
&\leq \inf_{\lambda \geq 0}\left( \lambda \left( \sum_{j=1}^{N}\mathbb{E}_{%
{\mathbb{Q}_X}^{j}}\left[ X^{j}\right] +A\right) +\mathbb{E}_{\mathbb{P%
}}\left[ V\left( \lambda \frac{\mathrm{d}{{\mathbb{Q}_X}}}{\mathrm{d}{%
\mathbb{P}}}\right) \right] \right)  \label{mSORTEeq999} \\
&\overset{\text{Thm.}\ref{mSORTEthmmaingeneral2}}{=}\sup \left\{ \mathbb{E}_{%
\mathbb{P}}\left[ U(X+Y)\right] \mid Y\in \mathcal{L},\,\sum_{j=1}^{N}%
\mathbb{E}_{{\mathbb{Q}_X}^{j}}\left[ Y^{j}\right] \leq A\right\} 
\notag \\
&=\sup\left\{
\sup \left\{ \mathbb{E}_{\mathbb{P}}\left[ U(X+Y)\right] \mid Y\in \mathcal{L%
},\,\mathbb{E}_{{\mathbb{Q}_X}^{j}}\left[ Y^{j}\right] \leq
a^{j}\,\forall j\right\}\mid  a\in {\mathbb{R}}^{N}, \sum_{j=1}^{N}a_{j}=A \right\}  \label{mSORTEeq666} \\
&=\sup \left\{ \sup_{\widetilde{Y}\in\mathcal{L}}\mathbb{E}_{\mathbb{P}}\left[ U\left( a+X+%
\widetilde{Y}-{\mathbb{E}_{{{\mathbb{Q}}}_{X}}[\widetilde{Y}]}\right) \right]
\mid a\in {\mathbb{R}}^{N},\,\sum_{j=1}^{N}a^{j}=A\right\}\label{mSORTEeq661}\\
&=\Ep{U\left(a_X+X+\widetilde{Y}_X-\mathbb{E}_{\probq_X}[\widetilde{Y}_X]\right)}=\mathbb{E}_{\mathbb{P}}\left[ U(X+\widehat{Y})\right] \,
\label{mSORTEeq777}
\end{align}%
where \eqref{mSORTEeq999} is a consequence of Fenchel inequality \eqref{mSORTEfenchelineq}, the expression
in \eqref{mSORTEeq666} is a reformulation of the one in the previous line, \eqref{mSORTEeq661} follows from Lemma \ref{mSORTESequalsH}
and \eqref{mSORTEeq777} holds true because$(\widetilde{Y}_X,{\mathbb{Q}_X},%
a_X)$ is a mSORTE and therefore $(\widetilde{Y}_X,
a_X)$ is an optimizer of
the problem in \eqref{mSORTEeq661}. Notice that Theorem \ref%
{mSORTEthmmaingeneral2} guarantees that the $\inf $ in \eqref{mSORTEeq999}
is a $\min $. We then deduce that all above inequalities are equalities and $%
\widehat{Y}$ is an optimizer of RHS of \eqref{mSORTEeqonmphiequalLv} and $%
{{\mathbb{Q}}_X}$ is an optimizer of %
\eqref{mSORTEeqminimaxgeneralintro1}.
Now, take another mSORTE $(\widetilde{Z}_X,\mathbb{D}_X,b_X)$ with $\mathbb{E}_{\mathbb{D}^j_X}[\widetilde{Z}^j_X]=0$ for every $j=0,\dots,N$. Arguing exactly as above for $\widehat{Z}:=\widetilde{Z}_X+b_X$ we get that $%
\widehat{Z}$ is an optimizer of RHS of \eqref{mSORTEeqonmphiequalLv} and $\mathbb{D}_X$ is an optimizer of %
\eqref{mSORTEeqminimaxgeneralintro1}.  Theorem \ref%
{mSORTEthmmaingeneral1} Item 3 yields $\widehat{Z}=\widehat{Y}$ and $\probq_X=\mathbb{D}_X$. Taking expectations componentwise we get that $b_X^j=\mathbb{E}_{\mathbb{D}^j_X}[\widehat{Z}^j]=\mathbb{E}_{\probq^j_X}[\widehat{Y}^j]=a^j_X$ for every $j=1,\dots,N$, which yields $a_X=b_X$. Finally, from the definitions of $\widehat{Z},\widehat{Y}$ we get $\widetilde{Z}_X=\widetilde{Y}_X$.
\end{proof}

\begin{corollary}
\label{mSORTEcoroptsumA} Under Assumption \ref{mSORTEA1} and if $U$ is
differentiable there exists a unique optimum $\widehat{Y}$ for the problem $\Pi _{A}^{\mathrm{ran}}(X)$ in \eqref{mSORTEproblem random}.
Such an optimum is given by $\widehat{Y}=\widetilde{Y}_X+a_X$, where $(\widetilde{Y}_X,{\mathbb{Q}_X},%
a_X)$ is the  unique mSORTE   with  $\mathbb{E}_{\probq^j_X}[\widetilde{Y}^j_X]=0$ for every $j=1,\dots,N$ and budget $A$.
\end{corollary}

\begin{proof}
Take $\widehat{Y}=\widetilde{Y}_X+a_X$ for  $(\widetilde{Y}_X,{\mathbb{Q}_X},%
a_X)$ as described in the statement (Theorems %
\ref{mSORTEthmmaingeneral1} and \ref{mSORTEthmstrongunique}). By Lemma \ref{lemmafairnessgeneral}, $\Pi _{A}^{\mathrm{ran}}(X)\leq \text{RHS of }%
\eqref{mSORTEeqonmphiequalLv}$. Since $\widehat{Y}\in\mathcal{L}\cap\mathcal{%
B}_A$ is an optimum for RHS of \eqref{mSORTEeqonmphiequalLv} (see the proof
of Theorem \ref{mSORTEthmstrongunique}), existence follows. Uniqueness is a
consequence of strict concavity of $U$, by standard arguments. The arguments
above also show automatically the link between the unique optimum for $\Pi _{A}^{\mathrm{ran}}(X)$
and  the unique mSORTE for the
budget $A$.
\end{proof}
\subsection{Explicit computation in an exponential framework \label{secexplicitformulas}}
With a particular choice of the utility function $U$ (see \eqref{partcaseU} below) we can explicitly compute primal ($\widehat{Y}$) and dual ($\widehat{\probq}$) optima in Theorem \ref{mSORTEthmmaingeneral1}. Consequently, the (unique) mSORTE can be explicitly computed. Recall from  the proof of Theorem \ref{mSORTEthmmsorteexists} that the mSORTE is produced by setting ${a}_X^{j}:=\mathbb{E}_{%
\widehat{{\mathbb{Q}}}^{j}}[\widehat{Y}^{j}]$, $\widetilde{Y}_X=\widehat{Y}-a_X$, $\probq_X=\widehat{\probq}$.
As we are able to find explicit formulas, we prefer to anticipate this example even if the proof of Proposition \ref{propexplformulas} relies on two propositions in Section \ref{mSORTEsecexamples}.

\begin{proposition}
\label{propexplformulas}
Take $\alpha_1,\dots,\alpha_N>0$ and consider for $x=[x^1,\dots,x^N]\in\R^N$ 
\begin{equation}
\label{partcaseU}
U(x^1,\dots,x^N):=\frac{1}{2}\sum_{j=1}^N \left(1-e^{-2\alpha_j x^j}\right)+\frac{1}{2}\sum_{\substack{i,j\in\{1,\dots,N\}\\i\neq j}}\left(1-e^{-(\alpha_ix^i+\alpha_jx^j)}\right)\,.
\end{equation}

Select $\mathcal{B}=\mathcal{C}_\R$ and define 

$$\beta=\sum_{j=1}^N\frac{1}{\alpha_j}\,,\,\,\,\,\Gamma=\sum_{j=1}^N\frac{1}{\alpha_j}\log\left(\frac{1}{\alpha_j}\right)\,,\,\,\,\overline{X}:=\sum_{j=1}^NX^j\,.$$
Then the dual optimum in LHS of \eqref{mSORTEeqminimaxgeneralintro1} is given by
\begin{equation}
\label{eqoptimumexp}
\rn{\widehat{\probq}}{\probp}=\frac{\exp\left(-\frac{2\overline{X}}{\beta}\right)}{\Ep{\exp\left(-\frac{2\overline{X}}{\beta}\right)}}
\end{equation}
and the primal optimum $\widehat{Y}$ in LHS of \eqref{mSORTEeqonmphiequalLv} is given by 
\begin{align}
\widehat{Y}^j=&-X^j+\frac{1}{\beta\alpha_j}\overline{X}+\kappa_X\label{formulayhat}\\
\kappa_X=&-\left[\frac{1}{\alpha_j}+\frac{1}{\alpha_j}\log\left(\frac{1}{\alpha_j}\right)-\frac{1}{2\alpha_j}\log\left(\beta\right)\right]-\frac{1}{2\alpha_j}\left[\log(\widehat{\lambda}_X)-\log\left(\Ep{\exp\left(-\frac{2\overline{X}}{\beta}\right)}\right)\right]\notag\\
\widehat{\lambda}_X=&\exp\left(-\frac{2}{\beta}\left[A+\beta+\Gamma-\frac{\beta\log(\beta)}{2}-\frac{\beta}{2}\log\left(\Ep{\exp\left(-\frac{2\overline{X}}{\beta}\right)}\right)\right]\right)\,.\label{formulalambdahat}
\end{align}

\end{proposition}

\begin{proof}
To begin with, one can easily check that the assumptions of Proposition \ref{mSORTEPropBLambda} and Proposition \ref{mSORTEPropB} are satisfied, and so is then Standing Assumption  I. Standing Assumption II is trivially satisfied, once we take $X\in M^\Phi$. Also, $U(x)=\frac{N^2}{2}-\frac12\left(\sum_{j=1}^N e^{-\alpha_j x^j}\right)^2$. The conjugate $V$, as well as its gradient, are computed explicitly in Lemma \ref{lemmacomputeV}.

Our aim is to show that primal and dual optima, whose existence and uniqueness are  stated in Theorem \ref{mSORTEthmmaingeneral1} Item 2 and 3, are given by $\widehat{Y}^j:=-X^j-\frac{\partial V}{\partial w^j}\left(\widehat{\lambda}_X\rn{\widehat{\probq}}{\probp},\dots,\widehat{\lambda}_X\rn{\widehat{\probq}}{\probp}\right)$ and $\left(\widehat{\lambda}_X,\left[\rn{\widehat{\probq}}{\probp},\dots,\rn{\widehat{\probq}}{\probp}\right]\right)\in(0,+\infty)\times\mathcal{Q}_{\mathcal{B},V}$, for  $\widehat{\lambda}_X$ and  $\widehat{\probq}$ given in \eqref{formulalambdahat} and \eqref{eqoptimumexp} respectively.
 By direct computation one can then check that $\probp\sim[\widehat{\probq},\dots,\widehat{\probq}]\in\mathcal{Q}_{\mathcal{B},V}$ and $\widehat{Y}\in M^\Phi$ whenever $X\in M^\Phi$,
  where in view of Proposition \ref{mSORTEUaddLambda} and Proposition \ref{mSORTEPropB} we have $M^\Phi=M^{\exp}\times\dots\times M^{\exp}$ for $$M^{\exp}:=\left\{ X\in \left( L^{0}\left( (\Omega ,\mathcal{F},{\mathbb{P}}%
);[-\infty ,+\infty ]\right) \right)\mid \,\forall \,\lambda \in
(0,+\infty ),\mathbb{E}_{\mathbb{P}}\left[ \exp(\lambda \left\vert
X\right\vert )\right] <+\infty \right\}\,.$$

Recall now that
$U(-\nabla V (w))=V(w)-\sum_{j=1}^Nw^j\frac{\partial V}{\partial w^j}(w)$ for every $w\in (0,+\infty)^N$ (see \cite{Ro70} Chapter V). 
Then, given \eqref{sumgradient}, we can write 
\begin{align*}
&\Ep{U(X+\widehat{Y})}=\Ep{U\left(-\nabla V\left(\widehat{\lambda}_X\rn{\widehat{\probq}}{\probp},\dots,\widehat{\lambda}_X\rn{\widehat{\probq}}{\probp}\right)\right)}\\
&=-\widehat{\lambda}_X\left[\beta+\Gamma-\frac{\beta\log(\beta)}{2}+\frac{\beta}{2}\Ep{\rn{\widehat{\probq}}{\probp}\log\left(\rn{\widehat{\probq}}{\probp}\right)}\right]-\frac{\beta}{2}\widehat{\lambda}_X\log(\widehat{\lambda}_X)+\Ep{V\left(\widehat{\lambda}_X\rn{\widehat{\probq}}{\probp},\dots,\widehat{\lambda}_X\rn{\widehat{\probq}}{\probp}\right)}\\
&=\widehat{\lambda}_X
\left( \mathbb{E}_{\widehat{\mathbb{Q}}}\left[ \sum_{j=1}^{N}X^{j}\right]
+A\right) +\mathbb{E}_{\mathbb{P}}\left[ V\left(\widehat{\lambda}_X\rn{\widehat{\probq}}{\probp},\dots,\widehat{\lambda}_X\rn{\widehat{\probq}}{\probp} \right) \right]
\end{align*}
for $\rn{\widehat{\probq}}{\probp}$ and  $\widehat{\lambda}_X$ given in \eqref{eqoptimumexp} and \eqref{formulalambdahat} respectively.
%
%
%
%
In view of Theorem \ref{mSORTEthmmaingeneral1} primal and dual optimality then follow, once we check that $\widehat{Y}\in \mathcal{B}_A\cap M^\Phi$.
Given $\widehat{\probq}$ and $\widehat{\lambda}_X$, $\widehat{Y}$ can be directly computed as \eqref{formulayhat} and we can check that $\sum_{j=1}^N \widehat{Y}^j=A$, which completes the proof. 
 \end{proof}
 \begin{remark}
 The key point in the proof of Proposition \ref{propexplformulas} above was guessing the particular form of $\widehat{\probq}$. Such a guess is a consequence of imposing $\sum_{j=1}^N\widehat{Y}^j=A$ for the candidate optimum $\widehat{Y}$. Indeed, imposing that $\sum_{j=1}^N\widehat{Y}^j=A$ for $\widehat{Y}^j:=-X^j-\frac{\partial V}{\partial w^j}\left(\widehat{\lambda}\rn{\widehat{\probq}}{\probp},\dots,\widehat{\lambda}\rn{\widehat{\probq}}{\probp}\right)$ ($\widehat{\lambda}>0$ is a constant, yet to be found precisely at this initial stage), we get that for some  $\eta\in\R$
$$A=\sum_{j=1}^N \widehat{Y}^j=-\sum_{j=1}^NX^j+\eta-\sum_{j=1}^N\frac{1}{2\alpha_j}\log\left(\widehat{\lambda}\rn{\widehat{\probq}}{\probp}\right)=-\overline{X}+\eta-\frac{\beta}{2}\log\left(\widehat{\lambda}\rn{\widehat{\probq}}{\probp}\right)\,.$$
Here, we also used   the explicit formula \eqref{formulagradient} for the gradient of $V$. This computation implies that $\widehat{\probq}$ needs to satisfy \eqref{eqoptimumexp}.
 \end{remark}
\subsection{Dependence on $X$ of mSORTE}

We study here the dependence of mSORTE on the initial data $X$. We recall
that both existence and uniqueness are guaranteed (see Theorem \ref%
{mSORTEthmmsorteexists} and Theorem \ref{mSORTEthmstrongunique}).

\begin{proposition}
\label{mSORTEpropYfsumX} Under the same assumptions of Theorem \ref%
{mSORTEthmstrongunique} and for $\mathcal{B}=\mathcal{C}_{\mathbb{R}}$,
given the unique mSORTE $(\widetilde{Y}_X,{\mathbb{Q}_X},%
a_X)$ with $\mathbb{E}_{\probq^j_X}[\widetilde{Y}^j_X]=0$ for every $j=1,\dots,N$, the
variables $\frac{\mathrm{d}{\probq_X}}{\mathrm{d}{\mathbb{P}}}$
and $X+a_X+\widetilde{Y}_X$ are $\sigma (X^{1}+\dots +X^{N})$ (essentially)
measurable.
\end{proposition}

\begin{proof}
By Theorem \ref{mSORTEthmstrongunique} there exists a unique mSORTE. Recall
the proof of Theorem \ref{mSORTEthmmsorteexists}, where we showed that the
optimizers $(\widehat{Y},\widehat{{\mathbb{Q}}})$ in Theorem \ref%
{mSORTEthmmaingeneral1}, together with $\widehat{a}^{j}:=\mathbb{E}_{%
\widehat{{\mathbb{Q}}}^{j}}[\widehat{Y}^{j}]$, $j=1,\dots ,N,$ yield the
mSORTE via $a_X=\widehat{a}$, $\probq_X=\widehat{\probq}$, $\widetilde{Y}_X=\widehat{Y}-\widehat{a}$. Notice that in this specific case $Y:={e_{i}}1_{A}-{e_{j}}1_{A}\in 
\mathcal{B}\cap M^{\Phi }$ for all $i,j$. The same argument used in the
proof of \cite{BDFFM} Proposition 4.18 can be then applied with obvious
minor modifications (i.e. using $V(\cdot)$ in place of $\sum_{j=1}^Nv_j(%
\cdot)$ and taking any ${\mathbb{Q}}\in\mathcal{Q}_{\mathcal{B},V}$) to show that $\frac{%
\mathrm{d}\widehat{{\mathbb{Q}}}}{\mathrm{d}{\mathbb{P}}}$ is $\mathcal{G}%
:=\sigma(X^1+\dots+X^N)$-(essentially) measurable. We stress the fact that,
similarly to \cite{BDFFM} Proposition 4.18, all the components of any ${%
\mathbb{Q}}\in\mathcal{Q}_{\mathcal{B},V}$ are equal.

We now focus on $X+\widehat{Y}$: consider $\widehat{Z}:=\mathbb{E}_\mathbb{P}
\left[X+\widehat{Y}\mid\mathcal{G}\right]-X$ (the conditional expectation is taken
componentwise). Then it is easy to check that $\sum_{j=1}^N\widehat{Z}%
^j=\sum_{j=1}^N \widehat{Y}^j=A$ which yields $\widehat{Z}\in\mathcal{B}_A$.
We now prove that $\widehat{Z}\in \mathcal{L}=\bigcap_{{\mathbb{Q}}\in\mathcal{Q}%
_{\mathcal{B},V}}L^1({\mathbb{Q}})$. This will imply that $\sum_{j=1}^N\mathbb{E}_{\mathbb{%
Q}^j} \left[\widehat{Z}^j\right]\leq A$ for all ${\mathbb{Q}}\in\mathcal{Q}%
_{\mathcal{B},V} $, by Lemma \ref{lemmafairnessgeneral}. Since $X\in
M^\Phi $, it is clearly enough to prove that $\mathbb{E}_\mathbb{P} \left[X+%
\widehat{Y}\middle|\mathcal{G}\right]\in \mathcal{L}$. Observe first that
for any given ${\mathbb{Q}}\ll{\mathbb{P}}$, the measure ${\mathbb{Q}}_%
\mathcal{G}$ defined by $\frac{\mathrm{d}{\mathbb{Q}}_\mathcal{G}}{\mathrm{d}%
{\mathbb{P}}}:=\mathbb{E}_\mathbb{P} \left[\frac{\mathrm{d}{\mathbb{Q}}^j}{%
\mathrm{d}{\mathbb{P}}}\middle|\mathcal{G}\right]$ satisfies

\begin{equation}  \label{mSORTEqginqv}
{\mathbb{Q}}\in\mathcal{Q}_{\mathcal{B},V}\Longrightarrow{\mathbb{Q}}_\mathcal{G}\in%
\mathcal{Q}_{\mathcal{B},V}\,.
\end{equation}

To see this, recall that all the components of ${\mathbb{Q}}$ are equal,
hence so are those of ${\mathbb{Q}}_\mathcal{G}$. Moreover 
\begin{equation*}
\sum_{j=1}^N\mathbb{E}_\mathbb{P} \left[Y^j\frac{\mathrm{d}{\mathbb{Q}}^j_%
\mathcal{G}}{\mathrm{d}{\mathbb{P}}}\right]=\mathbb{E}_\mathbb{P} \left[%
\sum_{j=1}^NY^j\frac{\mathrm{d}{\mathbb{Q}}^1_\mathcal{G}}{\mathrm{d}{%
\mathbb{P}}}\right]=\sum_{j=1}^N Y^j\leq 0\,\,\,\, \forall\, Y\in\mathcal{B}%
_0\cap M^\Phi
\end{equation*}
and $\mathbb{E}_\mathbb{P} \left[V\left(\lambda\frac{\mathrm{d}{\mathbb{Q}}%
^j_\mathcal{G}}{\mathrm{d}{\mathbb{P}}}\right)\right]\leq \mathbb{E}_\mathbb{%
P} \left[V\left(\lambda\frac{\mathrm{d}{\mathbb{Q}}}{\mathrm{d}{\mathbb{P}}}%
\right)\right]$ by conditional Jensen inequality.

Now, for any $j=1,\dots,N$ and ${\mathbb{Q}}\in\mathcal{Q}_{\mathcal{B},V}$ 
\begin{equation*}
\mathbb{E}_\mathbb{P} \left[\left|\mathbb{E}_\mathbb{P} \left[X^j+Y^j\middle|%
\mathcal{G}\right]\right|\frac{\mathrm{d}{\mathbb{Q}}^j}{\mathrm{d}{\mathbb{P%
}}}\right]\leq\mathbb{E}_\mathbb{P} \left[\mathbb{E}_\mathbb{P} \left[%
\mathbb{E}_\mathbb{P} \left[\left|X^j+Y^j\right|\middle|\mathcal{G}\right]%
\frac{\mathrm{d}{\mathbb{Q}}^j}{\mathrm{d}{\mathbb{P}}}\middle|\mathcal{G}%
\right]\right]
\end{equation*}
\begin{equation*}
=\mathbb{E}_\mathbb{P} \left[\mathbb{E}_\mathbb{P} \left[\left|X^j+Y^j\right|%
\mathbb{E}_\mathbb{P} \left[\frac{\mathrm{d}{\mathbb{Q}}^j}{\mathrm{d}{%
\mathbb{P}}}\middle|\mathcal{G}\right]\middle|\mathcal{G}\right]\right]=%
\mathbb{E}_\mathbb{P} \left[\left|X^j+Y^j\right|\mathbb{E}_\mathbb{P} \left[%
\frac{\mathrm{d}{\mathbb{Q}}^j}{\mathrm{d}{\mathbb{P}}}\middle|\mathcal{G}%
\right]\right]\,.
\end{equation*}
As a consequence, since by \eqref{mSORTEqginqv} $\mathcal{L} \subseteq L^1({%
\mathbb{Q}}_\mathcal{G})$ and $\widehat{Y}\in \mathcal{L}$, we get $X+%
\widehat{Y}\in L^1({\mathbb{Q}})$, and the fact that $\widehat{Z}\in 
\mathcal{L}$ follows.

Finally, observe that $\mathbb{E}_\mathbb{P} \left[U\left(X+\widehat{Z}%
\right)\right]= \mathbb{E}_\mathbb{P} \left[U\left(\mathbb{E}_\mathbb{P} %
\left[X+\widehat{Y}\middle|\mathcal{G}\right]\right)\right]\geq \mathbb{E}_%
\mathbb{P} \left[U(X+\widehat{Y})\right]$ by conditional Jensen inequality.
Hence $\widehat{Z}$ is another optimum for the optimization problem in RHS
of \eqref{mSORTEeqonmphiequalLv}. By strict concavity of $U$ then we get $%
\widehat{Y}=\widehat{Z}$. Since $X+\widehat{Z}$ is $\mathcal{G}$%
-(essentially) measurable, so is clearly $X+\widehat{Y}$.
\end{proof}

It is interesting to notice that this dependence on the componentwise sum of 
${X}$ also holds in the case of SORTE (\cite{BDFFM} Section 4.5) and of B%
\"{u}hlmann's equilibrium (see \cite{Buhlmann} page 16, which partly
inspired the proof above, and \cite{Borch}).

\begin{remark}
In the case of clusters of agents, the above result can be clearly
generalized (see \cite{BDFFM} Remark 4.19).
\end{remark}

\section{Systemic utility maximization and duality}

\label{mSORTEsecutmaxanddual}

In this Section we collect some remarks and properties of the polar cone of $%
\mathcal{B}_0\cap M^\Phi$, which will play an important role in the
following.

\begin{remark}
\label{mSORTEremcomponentwiseintegr} If $X\in M^{\Phi }$, then for any fixed 
$k=1,\dots ,N$ we have $[0,\dots ,0,X^{k},0,\allowbreak\dots ,0]\in M^{\Phi
} $. This in turns implies that for any $Z\in K_{\Phi }$ and $X\in M^{\Phi }$%
, $X^{j}Z^{j}\in L^{1}({\mathbb{P}})$ for any $j=1,\dots ,N$.
\end{remark}

\begin{remark}
\label{mSORTEremarkpolarisnice} In the dual pair $(M^{\Phi },K_{\Phi })$
take the polar $(\mathcal{B}_{0}\cap M^{\Phi })^{0}$ of $\mathcal{B}_{0}\cap
M^{\Phi } $. Since all (deterministic) vector in the form $e^{i}-e^{j}$
belong to $\mathcal{B}_{0}\cap M^{\Phi }$, we have that for all $Z\in (%
\mathcal{B}_{0}\cap M^{\Phi })^{0}$ and for all $i,j\in \{1,\dots ,N\}$ $%
\mathbb{E}_{\mathbb{P}}\left[ Z^{i}\right] -\mathbb{E}_{\mathbb{P}}\left[
Z^{j}\right] \leq 0$. It is clear that, as a consequence, $Z\in (\mathcal{B}%
_{0}\cap M^{\Phi })^{0}\Rightarrow \mathbb{E}_{\mathbb{P}}\left[ Z^{1}\right]
=\dots =\mathbb{E}_{\mathbb{P}}\left[ Z^{N}\right] $. Recall that ${\mathbb{R%
}}_{+}:=\{b\in {\mathbb{R}},b\geq 0\}$ and the definition of $\mathcal{Q}$
provided in \eqref{mSORTEdefsetQ}. We then see: 
\begin{equation}
(\mathcal{B}_{0}\cap M^{\Phi })^{0}\cap (L_{+}^{0})^{N}={\mathbb{R}}%
_{+}\cdot \mathcal{Q}  \label{mSORTEpolarityisnice}
\end{equation}%
that is, $(\mathcal{B}_{0}\cap M^{\Phi })^{0}$ is the cone generated by $%
\mathcal{Q}$.
\end{remark}

\begin{remark}
\label{mSORTEremthereexistsanequiv} The condition $\mathcal{B}\subseteq 
\mathcal{C}_{{\mathbb{R}}}$ implies $\mathcal{B}_{0}\cap M^{\Phi }\subseteq (%
\mathcal{C}_{{\mathbb{R}}}\cap M^{\Phi }\cap \{\sum_{j=1}^{N}Y^{j}\leq 0\})$%
, so that the polars satisfy the opposite inclusion: $(\mathcal{C}_{{\mathbb{%
R}}}\cap M^{\Phi }\cap \{\sum_{j=1}^{N}Y^{j}\leq 0\})^{0}\subseteq (\mathcal{%
B}_{0}\cap M^{\Phi })^{0}$. Observe now that any vector $[Z,\dots ,Z]$, for $%
Z\in L_{+}^{\infty }$, belongs to $(\mathcal{C}_{{\mathbb{R}}}\cap M^{\Phi
}\cap \{\sum_{j=1}^{N}Y^{j}\leq 0\})^{0}$. In particular, take $a>0$ for
which \eqref{mSORTElemmacontrolwithline} is satisfied. Then 
\begin{equation*}
V([a,\dots,a])\leq\sup_{x\in{\mathbb{R}}^N}\left(a\sum_{j=1}^{N}x^{j}-a%
\sum_{j=1}^{N}((x^{j})^{-})+b-a\sum_{j=1}^Nx^j\right)\leq b<+\infty\,.
\end{equation*}
It follows that $[{\mathbb{P}},\dots,{\mathbb{P}}]\in\mathcal{Q}_{\mathcal{B},V}\neq \emptyset$ since $\mathbb{E}_\mathbb{P} \left[V\left(a\left[\frac{%
\mathrm{d}{\mathbb{P}}}{\mathrm{d}{\mathbb{P}}},\dots,\frac{\mathrm{d}{%
\mathbb{P}}}{\mathrm{d}{\mathbb{P}}}\right]\right)\right]<+\infty$.

\end{remark}

\begin{proposition}[Fairness]
\label{mSORTEpropfairprob} For all ${\mathbb{Q}}\in \mathcal{Q}$ 
\begin{equation*}
\sum_{j=1}^{N}\mathbb{E}_{\mathbb{Q}^{j}}\left[ Y^{j}\right] \leq
\sum_{j=1}^{N}Y^{j}\,\,\,\,\forall \,Y\in \mathcal{B}\cap M^{\Phi }\,.
\end{equation*}
\end{proposition}

\begin{proof}
Let $Y\in \mathcal{B}\cap M^{\Phi }$. Notice that the hypothesis ${\mathbb{R}%
}^{N}+\mathcal{B}=\mathcal{B}$ implies that the vector $Y_{0}$, defined by $%
Y_{0}^{j}:=Y^{j}-\frac{1}{N}\sum_{k=1}^{N}Y^{k}$, belongs to $\mathcal{B}%
_{0} $. Indeed, $\sum_{k=1}^{N}Y^{k}\in {\mathbb{R}}$ and so $Y_{0}\in 
\mathcal{B} $ and $\sum_{j=1}^{N}Y_{0}^{j}=0$. By definition of polar, $%
\sum_{j=1}^{N}\mathbb{E}_{\mathbb{P}}\left[ Y_{0}^{j}Z^{j}\right] \leq 0$
for all $Z\in (\mathcal{B}\cap M^{\Phi })^{0}$, and in particular for all ${%
\mathbb{Q}}\in \mathcal{Q}$ 
\begin{equation*}
0\geq \sum_{j=1}^{N}\mathbb{E}_{\mathbb{P}}\left[ Y_{0}^{j}\frac{\mathrm{d}{%
\mathbb{Q}}^{j}}{\mathrm{d}{\mathbb{P}}}\right] =\sum_{j=1}^{N}\mathbb{E}_{%
\mathbb{P}}\left[ Y^{j}\frac{\mathrm{d}{\mathbb{Q}}^{j}}{\mathrm{d}{\mathbb{P%
}}}\right] -\sum_{j=1}^{N}\mathbb{E}_{\mathbb{P}}\left[ \frac{1}{N}\left(
\sum_{k=1}^{N}Y^{k}\right) \frac{\mathrm{d}{\mathbb{Q}}^{j}}{\mathrm{d}{%
\mathbb{P}}}\right] \,
\end{equation*}
and we recognize $\sum_{j=1}^{N}\mathbb{E}_{\mathbb{Q}^{j}}\left[ Y^{j}%
\right] -\sum_{j=1}^{N}Y^{j}$ in RHS.
\end{proof}

\begin{theorem}
\label{mSORTEthmminimax} Let $\mathcal{C}\subseteq M^\Phi$ be a convex cone
with $0\in\mathcal{C}$ and $e_i-e_j\in\mathcal{C}$ for every $%
i,j\in\{1,\dots,N\}$. Denote by $\mathcal{C}^0$ the polar of the cone $%
\mathcal{C}$ in the dual pair $(M^\Phi,K_\Phi)$: 
\begin{equation*}
\mathcal{C}^0:=\left\{Z\in K_\Phi\mid \sum_{j=1}^N\mathbb{E}_\mathbb{P} %
\left[Y^jZ^j\right]\leq 0\,\,\forall\,Y\in\mathcal{C}\right\}
\end{equation*}
and set 
\begin{equation*}
\mathcal{C}^0_1:=\left\{Z\in\mathcal{C}^0\mid \mathbb{E}_\mathbb{P} \left[Z^j%
\right]=1\,\forall j=1,\dots,N\right\}
\end{equation*}
\begin{equation*}
(\mathcal{C}^0_1)^+:=\left\{Z\in\mathcal{C}^0_1\mid Z^j\geq 0\,\forall
j=1,\dots,N\right\}\,.
\end{equation*}
Suppose that for every $X\in M^\Phi$ 
\begin{equation*}
\sup_{Y\in\mathcal{C}}\mathbb{E}_\mathbb{P} \left[U(X+Y)\right]<+\infty\,.
\end{equation*}
Then the following holds:

\begin{equation}  \label{mSORTEeqwithCminimax}
\sup_{Y\in \mathcal{C}}\mathbb{E}_{\mathbb{P}}\left[ U(X+Y)\right]
=\min_{\lambda \geq 0,\,{\mathbb{Q}}\in (\mathcal{C}_{1}^{0})^{+}}\left(
\lambda \sum_{j=1}^{N}\mathbb{E}_{{\mathbb{Q}}^j}\left[ X^{j}\right] +%
\mathbb{E}_{\mathbb{P}}\left[ V\left( \lambda \frac{\mathrm{d}{\mathbb{Q}}}{%
\mathrm{d}{\mathbb{P}}}\right) \right] \right) \,.
\end{equation}%
If any of the two expressions is strictly smaller than $V(0)=\sup_{{\mathbb{R%
}}^{N}}U$, then the condition $\lambda\geq 0$ in \eqref{mSORTEeqwithCminimax}
can be replaced with the condition $\lambda>0$.
\end{theorem}

\begin{proof}
The proof can be obtained with minor and obvious modifications of the one in 
\cite{BDFFM} Theorem A.3 by replacing $\sum_{j=1}^N
u_j(\cdot),\,\sum_{j=1}^N v_j(\cdot), L^{\Phi^*}$ there with $%
U(\cdot),\,V(\cdot),\,K_\Phi$ respectively.
\end{proof}

We also provide an analogous result when working with the pair $((L^\infty({%
\mathbb{P}}))^N,(L^1({\mathbb{P}}))^N)$ in place of $(M^\Phi, K_\Phi)$,
which will be used in Section \ref{mSORTEworkonlinfty}.

\begin{theorem}
\label{mSORTEthmminimaxlinfty} Replacing $M^\Phi$ with $(L^\infty({\mathbb{P}%
}))^N$ and $K_\Phi$ with $(L^1({\mathbb{P}}))^N$ in the statement of Theorem %
\ref{mSORTEthmminimax}, all the claims in it remain valid.
\end{theorem}

\begin{proof}
As in Theorem \ref{mSORTEthmminimax}, the proof can be obtained with minor
and obvious modifications of the one in \cite{BDFFM} Theorem A.3, using
Theorem 4 of \cite{Rockafellar} in place of Corollary on page 534 of \cite%
{Rockafellar}.
\end{proof}

\section{Proof of Theorem \protect\ref{mSORTEthmmaingeneral1}}

\label{mSORTESecProof} We now take care of the proof of Theorem \ref%
{mSORTEthmmaingeneral1} for the case $A=0$. For the reader's convenience, we
first provide a more detailed statement for all the results we prove, which
also provides a road map for the proof. It is easy to reconstruct the
content of Theorem \ref{mSORTEthmmaingeneral1} from the statement below.

\begin{theorem}
\label{mSORTEthmadditionalreq} $\,$

Under Assumption \ref{mSORTEA1} the following hold:

\begin{itemize}
\item[1.] for every $X\in M^\Phi$ {\small 
\begin{align}
&\sup_{Y\in \mathcal{B}_{0}\cap M^{\Phi }}\mathbb{E}_{\mathbb{P}}\left[
U(X+Y)\right] =\sup\left\{\mathbb{E}_{\mathbb{P}}\left[ U(X+Y)\right]\mid {%
Y\in \mathcal{L}},\,\sum_{j=1}^N\mathbb{E}_{\mathbb{Q}^j} \left[Y^j\right]%
\leq 0\,\,\forall{\mathbb{Q}}\in\mathcal{Q}_{\mathcal{B},V}\right\}
\label{mSORTEeqminimaxappliedcor2} \\
&=\min_{\mathbb{Q}\in \mathcal{Q}_{\mathcal{B},V}}\min_{\lambda\geq 0}\left( \lambda
\sum_{j=1}^{N}\mathbb{E}_{\mathbb{Q}^{j}}\left[ X^{j}\right] +\mathbb{E}_{%
\mathbb{P}}\left[ V\left( \lambda \frac{\mathrm{d}{\mathbb{Q}}}{\mathrm{d}{%
\mathbb{P}}}\right) \right] \right) \,.  \label{mSORTEeqminimaxappliedcor3}
\end{align}%
}Every optimum $(\widehat{\lambda },\widehat{{\mathbb{Q}}})$ of %
\eqref{mSORTEeqminimaxappliedcor3} satisfies $\widehat{\lambda}>0$ and $%
\widehat{{\mathbb{Q}}}\sim {\mathbb{P}}$. Moreover, if $U$ is
differentiable, \eqref{mSORTEeqminimaxappliedcor3} admits a unique optimum $(%
\widehat{\lambda },\widehat{{\mathbb{Q}}})$, with $\widehat{{\mathbb{Q}}}%
\sim {\mathbb{P}}$.
\end{itemize}

Furthermore, there exists a random vector $\widehat{Y}\in (L^1({\mathbb{P}}%
))^N$ such that:

\begin{itemize}
\item[2.] $\widehat{Y}$ satisfies: 
\begin{equation}
\widehat{Y}^{j}\frac{\mathrm{d}{\mathbb{Q}}^{j}}{\mathrm{d}{\mathbb{P}}}\in
L^{1}({\mathbb{P}})\,\,\,\forall {\mathbb{Q}}\in \mathcal{Q}_{\mathcal{B},V},\forall
\,j=1,\dots ,N\,\text{ and }\sum_{j=1}^{N}\mathbb{E}_{\mathbb{Q}^{j}}\left[ 
\widehat{Y}^{j}\right] \leq 0\,\,\,\forall {\mathbb{Q}}\in \mathcal{Q}_{\mathcal{B},V}\,.
\label{mSORTEeqinandexpwidehat}
\end{equation}
and for any optimizer $(\widehat{\lambda},\widehat{{\mathbb{Q}}})\in{\mathbb{%
R}}_+\times \mathcal{Q}_{\mathcal{B},V}$ of \eqref{mSORTEeqminimaxappliedcor3} 
\begin{equation}
\sum_{j=1}^{N}\mathbb{E}_{\widehat{\mathbb{Q}}^{j}}\left[ \widehat{Y}^{j}%
\right] =0=\sum_{j=1}^{N}\widehat{Y}^{j}\,.
\label{mSORTEpropertieswidehatwidehatq2}
\end{equation}
\end{itemize}

\item[3.] $\widehat{Y}$ is the unique optimum to the following extended
maximization problem: 
\begin{equation}  \label{mSORTEeqoptimizationadditional}
\sup\left\{\mathbb{E}_{\mathbb{P}}\left[ U(X+Y)\right]\mid {Y\in \mathcal{L}}%
,\,\sum_{j=1}^N\mathbb{E}_{\mathbb{Q}^j} \left[Y^j\right]\leq 0\,\,\forall{%
\mathbb{Q}}\in\mathcal{Q}_{\mathcal{B},V}\right\} =\mathbb{E}_\mathbb{P} \left[U(X+%
\widehat{Y})\right]\,.
\end{equation}
\end{theorem}

\begin{proof}

$\,$

\textbf{STEP 1}: \textit{we show the equality chain in %
\eqref{mSORTEeqminimaxappliedcor2} and \eqref{mSORTEeqminimaxappliedcor3}.}

We introduce for $A\in\R$
\begin{equation}
\pi_A(X):=\sup \left\{ \mathbb{E}_\mathbb{P} \left[U(X+Y)\right]\mid Y\in 
\mathcal{B}\cap M^\Phi,\,\sum_{j=1}^NY^j\leq A\right\} \,.  \label{mSORTEpiA}
\end{equation}
We start recognizing $\pi_0(X)$ as
the LHS of \eqref{mSORTEeqminimaxappliedcor2} and observing that $%
-\infty<\pi_0(X)$ since $\mathcal{B}_A\cap M^\Phi\neq\emptyset$ and by
Fenchel inequality $\pi_0(X)<+\infty$ (combining Remark \ref%
{mSORTEremthereexistsanequiv} to guarantee $\mathcal{Q}_{\mathcal{B},V}\neq \emptyset$,
Proposition \ref{mSORTEpropfairprob} and the inequality chain in Remark \ref%
{mSORTER1}). Again by Proposition \ref{mSORTEpropfairprob} and Remark \ref%
{mSORTER1}, it is enough to show that 
\begin{equation*}
\sup_{Y\in \mathcal{B}_{0}\cap M^{\Phi }}\mathbb{E}_{\mathbb{P}}\left[ U(X+Y)%
\right] =\min_{\mathbb{Q}\in \mathcal{Q}_{\mathcal{B},V}}\min_{\lambda\geq 0}\left(
\lambda \sum_{j=1}^{N}\mathbb{E}_{\mathbb{Q}^{j}}\left[ X^{j}\right] +%
\mathbb{E}_{\mathbb{P}}\left[ V\left( \lambda \frac{\mathrm{d}{\mathbb{Q}}}{%
\mathrm{d}{\mathbb{P}}}\right) \right] \right)\,.
\end{equation*}

This equality follows by Theorem \ref{mSORTEthmminimax}, taking $\mathcal{C}%
:=\mathcal{B}_{0}\cap M^{\Phi }$ and noticing that minima over $\mathcal{Q}$
can be substituted with minima over $\mathcal{Q}_{\mathcal{B},V}$ since the expression
in LHS of \eqref{mSORTEeqminimaxappliedcor2} is finite. Observe at this
point that, if any of the three expressions in %
\eqref{mSORTEeqminimaxappliedcor2} and \eqref{mSORTEeqminimaxappliedcor3}
were strictly smaller than $V(0)=\sup_{x\in{\mathbb{R}}^N}U(x)$, direct
substitution of $\lambda=0$ in the expression would give a contradiction, no
matter what the optimal probability measure is.

\textbf{STEP 2}: \textit{we show the existence of a vector $\widehat{Y}$ as
described in Items 2 and 3 of the statement, made exception for the first
equality in \eqref{mSORTEpropertieswidehatwidehatq2}. More precisely, we
first (\underline{Step 2a}) identify a natural candidate $\widehat{Y}$ using
a maximizing sequence, and we show that it satisfies $\sum_{j=1}^N\widehat{Y}%
^j=0$. Then (\underline{Step 2b}) we show that such candidate satisfies the
integrability conditions and inequalities in \eqref{mSORTEeqinandexpwidehat}%
. Finally (\underline{Step 2c}) we show optimality stated in Item 3. The
proof of the first equality in \eqref{mSORTEpropertieswidehatwidehatq2} is
postponed to STEP 5.}

\underline{Step 2a}. First observe that $X+Y\geq -\left( \left\vert
X\right\vert +\left\vert Y\right\vert \right) $ in the componentwise order,
hence for $Z\in \mathcal{B}_0\cap M^{\Phi }\neq \emptyset $, as $X,Z\in
M^{\Phi }$, we get 
\begin{equation*}
\sup_{Y\in \mathcal{B}_0\cap M^{\Phi }}\mathbb{E}_{\mathbb{P}}\left[ U(X+Y)%
\right] \geq \mathbb{E}_{\mathbb{P}}\left[ U(X+Z)\right] \geq \mathbb{E}_{%
\mathbb{P}}\left[ U\left( -(\left\vert X\right\vert +\left\vert Z\right\vert
)\right) \right] >-\infty\,.
\end{equation*}%
Take now a maximizing sequence $(Y_{n})_{n}$ in $\mathcal{B}_0\cap M^{\Phi }$%
. W.l.o.g. we can assume that 
\begin{equation}
\sum_{j=1}^N Y^j_n=0\,\,\,\,\forall n\,\,{\mathbb{P}}-a.s.  \notag
\end{equation}
since if this were not the case (i.e. if the inequality were strict) we
could add a $\varepsilon>0$ small enough to each component without
decreasing the utility of the system or violating the constraint $Y_n\in%
\mathcal{B}_0$.

Observe that 
\begin{equation*}
\sup_{n}\left\vert \sum_{j=1}^{N}\mathbb{E}_{\mathbb{P}}\left[
X^{j}+Y_{n}^{j}\right] \right\vert \leq \left\vert \sum_{j=1}^{N}\mathbb{E}_{%
\mathbb{P}}\left[ X^{j}\right] \right\vert +\left\vert 0\right\vert <+\infty
\end{equation*}%
and $\mathbb{E}_{\mathbb{P}}\left[ U(X+Y_{n})\right] \geq \mathbb{E}_{%
\mathbb{P}}\left[ U(X+Y_{1})\right] =:B\in {\mathbb{R}}$. Then Lemma \ref%
{mSORTElemmakomlos} Item 1 applies with $Z_{n}:=X+Y_{n}$. Using also $%
\left\vert X^{j}\right\vert +\left\vert Y_{n}^{j}\right\vert \leq \left\vert
X^{j}+Y_{n}^{j}\right\vert +2\left\vert X^{j}\right\vert ,\,j=1,\dots ,N$ we
get 
\begin{equation*}
\sup_{n}\sum_{j=1}^{N}\mathbb{E}_{\mathbb{P}}\left[ \left\vert
X^{j}\right\vert +\left\vert Y_{n}^{j}\right\vert \right] <\infty \,.
\end{equation*}%
Now we apply Corollary \ref{mSORTEcorkomplosmultidim} with ${\mathbb{P}}%
_{1}=\dots ={\mathbb{P}}_{N}={\mathbb{P}}$ and extract the subsequence $%
(Y_{n_{h}})_{h}$ such that for some $\widehat{Y}\in (L^{1}({\mathbb{P}}%
))^{N} $ 
\begin{equation}
W_{H}:=\frac{1}{H}\sum_{h=1}^{H}Y_{n_{h}}\xrightarrow[H\rightarrow+\infty]{}%
\widehat{Y}\quad {\mathbb{P}}-\text{a.s.\quad }\text{and\quad }%
\sup_{H}\sum_{j=1}^{N}\mathbb{E}_{\mathbb{P}}\left[ \left\vert
W_{H}^{j}\right\vert \right] <+\infty \,.  \label{mSORTEdefWY}
\end{equation}%
We notice that by convexity the random vectors $W_{H}$ still belong to $%
\mathcal{B}_0\cap M^{\Phi }$, and $\widehat{Y}\in \mathcal{B}_0$ as  $%
\mathcal{B}_0$ is closed in probability (since so is $\mathcal{B}$). Moreover, we
have that 
\begin{equation*}
0=\sum_{j=1}^NW^j_{H}:=\frac{1}{H}\sum_{h=1}^{H}\sum_{j=1}^NY^j_{n_{h}}%
\xrightarrow[H\rightarrow+\infty]{}\sum_{j=1}^N\widehat{Y}^j\quad {\mathbb{P}%
}-\text{a.s.\quad }
\end{equation*}
so that the second equality in \eqref{mSORTEpropertieswidehatwidehatq2}
holds.

\underline{Step 2b}. We first work on integrability. We proceed as follows:
we show that for any ${\mathbb{Q}}\in\mathcal{Q}_{\mathcal{B},V}$ we have $\sum_{j=1}^N 
\widehat{Y}^j\frac{\mathrm{d}{\mathbb{Q}}^j}{\mathrm{d}{\mathbb{P}}}\in L^1({%
\mathbb{P}})$. Then we show that $(\widehat{Y})^-\in L^{\widehat{\Phi}}$,
and conclude the integrability conditions in %
\eqref{mSORTEpropertieswidehatwidehatq2}.

Let us begin with showing $\sum_{j=1}^N \widehat{Y}^j\frac{\mathrm{d}{%
\mathbb{Q}}^j}{\mathrm{d}{\mathbb{P}}}\in L^1({\mathbb{P}})$. By definition
of $V(\cdot)$, we have 
\begin{equation}  \label{mSORTEfenchelinner}
U(X+\widehat{Y})\leq V(\lambda Z)+\left< X+\widehat{Y},\lambda Z\right>\,\,{%
\mathbb{P}}-\text{a.s. }\forall\,\lambda\geq0,\,Z\in K_\Phi\,.
\end{equation}
This implies 
\begin{equation*}
(U(X+\widehat{Y}))^-\geq \left(V(\lambda Z)+\left< X+\widehat{Y},\lambda
Z\right>\right)^-
\end{equation*}
so that $\left(V(\lambda Z)+\left< X+\widehat{Y}, \lambda
Z\right>\right)^-\in L^1({\mathbb{P}})$.

We prove integrability also for the positive part, assuming now $Z=\frac{%
\mathrm{d}{\mathbb{Q}}}{\mathrm{d}{\mathbb{P}}},\,{\mathbb{Q}}\in\mathcal{Q}_{\mathcal{B},V}$ and taking $\lambda>0$ such that $\mathbb{E}_\mathbb{P} \left[V(\lambda
Z)\right]<+\infty$. By \eqref{mSORTEdefWY} $W_H\rightarrow_H\widehat{Y}\,{%
\mathbb{P}}$-a.s. so that 
\begin{align}
&\mathbb{E}_\mathbb{P} \left[\left(V(\lambda Z)+\left< X+\widehat{Y},
\lambda Z\right>\right)^+\right]=\mathbb{E}_\mathbb{P} \left[%
\liminf_H\left(V(\lambda Z)+\left< X+W_H, \lambda Z\right>\right)^+\right] 
\notag \\
&\leq \liminf_H \mathbb{E}_\mathbb{P} \left[\left(V(\lambda Z)+\left< X+W_H,
\lambda Z\right>\right)^+\right]  \notag \\
& \leq \sup_H\left(\mathbb{E}_\mathbb{P} \left[V(\lambda Z)+\left< X+W_H,
\lambda Z\right>\right]\right)+\sup_H\left(\mathbb{E}_\mathbb{P} \left[%
\left(V(\lambda Z)+\left< X+W_H, \lambda Z\right>\right)^-\right]\right)\,.
\label{mSORTEno10}
\end{align}
Now since $\mathbb{E}_\mathbb{P} \left[\left< W_H, \lambda Z\right>\right]%
\leq 0$ 
\begin{equation}  \label{mSORTEno11}
\sup_H\left(\mathbb{E}_\mathbb{P} \left[V(\lambda Z)+\left< X+W_H, \lambda
Z\right>\right]\right)\leq \mathbb{E}_\mathbb{P} \left[V(\lambda Z)+\left<
X, \lambda Z\right>\right]<+\infty\,.
\end{equation}
Also by \eqref{mSORTEfenchelinner} 
\begin{align}
&\sup_H\left(\mathbb{E}_\mathbb{P} \left[\left(V(\lambda Z)+\left< X+W_H,
\lambda Z\right>\right)^-\right]\right)\leq \sup_H\left(\mathbb{E}_\mathbb{P}
\left[(U(X+W_H))^-\right]\right)  \notag \\
&\leq \sup_H\left(\mathbb{E}_\mathbb{P} \left[(U(X+W_H))^+-U(X+W_H)\right]%
\right)  \notag \\
&\leq\sup_H\left(\mathbb{E}_\mathbb{P} \left[(U(X+W_H))^+\right]%
\right)-\inf_H\left(\mathbb{E}_\mathbb{P} \left[U(X+W_H)\right]\right) .
\label{mSORTEeqintermediate}
\end{align}

Now use \eqref{mSORTEcontrolwithepsilon}: 
\begin{equation*}
\sup_H\left(\mathbb{E}_\mathbb{P} \left[(U(X+W_H))^+\right]\right)\leq
\sup_H\varepsilon\left(\sum_{j=1}^N\mathbb{E}_\mathbb{P} \left[(X^j+W_H^j))^+%
\right]\right)+b_\varepsilon\,.
\end{equation*}
We also have, $Y_1$ being the first element in the maximizing sequence, that 
$\inf_H\left(\mathbb{E}_\mathbb{P} \left[U(X+W_H)\right]\right)\geq\mathbb{E}%
_\mathbb{P} \left[U(X+Y_1)\right]$ by construction. Thus, continuing from %
\eqref{mSORTEeqintermediate}, we get 
\begin{align}
&\sup_H\left(\mathbb{E}_\mathbb{P} \left[\left(V(\lambda Z)+\left< X+W_H,
\lambda Z\right>\right)^-\right]\right)  \notag \\
&\leq \varepsilon\sup_H\left(\sum_{j=1}^N\mathbb{E}_\mathbb{P} \left[%
(X^j+W_H^j))^+\right]\right)+b_\varepsilon-\mathbb{E}_\mathbb{P} \left[%
U(X+Y_1)\right]<+\infty  \label{mSORTEno12}
\end{align}
since the sequence $W_H$ is bounded in $(L^1({\mathbb{P}}))^N$ (see %
\eqref{mSORTEdefWY}) and $\mathbb{E}_\mathbb{P} \left[U(X+Y_1)\right]%
>-\infty $.

From \eqref{mSORTEno10}, \eqref{mSORTEno11}, \eqref{mSORTEno12} we conclude
that

\begin{equation*}
\mathbb{E}_{\mathbb{P}}\left[ \left( V(\lambda Z)+\left< X+\widehat{Y}%
,\lambda Z\right> \right) ^{+}\right] <+\infty \,.
\end{equation*}%
To sum up, for $Z\in \mathcal{Q}_{\mathcal{B},V}$ and $\lambda $ s.t. $\mathbb{E}_{%
\mathbb{P}}\left[ V(\lambda Z)\right] <+\infty $ 
\begin{equation*}
\left< X,\lambda Z\right> ,\,V(\lambda Z),\,\left( V(\lambda Z)+\left< X+%
\widehat{Y},\lambda Z\right> \right) ^{+ },\,\left( V(\lambda Z)+\left< X+%
\widehat{Y},\lambda Z\right> \right) ^{- }\in L^{1}({\mathbb{P}})
\end{equation*}%
which gives $\left< \widehat{Y},Z\right> \in L^{1}({\mathbb{P}}),\forall
Z\in \mathcal{Q}_{\mathcal{B},V}$. Hence we have 
\begin{equation}  \label{mSORTEintegrabofsum}
\sum_{j=1}^N \widehat{Y}^j\frac{\mathrm{d}{\mathbb{Q}}^j}{\mathrm{d}{\mathbb{%
P}}}\in L^1({\mathbb{P}})\,\,\,\,\,\forall\,{\mathbb{Q}}\in\mathcal{Q}_{\mathcal{B},V}\,.
\end{equation}

Next, we prove that $(\widehat{Y})^-\in L^{\widehat{\Phi}}$. To see this, we
observe from \eqref{mSORTEcontrolwithphihat} that for $\varepsilon>0$
sufficiently small 
\begin{equation*}
\mathbb{E}_\mathbb{P} \left[U(X+W_H)\right]\leq -\mathbb{E}_\mathbb{P} \left[%
\widehat{\Phi}\left((X+W_H)^-\right)\right]+\varepsilon\sum_{j=1}^N\mathbb{E}%
_\mathbb{P} \left[\left|W_H^j\right|\right]+f(\varepsilon)
\end{equation*}
which readily gives 
\begin{equation}  \label{mSORTEnormbddnegparts}
\sup_H\mathbb{E}_\mathbb{P} \left[\widehat{\Phi}\left((X+W_H)^-\right)\right]%
\leq \varepsilon\sup_H\left(\sum_{j=1}^N\mathbb{E}_\mathbb{P} \left[%
\left|W_H^j\right|\right]\right)-\inf_H\mathbb{E}_\mathbb{P} \left[U(X+W_H)%
\right]+f(\varepsilon)<+\infty\,.
\end{equation}
By Fatou Lemma we than conclude 
\begin{equation*}
\mathbb{E}_\mathbb{P} \left[\widehat{\Phi}\left((X+\widehat{Y})^-\right)%
\right]\leq \varepsilon\sup_H\left(\sum_{j=1}^N\mathbb{E}_\mathbb{P} \left[%
\left|W_H^j\right|\right]\right)-\inf_H\mathbb{E}_\mathbb{P} \left[U(X+W_H)%
\right]+f(\varepsilon)<+\infty\,.
\end{equation*}
From this we infer $(X+\widehat{Y})^-\in L^{\widehat{\Phi}}$.

We are now almost done showing integrability. Since by %
\eqref{mSORTEcontrolwithphihat11} $L^{\widehat{\Phi}}=L^\Phi$, we conclude
that $(X+\widehat{Y})^-\in L^{\Phi}$. By the very definition of $K_\Phi$ and 
$\mathcal{Q}_{\mathcal{B},V}\subseteq K_\Phi$, we have $\sum_{j=1}^N(X^j+\widehat{Y}^j)^-%
\frac{\mathrm{d}{\mathbb{Q}}^j}{\mathrm{d}{\mathbb{P}}}\in L^1({\mathbb{P}})$%
, which clearly implies $(X^j+\widehat{Y}^j)^-\frac{\mathrm{d}{\mathbb{Q}}^j%
}{\mathrm{d}{\mathbb{P}}}\in L^1({\mathbb{P}})$ for every $j=1,\dots,N$. At
the same time we observe that 
\begin{equation*}
\sum_{j=1}^N(X^j+\widehat{Y}^j)^+\frac{\mathrm{d}{\mathbb{Q}}^j}{\mathrm{d}{%
\mathbb{P}}}=\sum_{j=1}^NX^j\frac{\mathrm{d}{\mathbb{Q}}^j}{\mathrm{d}{%
\mathbb{P}}}\in L^1({\mathbb{P}})+\sum_{j=1}^N\widehat{Y}^j\frac{\mathrm{d}{%
\mathbb{Q}}^j}{\mathrm{d}{\mathbb{P}}}\in L^1({\mathbb{P}})+\sum_{j=1}^N(X^j+%
\widehat{Y}^j)^-\frac{\mathrm{d}{\mathbb{Q}}^j}{\mathrm{d}{\mathbb{P}}}
\end{equation*}
and all the terms in RHS are in $L^1({\mathbb{P}})$ (recall %
\eqref{mSORTEintegrabofsum}). Thus, $\sum_{j=1}^N(X^j+\widehat{Y}^j)^+\frac{%
\mathrm{d}{\mathbb{Q}}^j}{\mathrm{d}{\mathbb{P}}}\in L^1({\mathbb{P}})$ and $%
(X^j+\widehat{Y}^j)^+\frac{\mathrm{d}{\mathbb{Q}}^j}{\mathrm{d}{\mathbb{P}}}%
\in L^1({\mathbb{P}})$ for every $j=1,\dots,N$. We finally get that $%
\widehat{Y}\in\bigcap_{{\mathbb{Q}}\in\mathcal{Q}_{\mathcal{B},V}}L^1({\mathbb{Q}})$ and
our integrability conditions in \eqref{mSORTEeqinandexpwidehat} are now
proved.

To conclude Step 2b, we need to show that $\sum_{j=1}^{N}\mathbb{E}_{\mathbb{%
Q}^{j}}\left[ \widehat{Y}^{j}\right] \leq 0\,\forall {\mathbb{Q}}\in 
\mathcal{Q}_{\mathcal{B},V}\,.$ This is a consequence of Lemma \ref{lemmafairnessgeneral} with $A=0$, since $\widehat{Y}\in\mathcal{B}_0\cap \mathcal{L}$.

\underline{Step 2c.} Observe now that 
\begin{equation}
\mathbb{E}_{\mathbb{P}}\left[ U(X+W_{H})\right] \geq \frac{1}{H}%
\sum_{h=1}^{H}\mathbb{E}_{\mathbb{P}}\left[ U(X+Y_{n_{h}})\right] %
\xrightarrow[H\rightarrow+\infty]{}\sup_{Y\in \mathcal{B}_0\cap M^{\Phi }}%
\mathbb{E}_{\mathbb{P}}\left[ U(X+Y)\right]  \label{mSORTEeqlowerboundexpW}
\end{equation}%
by concavity of $U$ and the fact that $(Y_{n_{h}})_{h}$ is again a
maximizing sequence. From the expression in Equation %
\eqref{mSORTEeqlowerboundexpW} we get that for every $\varepsilon >0$,
definitely (in $H$) 
\begin{equation*}
\mathbb{E}_{\mathbb{P}}\left[ U(X+W_{H})\right] \geq \sup_{Y\in \mathcal{B}%
_0\cap M^{\Phi }}\mathbb{E}_{\mathbb{P}}\left[ U(X+Y)\right] -\varepsilon
\end{equation*}

Apply now Lemma \ref{mSORTElemmakomlos} Item \ref{mSORTElemmabdd2} for $%
B=\sup_{Y\in \mathcal{B}_0\cap M^{\Phi }}\mathbb{E}_{\mathbb{P}}\left[ U(X+Y)%
\right] -\varepsilon $ to the sequence $(X+W_{H})_{H}$ for $H$ big enough
(this sequence is bounded in $(L^{1}({\mathbb{P}}))^{N}$ by %
\eqref{mSORTEdefWY}) to get that for every $\varepsilon >0$ 
\begin{equation*}
\mathbb{E}_{\mathbb{P}}\left[ U(X+\widehat{Y})\right] \geq \sup_{Y\in 
\mathcal{B}_0\cap M^{\Phi }}\mathbb{E}_{\mathbb{P}}\left[ U(X+Y)\right]
-\varepsilon
\end{equation*}%
Clearly then $\widehat{Y}$ satisfies 
\begin{equation}  \label{mSORTEitsmorethanthesup}
\mathbb{E}_{\mathbb{P}}\left[ U(X+\widehat{Y})\right] \geq \sup_{Y\in 
\mathcal{B}_0\cap M^{\Phi }}\mathbb{E}_{\mathbb{P}}\left[ U(X+Y)\right] .
\end{equation}%
Now recall that by \eqref{mSORTElemmacontrolwithline} for some $a>0, b\in{%
\mathbb{R}}$ 
\begin{equation*}
U(X+\widehat{Y})\leq a\sum_{j=1}^{N}(X^{j}+\widehat{Y}^{j})+a%
\sum_{j=1}^{N}(-(X^{j}+\widehat{Y}^{j})^{-})+b
\end{equation*}%
and since RHS is in $L^{1}({\mathbb{P}})$ we conclude that $\mathbb{E}_{%
\mathbb{P}}\left[ U(X+\widehat{Y})\right] <+\infty $. Hence: 
\begin{align}
&\mathbb{E}_\mathbb{P} \left[U(X+\widehat{Y})\right]\overset{\text{Eq.}%
\eqref{mSORTEitsmorethanthesup}}{\geq }\sup_{Y\in \mathcal{B}_0\cap M^{\Phi
}}\mathbb{E}_{\mathbb{P}}\left[ U(X+Y)\right]  \notag \\
&\overset{\text{Eq.}\eqref{mSORTEeqminimaxappliedcor2}}{=}\sup\left\{\mathbb{%
E}_{\mathbb{P}}\left[ U(X+Y)\right]\mid {Y\in \mathcal{L}},\,\sum_{j=1}^N%
\mathbb{E}_{\mathbb{Q}^j} \left[Y^j\right]\leq 0\,\,\forall{\mathbb{Q}}\in%
\mathcal{Q}_{\mathcal{B},V}\right\}\,.  \notag
\end{align}
It is now enough to recall that by \eqref{mSORTEeqinandexpwidehat}, which we
proved in Step 2b, $\widehat{Y}$ satisfies the constraints in RHS of %
\eqref{mSORTEeqoptimizationadditional}: ${Y\in \mathcal{L}},\,\sum_{j=1}^N%
\mathbb{E}_{\mathbb{Q}^j} \left[Y^j\right]\leq 0\,\,\forall{\mathbb{Q}}\in%
\mathcal{Q}_{\mathcal{B},V}$. Optimality claimed in Item 3 then follows.

\textbf{STEP 3:} \textit{we prove uniqueness of the optimum for the
maximization problem in Item 3 and condition $\widehat{\lambda}>0$ for every
optimum $(\widehat{\lambda},\widehat{{\mathbb{Q}}})$ of %
\eqref{mSORTEeqminimaxappliedcor3}.}

The uniqueness for the optimum follows from strict concavity of $U$ (see
Standing Assumption I): if two distinct optima existed, any strict convex
combination of the two would still satisfy the constraint and would produce
a value for $\mathbb{E}_{\mathbb{P}}\left[ U(X+\bullet )\right] $ strictly
greater than the supremum.

Recall now from STEP 1 that to prove the claimed $\widehat{\lambda}>0$ it is
enough to show that any of the three expressions in %
\eqref{mSORTEeqminimaxappliedcor2} and \eqref{mSORTEeqminimaxappliedcor3} is
strictly smaller than $\sup_{x\in{\mathbb{R}}^N}U8x)$. Property $\widehat{%
\lambda}>0$ is easily obtained if $\sup_{z\in {\mathbb{R}}^{N}}U(z)=+\infty $%
, since we proved that $\mathbb{E}_\mathbb{P} \left[U(X+\widehat{Y})\right]%
<+\infty$ in Step 2c. Suppose that instead $\sup_{z\in {\mathbb{R}}%
^{N}}U(z)<+\infty $ and notice that, setting $\mathcal{K}:=\{Y\in \mathcal{L}%
,\,\sum_{j=1}^N\mathbb{E}_{\mathbb{Q}^j} \left[Y^j\right]\leq 0\,\,\forall{%
\mathbb{Q}}\in\mathcal{Q}_{\mathcal{B},V}\}$ 
\begin{equation*}
\sup_{Y\in \mathcal{K}}\mathbb{E}_{\mathbb{P}}\left[ U(X+Y)\right] \leq 
\mathbb{E}_{\mathbb{P}}\left[ U(X+\widehat{Y})\right] \leq \sup_{z\in {%
\mathbb{R}}^{N}}U(z).
\end{equation*}
If we had $\sup_{Y\in \mathcal{K}}\mathbb{E}_{\mathbb{P}}\left[ U(X+Y)\right]
=\sup_{z\in {\mathbb{R}}^{N}}U(z)$, then we would also have%
\begin{equation*}
\sup_{z\in {\mathbb{R}}^{N}}U(z)=\mathbb{E}_{\mathbb{P}}\left[ U(X+\widehat{Y%
})\right]
\end{equation*}%
so that: 
\begin{equation*}
0=\mathbb{E}_{\mathbb{P}}\left[ \sup_{z\in {\mathbb{R}}^{N}}U(z)-U(X+%
\widehat{Y})\right] =\mathbb{E}_{\mathbb{P}}\left[ \left\vert \sup_{z\in {%
\mathbb{R}}^{N}}U(z)-U(X+\widehat{Y})\right\vert \right] ,
\end{equation*}%
which implies $\sup_{z\in {\mathbb{R}}^{N}}U(z)=U(X+\widehat{Y})$ ${\mathbb{P%
}}$-a.s.. In particular, from the fact that $X+\widehat{Y}$ is finite almost
surely, it would follow that $U$ almost surely attains its supremum on some
compact subset of ${\mathbb{R}}^{N}$, which is clearly a contradiction given
that $U$ is strictly increasing.

\textbf{STEP 4:} \textit{we study a related optimization problem when a ${%
\mathbb{Q}}\in\mathcal{Q}_{\mathcal{B},V}$ is fixed}.

We show that for any fixed $X\in M^\Phi$ and ${\mathbb{Q}}\in\mathcal{Q}_{\mathcal{B},V}$
it holds that 
\begin{align}
+\infty>\pi_0^{\mathbb{Q}}(X):=&\sup\left\{\mathbb{E}_\mathbb{P} \left[U(X+Y)%
\right]\mid Y\in M^\Phi, \sum_{j=1}^N\mathbb{E}_{\mathbb{Q}^j} \left[Y^j%
\right]\leq 0\right\}  \label{mSORTEeqqfixedsup1} \\
=&\min_{\lambda \geq 0}\left( \lambda \sum_{j=1}^{N}\mathbb{E}_{\mathbb{Q}%
^{j}}\left[ X^{j}\right] +\mathbb{E}_{\mathbb{P}}\left[ V\left( \lambda 
\frac{\mathrm{d}{\mathbb{Q}}}{\mathrm{d}{\mathbb{P}}}\right) \right] \right)
\,.  \label{mSORTEeqqfixedsup12}
\end{align}

$\pi _{0}^{\mathbb{Q}}(X)<+\infty $ follows from Remark \ref{mSORTER1}. The
equality between \eqref{mSORTEeqqfixedsup1} and \eqref{mSORTEeqqfixedsup12}
follows from Theorem \ref{mSORTEthmminimax}, and the fact that 
\begin{equation*}
\mathcal{C}:=\left\{ Y\in M^{\Phi },\allowbreak\sum_{j=1}^{N}\mathbb{E}_{%
\mathbb{Q}^{j}}\left[ Y^{j}\right] \leq 0\right\} \Longrightarrow(\mathcal{C}%
_{1}^{0})^{+}=\left\{ \frac{\mathrm{d}{\mathbb{Q}}}{\mathrm{d}{\mathbb{P}}}%
\right\} \subseteq K_{\Phi } \text{ as }{\mathbb{Q}}\in \mathcal{Q}_{\mathcal{B},V}.
\end{equation*}

We stress the fact that \eqref{mSORTEeqqfixedsup1} and %
\eqref{mSORTEeqqfixedsup12} hold also dropping Assumption \ref{mSORTEA1}.

We observe that if \eqref{mSORTEeqqfixedsup1} is strictly smaller than $V(0)$
then the minimum in \eqref{mSORTEeqqfixedsup12} can be taken over $%
(0,+\infty)$ in place of $[0,+\infty)$.

Let now $X\in M^\Phi$ be fixed and $\pi_0(\cdot),\, \pi_0^{\mathbb{Q}%
}(\cdot) $ be as in \eqref{mSORTEpiA}, \eqref{mSORTEeqqfixedsup1}
respectively. Then from STEP 1 together with \eqref{mSORTEeqqfixedsup1} and %
\eqref{mSORTEeqqfixedsup12} 
\begin{equation}  \label{mSORTEeqminimaxforpis}
\pi_0(X)=\min_{{\mathbb{Q}}\in\mathcal{Q}_{\mathcal{B},V}}\left(\pi_0^{\mathbb{Q}%
}(X)\right)\,
\end{equation}
and whenever $(\widehat{\lambda,}\widehat{{\mathbb{Q}}})$ is an optimum for %
\eqref{mSORTEeqminimaxappliedcor3}, then $\widehat{{\mathbb{Q}}}$ is an
optimum for \eqref{mSORTEeqminimaxforpis}.

\textbf{STEP 5:} \textit{we show that for any optimum $(\widehat{\lambda},%
\widehat{{\mathbb{Q}}})\in{\mathbb{R}}_+\times\mathcal{Q}_{\mathcal{B},V}$ of %
\eqref{mSORTEeqminimaxappliedcor3} we have $\widehat{{\mathbb{Q}}}\sim{\mathbb{P}}$ and the first equality in %
\eqref{mSORTEpropertieswidehatwidehatq2} holds.}

We start observing that for any optimal $\widehat{{\mathbb{Q}}}$ as in the
claim we have that by STEP 5 
\begin{equation}  \label{mSORTEhatYoptforpiq}
\mathbb{E}_\mathbb{P} \left[U(X+\widehat{Y})\right]\overset{\text{Eq.}%
\eqref{mSORTEeqminimaxforpis}}{=}\pi^{\widehat{{\mathbb{Q}}}%
}_0(X)=\sup\left\{\mathbb{E}_\mathbb{P} \left[U(X+Y)\right]\mid Y\in 
\mathcal{L}, \sum_{j=1}^N\mathbb{E}_{\widehat{{\mathbb{Q}}}^j}[Y^j]\leq
0\right\}
\end{equation}

The last equality in particular follows observing that by trivial set
inclusions and Fenchel inequality 
\begin{align*}
\pi_0^{\mathbb{Q}}(X)\leq &\sup\left\{\mathbb{E}_\mathbb{P} \left[U(X+Y)%
\right]\mid Y\in \mathcal{L}, \sum_{j=1}^N\mathbb{E}_{\widehat{\mathbb{Q}}%
^j} \left[Y^j\right]\leq 0\right\} \\
\leq&\inf_{\lambda \geq 0}\left( \lambda \sum_{j=1}^{N}\mathbb{E}_{\mathbb{Q}%
^{j}}\left[ X^{j}\right] +\mathbb{E}_{\mathbb{P}}\left[ V\left( \lambda 
\frac{\mathrm{d}{\mathbb{Q}}}{\mathrm{d}{\mathbb{P}}}\right) \right] \right)%
\overset{\text{Eq.}\eqref{mSORTEeqqfixedsup12}}{=}\pi_0^{\mathbb{Q}}(X) \,.
\end{align*}

We now prove that $\widehat{{\mathbb{Q}}}\sim {\mathbb{P}}$, using an
argument inspired by \cite{FollmerSchied2} Remark 3.32: if this were not the
case then ${\mathbb{P}}(A_{k})>0,$ where $A_{k}:=\{\frac{\mathrm{d}\widehat{{%
\mathbb{Q}}}^{k}}{\mathrm{d}{\mathbb{P}}}=0\},$ for some component $k\in
\{1,\dots ,N\}$. Then the vector $\widetilde{Y}$ defined by $\widetilde{Y}%
^{k}:=\widehat{Y}^{k}+1_{A_{k}}$, $\widetilde{Y}^{j}:=\widehat{Y}^{j},j\neq
k $ would still satisfy $\sum_{j=1}^{N}\mathbb{E}_{\widehat{{\mathbb{Q}}}^j}[%
\widetilde{Y}^{j}] \leq 0$

and by monotonicity $\mathbb{E}_{\mathbb{P}}\left[ U(X+\widetilde{Y})\right]
\geq \mathbb{E}_{\mathbb{P}}\left[ U(X+\widehat{Y})\right] $. Thus $%
\widetilde{Y}$ would be another optimum, different from $\widehat{Y}$. Since
now $U$ is strictly concave and $\{ Y\in \mathcal{L}, \sum_{j=1}^N\mathbb{E}%
_{\widehat{{\mathbb{Q}}}^j}[Y^j]\leq 0\}$ is convex, we would get a
contradiction arguing in the same way as in STEP 3.

We move to $\sum_{j=1}^{N}\mathbb{E}_{\widehat{\mathbb{Q}}^{j}}\left[ 
\widehat{Y}^{j}\right] =0$: if this were not the case, by %
\eqref{mSORTEeqinandexpwidehat} we would have $\sum_{j=1}^{N}\mathbb{E}_{%
\widehat{\mathbb{Q}}^{j}}\left[ \widehat{Y}^{j}\right] <0$ so that adding $%
0<\varepsilon $ sufficiently small to each component of $\widehat{Y}$ would
give a vector still satisfying the constraints of RHS of %
\eqref{mSORTEhatYoptforpiq}, but having a corresponding expected utility
strictly grater that the supremum ($U$ is strictly increasing, and $\widehat{%
Y}$ is an optimum as showed by \eqref{mSORTEhatYoptforpiq}), which is a
contradiction.

\textbf{STEP 6}: \textit{we prove uniqueness of the optimum for %
\eqref{mSORTEeqminimaxappliedcor3} under the additional differentiability
assumption}.

Take $\mathcal{C}:=\mathcal{B}_{0}\cap M^{\Phi }$, and observe that %
\eqref{mSORTEeqminimaxappliedcor3} can be rewritten, by %
\eqref{mSORTEpolarityisnice}, as 
\begin{equation*}
\min \left\{ \sum_{j=1}^{N}\mathbb{E}_{\mathbb{P}}\left[ X^{j}Z^{j}\right] +%
\mathbb{E}_{\mathbb{P}}\left[ V(Z)\right] \mid Z\neq 0\in (\mathcal{C}%
_{0}^{1})_{+},\mathbb{E}_{\mathbb{P}}\left[ V(Z)\right] <+\infty \right\}
\end{equation*}%
which from strict convexity of $V(\cdot)$ (\cite{Ro70} Theorem 26.5) admits
a unique optimum $0\leq \widehat{Z}\neq 0$. We then get that, since $%
\widehat{\lambda }=\mathbb{E}_{\mathbb{P}}\left[ \widehat{Z}\right] ,\,\frac{%
\mathrm{d}\widehat{{\mathbb{Q}}}}{\mathrm{d}{\mathbb{P}}}=\frac{\widehat{Z}}{%
\mathbb{E}_{\mathbb{P}}\left[ \widehat{Z}\right] }$ (again by %
\eqref{mSORTEpolarityisnice}), uniqueness for optima in %
\eqref{mSORTEeqminimaxappliedcor3} follows.
\end{proof}

We now consider the possibility of weakening the requirements of \emph{strict%
} monotonicity and concavity of $U$. To do so, we introduce the additional
condition $\mathcal{B}=-\mathcal{B}$. We stress that, for $\mathcal{B}=%
\mathcal{C}_{\mathbb{R}}$ both Assumption \ref{mSORTEA1} and $-\mathcal{B}=%
\mathcal{B}$ hold true.

\begin{corollary}
\label{corollarystrict}
Suppose that the function $U:{\mathbb{R}}^{N}\rightarrow {\mathbb{R}}$ is
(not necessarily strictly) concave and (not necessarily strictly) increasing
with respect to the partial componentwise order. Suppose that there exist a
multivariate Orlicz function $\widehat{\Phi }:{\mathbb{R}}%
_{+}^{N}\rightarrow {\mathbb{R}}$ and a function $f:{\mathbb{R}}%
_{+}\rightarrow {\mathbb{R}}$ such that \eqref{mSORTEcontrolwithphihat}
holds and $L^{\widehat{\Phi }}=L^{\Phi }$. Suppose that Assumption \ref%
{mSORTEA1} holds and that $-\mathcal{B}=\mathcal{B}$. Then the following
statements in Theorem \ref{mSORTEthmadditionalreq} still hold: the
equalities in \eqref{mSORTEeqminimaxappliedcor2}$=$%
\eqref{mSORTEeqminimaxappliedcor3}, the fact that any optimum $\widehat{%
\lambda}$ for \eqref{mSORTEeqminimaxappliedcor3} is strictly positive, all
the claims in Item 2 and optimality (not uniqueness) in Item 3. Moreover,
existence (not uniqueness) for a mSORTE holds.
\end{corollary}

\begin{proof}
Observe first that all the statements in Lemma \ref%
{mSORTElemmaconswellcontrol} still hold, since we did not need strictness of
concavity and monotonicity in their proof. Looking back at the proof of
Theorem \ref{mSORTEthmadditionalreq} we observe that we used strict
convexity and strict monotonicity only from STEP 5 on and for the following
reasons: on the one hand, to prove uniqueness for $\widehat{Y},\widehat{%
\lambda },\widehat{{\mathbb{Q}}}$ and that $\widehat{{\mathbb{Q}}}\sim {%
\mathbb{P}}$, on the other hand to prove that $\sum_{j=1}^{N}\mathbb{E}_{%
\widehat{{\mathbb{Q}}}^{j}}[\widehat{Y}^{j}]=0$. We now show that under the
hypotheses of the Corollary it is indeed possible to show the latter
equality also when $U$ is just increasing and concave, neither of the two
being necessarily strict. More precisely, we show that for any optimum $(%
\widehat{\lambda },\widehat{{\mathbb{Q}}})$ of %
\eqref{mSORTEeqminimaxappliedcor3} we have $\sum_{j=1}^{N}\mathbb{E}_{%
\widehat{{\mathbb{Q}}}^{j}}[\widehat{Y}^{j}]=0$ for the vector $\widehat{Y}$
obtained in STEP 2.

To see this, it is indeed enough to observe that $Y\in\mathcal{B}_0
\Leftrightarrow-Y\in\mathcal{B}_0$, thus the sequence $Y_m$ we used in %
\eqref{mSORTEstep2beq} satisfies $\mathbb{E}_{\mathbb{P}}\left[
\sum_{j=1}^{N}Y_{m}^{j}\frac{\mathrm{d}{\mathbb{Q}}^{j}}{\mathrm{d}{\mathbb{P%
}}}\right]=0$ for any ${\mathbb{Q}}\in\mathcal{Q}_{\mathcal{B},V}$ (recall the definition
of $\mathcal{Q}_{\mathcal{B},V}$ in \eqref{mSORTEQV}). Then, the inequality on %
\eqref{mSORTEDOM} is an equality, and our claim follows. To conclude,
observe that from this point on literally the same arguments of the proof of
Theorem \ref{mSORTEthmadditionalreq} yield the equalities in %
\eqref{mSORTEeqminimaxappliedcor2}$=$\eqref{mSORTEeqminimaxappliedcor3}, the
fact that any optimum $\widehat{\lambda}$ for %
\eqref{mSORTEeqminimaxappliedcor3} is strictly positive, all the claims in
Item 2 and optimality (not uniqueness) in Item 3. Also, Theorem \ref%
{mSORTEthmmaingeneral2} does not in practice require \textit{strict}
concavity or monotonicity in its proof, so it still holds if this assumption
is dropped. Finally, a counterpart to the existence in Theorem \ref{mSORTEthmmsorteexists} can be obtained with the same argument. Take indeed $(\widehat{Y},\widehat{\probq})$ as in Theorem \ref{mSORTEthmmaingeneral1} and adopt the notation from Lemma \ref{mSORTESequalsH}. Then one can check that $S^{\widehat{\probq}}(A)=H^{\widehat{\probq}}(A)$, and that $(\widehat{a}:=\mathbb{E}_{\widehat{\probq}}[\widehat{Y}], \widetilde{Y}_X:=\widehat{Y}-\widehat{a})$ is an optimum for $H^{\widehat{\probq}}(A)$. To see this, it is enough to observe that the inequality $S^{\widehat{\probq}}(A)\geq H^{\widehat{\probq}}(A)$ can be obtained as in Lemma \ref{mSORTESequalsH}, and that $$S^{\widehat{\probq}}(A)\geq H^{\widehat{\probq}}(A)\geq\Ep{U(X+\widehat{a}+\widetilde{Y}_X-\mathbb{E}_{\widehat{\probq}}[\widetilde{Y}_X])}= \Ep{U(X+\widehat{Y})}=S^{\widehat{\probq}}(A)$$ since $\mathbb{E}_{\widehat{\probq}}[\widehat{Y}]=\widehat{a}$ with $\sum_{j=1}^N\widehat{a}^j=A$ and $(\widehat{Y},\widehat{a})$ is optimal for $S^{\widehat{\probq}}(A)$ as in the proof of Theorem \ref{mSORTEthmmsorteexists}.
\end{proof}

\subsection{Replacing Assumption \protect\ref{mSORTEA1}}

\label{mSORTEreplacing} In this Section we present a counterpart to some of
our findings, with Assumption \ref{mSORTEA1} replaced by the following one.

\begin{assumption}
\label{mSORTEA3}The function $\widehat{\Phi }(\cdot )$ in %
\eqref{mSORTEcontrolwithphihat} satisfies $\widehat{\Phi }(x)=\sum_{j=1}^{N}%
\widehat{\Phi }_{j}(x)\,\,\forall x\in {\mathbb{R}}_{+}^{N}$, with $\,%
\widehat{\Phi }_{j}:{\mathbb{R}}\rightarrow {\mathbb{R}}$ differentiable on $%
{\mathbb{R}_{++}}$ and 
\begin{equation*}
\liminf_{z\rightarrow +\infty }\frac{z\widehat{\Phi }_{j}^{\prime }(z)}{%
\widehat{\Phi }_{j}(z)}>1,\,\,\,\,\,\,\,\,\,\,\lim_{z\rightarrow +\infty }%
\frac{\widehat{\Phi }_{j}(z)}{z}=+\infty ,\,\quad \forall j=1,\dots ,N.
\end{equation*}
\end{assumption}

Notice in particular that the condition $\lim_{z\rightarrow +\infty }\frac{%
\widehat{\Phi }_{j}(z)}{z}=+\infty $ guarantees that ${\widehat{\Phi }}%
_{j}^{\ast }(z)<+\infty $ for every $z\geq 0$.

We now present two preliminary propositions before stating the main result
of the section. In Orlicz space theory the well known $\Delta _{2}$
condition on a Young function $\Phi :{\mathbb{R}}\rightarrow {\mathbb{R}}$
guarantees that $L^{\Phi }=M^{\Phi }$. We say that $\Phi \in \Delta _{2}$ if:

\begin{center}
There exists $y_{0}\geq 0,\,K>0$ such that $\Phi (2y)\leq K\Phi
(y)\,\,\forall \,y\text{ s.t. }\left|y\right|\geq y_{0}$.
\end{center}

First, we show how Assumption \ref{mSORTEA3} is linked to the $\Delta_2$
condition for the conjugate of ${\widehat\Phi}_j$ :

\begin{proposition}
\label{mSORTEpropdualisM} Let $\Phi :{\mathbb{R}}\rightarrow {\mathbb{R}}$
be a Young function differentiable on ${\mathbb{R}}\setminus \{0\}$ and let $%
\Phi ^{\ast }:{\mathbb{R}}\rightarrow {\mathbb{R}}$ be its conjugate
function. Then

\begin{equation}  \label{mSORTEeqequivalentdelta2}
\liminf_{z\rightarrow +\infty }\frac{z\Phi ^{\prime }(z)}{\Phi (z)}%
>1\,\,\Longleftrightarrow\,\,\Phi ^{\ast }\in \Delta _{2}.
\end{equation}
In particular, under Assumption \ref{mSORTEA3} we have ${\widehat\Phi}
_{1}^{\ast },\dots ,{\widehat\Phi}_{N}^{\ast }\in \Delta _{2}$ which implies 
\begin{equation}
K_{\Phi }=L^{{\widehat\Phi}_{1}^{\ast }}\times \dots \times L^{{\widehat\Phi}%
_{N}^{\ast }}=M^{{\widehat\Phi}_{1}^{\ast }}\times \dots \times M^{{%
\widehat\Phi}_{N}^{\ast }}.  \notag
\end{equation}
\end{proposition}

\begin{proof}
The equivalence of the two conditions in \eqref{mSORTEeqequivalentdelta2}
can be checked along the lines of Theorem 2.3.3 in \cite{RaoRen}, observing
that the argument still works in our slightly more general setup (use
Proposition 2.2 \cite{RaoRen} in place of Theorem 2.2.(a) \cite{RaoRen}). We
now prove the final claim. By Standing Assumption I and Remark \ref%
{mSORTERemMPHI}, $L^\Phi=L^{\widehat{\Phi}}=L^{{\widehat\Phi}_{1}}\times
\dots \times L^{{\widehat\Phi}_{N}}$. Since $L^\Phi=L^{\widehat{\Phi}}$, by
definition $K_\Phi=K_{\widehat\Phi}$. Furthermore, by Proposition \ref%
{mSORTEthmsummarykoethe} Item (3), $K_{\widehat\Phi}=L^{{\widehat\Phi}%
_{1}^{\ast }}\times \dots \times L^{{\widehat\Phi}_{N}^{\ast }}$. To
conclude, under Assumption \ref{mSORTEA3} ${\widehat\Phi}_{j}^{\ast }\in
\Delta _{2}$ by the previous part of this proof, which in turns implies $L^{{%
\widehat\Phi} _{j}^{\ast }}=M^{{\widehat\Phi}_{j}^{\ast }},\,j=1,\dots,N.$
\end{proof}

We also need a sequential w$^{\ast }$-compactness result, see \cite%
{DoldiThesis21} Proposition 2.6.10, partly inspired by \cite{Della} Chap.
II, proof of Theorem 24. A similar result is stated in \cite{Orihuela},
proof of Theorem 1, with a more technical (even though shorter) proof. For
these reasons we omit the proof.

\begin{proposition}
\label{mSORTEpropseqcpt} Assume that $\Phi $, $\Phi ^{\ast }:{\mathbb{R}}%
_+\rightarrow{\mathbb{R}}$ are (univariate) conjugate Young functions, both
everywhere finite valued. Then the balls in $L^{\Phi }$, endowed with Orlicz
norm, are $\sigma \left( L^{\Phi },M^{\Phi ^{\ast }}\right) $ sequentially
compact.
\end{proposition}

We are now ready for stating and proving the main result of this Section.

\begin{theorem}
Theorems \ref{mSORTEthmmaingeneral1}, \ref{mSORTEthmmaingeneral2}, \ref%
{mSORTEthmmsorteexists} and \ref{mSORTEthmadditionalreq} hold true by
replacing Assumption \ref{mSORTEA1} with Assumption \ref{mSORTEA3}
\end{theorem}

\begin{proof}
We do not need to start from scratch. In fact, most of the proof of Theorem %
\ref{mSORTEthmadditionalreq} carries over with no modifications. The only
point which needed closedness under truncation was proving that $%
\sum_{j=1}^{N}\mathbb{E}_{\mathbb{Q}^{j}}\left[ \widehat{Y}^{j}\right] \leq
0\,\forall {\mathbb{Q}}\in \mathcal{Q}_{\mathcal{B},V}\,$ using Lemma \ref{lemmafairnessgeneral}. If we show this in an alternative way, all the rest can
be done exactly in the same way. We observe that by %
\eqref{mSORTEnormbddnegparts} we have for each $j=1,\dots,N$ 
\begin{equation*}
\sup_H\mathbb{E}_\mathbb{P} \left[\widehat{\Phi}_j\left((X^j+W^j_H)^-\right)%
\right]\leq \varepsilon\sup_H\left(\sum_{j=1}^N\mathbb{E}_\mathbb{P} \left[%
\left|W_H^j\right|\right]\right)-\inf_H\mathbb{E}_\mathbb{P} \left[U(X+W_H)%
\right]+f(\varepsilon)=:\gamma<+\infty\,.
\end{equation*}
If $\gamma<1$ we immediately conclude that $\sup_H\left\|
(X^j+W^j_H)^-\right\|_{\widehat{\Phi}_j}\leq 1$. If conversely $\gamma>1$
using convexity of $\widehat{\Phi}_j$ we have 
\begin{equation*}
\sup_H\mathbb{E}_\mathbb{P} \left[\widehat{\Phi}_j\left(\frac{1}{\gamma}%
(X^j+W^j_H)^-\right)\right]\leq \frac{1}{\gamma}\sup_H\mathbb{E}_\mathbb{P} %
\left[\widehat{\Phi}_j\left((X+W_H)^-\right)\right]\leq 1
\end{equation*}
thus $\sup_H\left\| (X^j+W^j_H)^-\right\|_{\widehat{\Phi}_j}\leq \gamma$. In
conclusion, $\sup_H\left\| (X^j+W^j_H)^-\right\|_{\widehat{\Phi}_j}\leq
\max(\gamma,1)$. Now we apply Propositions \ref{mSORTEpropdualisM} and
proposition \ref{mSORTEpropseqcpt}. Given the sequences $%
(X^{j}+W_{H}^{j})^{-}$, $j=1,\dots ,N$, a diagonalization argument yields a
common subsequence such that $((X^{j}+W_{H}^{j})^{-})_{H}$ converges in $%
\sigma \left( L^{\widehat{\Phi}_{j}},M^{\widehat{\Phi}_{j}^{\ast }}\right) $
on $L^{\widehat{\Phi}_{j}}$ for every $j$. Call such limit $Z_{j}$. Almost
sure convergence 
\begin{equation*}
(X^{j}+W_{H}^{j})^{-}\rightarrow (X+\widehat{Y})^{-}\,\,{\mathbb{P}}-\text{%
a.s.}
\end{equation*}%
implies $Z=(X+Y)^{-}$. Indeed, if this were not the case assume without loss
of generality ${\mathbb{P}}(Z^{j}>(X^{j}+Y^{j})^{-})>0$ for some $j$. On a
measurable subset $D$ of the event $\{Z^{j}>(X^{j}+Y^{j})^{-}\}$, ${\mathbb{P%
}}(D)>0$, the convergence is uniform (by Egoroff Theorem, Theorem 10.38 in 
\cite{Aliprantis}). Consequently, by Dominated Convergence Theorem plus $%
\sigma \left( L^{\widehat{\Phi}_j }(\mathcal{F}),M^{\widehat{\Phi}_j ^{\ast
}}(\mathcal{F})\right) $ convergence and the fact that $L^{\infty }\subseteq
M^{\widehat{\Phi}_{j}^{\ast }},\,j=1,\dots ,N$ we get $\mathbb{E}_{\mathbb{P}%
}\left[ Z^{j}1_{D}\right] =\mathbb{E}_{\mathbb{P}}\left[
(X^{j}+Y^{j})^{-}1_{D}\right] $, which is a contradiction. Since by
Proposition \ref{mSORTEpropdualisM} 
\begin{equation*}
\mathcal{Q}_{\mathcal{B},V}\subseteq K_{\Phi }=M^{\widehat{\Phi}_{1}^{\ast }}\times
\dots \times M^{\widehat{\Phi}_{N}^{\ast }}
\end{equation*}%
we get for any ${\mathbb{Q}}\in \mathcal{Q}_{\mathcal{B},V}$
\begin{equation}
\sum_{j=1}^{N}\mathbb{E}_{\mathbb{P}}\left[ (X^{j}+W_{H}^{j})^{-}\frac{%
\mathrm{d}{\mathbb{Q}}^{j}}{\mathrm{d}{\mathbb{P}}}\right] \rightarrow
_{H}\sum_{j=1}^{N}\mathbb{E}_{\mathbb{P}}\left[ (X^{j}+\widehat{Y}^{j})^{-}%
\frac{\mathrm{d}{\mathbb{Q}}^{j}}{\mathrm{d}{\mathbb{P}}}\right] \,.
\label{mSORTEeqnegparts}
\end{equation}%
By Fatou Lemma and $x^{+}=x+x^{-}$ 
\begin{align*}
&\sum_{j=1}^{N}\mathbb{E}_{\mathbb{P}}\left[ (X^{j}+\widehat{Y}^{j})^{+}%
\frac{\mathrm{d}{\mathbb{Q}}^{j}}{\mathrm{d}{\mathbb{P}}}\right] \leq
\liminf_{H}\sum_{j=1}^{N}\mathbb{E}_{\mathbb{P}}\left[ (X^{j}+W_{H}^{j})^{+}%
\frac{\mathrm{d}{\mathbb{Q}}^{j}}{\mathrm{d}{\mathbb{P}}}\right]  \notag \\
&\leq \liminf_{H}\left( \mathbb{E}_{\mathbb{P}}\left[ \sum_{j=1}^{N}W_{H}^{j}%
\frac{\mathrm{d}{\mathbb{Q}}^{j}}{\mathrm{d}{\mathbb{P}}}\right]
+\sum_{j=1}^{N}\mathbb{E}_{\mathbb{P}}\left[ X^{j}\frac{\mathrm{d}{\mathbb{Q}%
}^{j}}{\mathrm{d}{\mathbb{P}}}\right] +\sum_{j=1}^{N}\mathbb{E}_{\mathbb{P}}%
\left[ (X^{j}+W_{H}^{j})^{-}\frac{\mathrm{d}{\mathbb{Q}}^{j}}{\mathrm{d}{%
\mathbb{P}}}\right] \right)  \notag \\
&\overset{\text{Prop.}\ref{mSORTEpropfairprob}}{\leq }\liminf_{H}\left(
\sum_{j=1}^{N}W_{H}^{j}+\sum_{j=1}^{N}\mathbb{E}_{\mathbb{P}}\left[ X^{j}%
\frac{\mathrm{d}{\mathbb{Q}}^{j}}{\mathrm{d}{\mathbb{P}}}\right] +\left(
\sum_{j=1}^{N}\mathbb{E}_{\mathbb{P}}\left[ (X^{j}+W_{H}^{j})^{-}\frac{%
\mathrm{d}{\mathbb{Q}}^{j}}{\mathrm{d}{\mathbb{P}}}\right] \right) \right) 
\notag \\
&=\lim_{H}\left( \sum_{j=1}^{N}W_{H}^{j}\right) +\sum_{j=1}^{N}\mathbb{E}_{%
\mathbb{P}}\left[ X^{j}\frac{\mathrm{d}{\mathbb{Q}}^{j}}{\mathrm{d}{\mathbb{P%
}}}\right] +\lim_{H}\left( \sum_{j=1}^{N}\mathbb{E}_{\mathbb{P}}\left[
(X^{j}+W_{H}^{j})^{-}\frac{\mathrm{d}{\mathbb{Q}}^{j}}{\mathrm{d}{\mathbb{P}}%
}\right] \right) \,.
\end{align*}%
where we used Equation \eqref{mSORTEeqnegparts} and the fact that $%
\sum_{j=1}^{N}W_{H}^{j}$ is a numeric sequence converging (${\mathbb{P}}$%
-a.s.) to $\sum_{j=1}^{N}\widehat{Y}^{j}$ to move from $\liminf $ to the sum
of limits. As a consequence 
\begin{equation}
\sum_{j=1}^{N}\mathbb{E}_{\mathbb{P}}\left[ (X^{j}+\widehat{Y}^{j})^{+}\frac{%
\mathrm{d}{\mathbb{Q}}^{j}}{\mathrm{d}{\mathbb{P}}}\right] \leq
\sum_{j=1}^{N}\widehat{Y}^{j}+\sum_{j=1}^{N}\mathbb{E}_{\mathbb{P}}\left[
X^{j}\frac{\mathrm{d}{\mathbb{Q}}^{j}}{\mathrm{d}{\mathbb{P}}}\right]
+\sum_{j=1}^{N}\mathbb{E}_{\mathbb{P}}\left[ (X^{j}+\widehat{Y}^{j})^{-}%
\frac{\mathrm{d}{\mathbb{Q}}^{j}}{\mathrm{d}{\mathbb{P}}}\right]\,.
\label{mSORTEeqsumhatuseful}
\end{equation}%
We get $(X+\widehat{Y})^\pm \in L^{1}({\mathbb{Q}})$, hence $\widehat{Y}\in
L^{1}({\mathbb{Q}})$ and rearranging terms in \eqref{mSORTEeqsumhatuseful} 
\begin{equation*}
\sum_{j=1}^{N}\mathbb{E}_{\mathbb{P}}\left[ (X^{j}+\widehat{Y}^{j})\frac{%
\mathrm{d}{\mathbb{Q}}^{j}}{\mathrm{d}{\mathbb{P}}}\right] \leq
\sum_{j=1}^{N}\widehat{Y}^{j}+\sum_{n=1}^{N}\mathbb{E}_{\mathbb{P}}\left[
X^{j}\frac{\mathrm{d}{\mathbb{Q}}^{j}}{\mathrm{d}{\mathbb{P}}}\right] \,.
\end{equation*}%
In particular, since $\widehat{Y}\in \mathcal{B}_{0}$, we conclude that $%
\sum_{j=1}^{N}\mathbb{E}_{\mathbb{Q}^{j}}\left[ \widehat{Y}^{j}\right] \leq
0\,\forall {\mathbb{Q}}\in \mathcal{Q}_{\mathcal{B},V}\,$. As mentioned before, all the
remaining parts of the proof are identical to the ones for Assumption \ref%
{mSORTEA1}. Hence Theorem \ref{mSORTEthmadditionalreq} holds true. Now we
get counterparts to Theorems \ref{mSORTEthmmaingeneral1} and \ref%
{mSORTEthmmaingeneral2} with the exact same arguments, recalling for the
latter that \eqref{mSORTEeqqfixedsup1} and \eqref{mSORTEeqqfixedsup12} still
hold dropping Assumption \ref{mSORTEA1} (see the proof of Theorem \ref%
{mSORTEthmadditionalreq}). Again using the same arguments of Theorem \ref%
{mSORTEthmmsorteexists} we then get its counterpart under the alternative
Assumption \ref{mSORTEA3}.
\end{proof}

\subsection{Working on $(L^\infty({\mathbb{P}}))^N$}

\label{mSORTEworkonlinfty}The following result is a counterpart to Theorem %
\ref{mSORTEthmadditionalreq} Item 1 when working with the dual system 
\begin{equation*}
((L^{\infty }({\mathbb{P}}))^{N},(L^{1}({\mathbb{P}}))^{N})
\end{equation*}
in place of $(M^{\Phi },K_{\Phi })$.

\begin{theorem}
\label{mSORTEthmlinfty}Under Assumption \ref{mSORTEA1} the following holds:%
\begin{equation}
\sup_{Y\in \mathcal{B}_{0}\cap (L^{\infty }({\mathbb{P}}))^{N}}\mathbb{E}_{%
\mathbb{P}}\left[ U(X+Y)\right] =\min_{\mathbb{Q}\in \mathcal{Q}_{\mathcal{B},V}}\min_{\lambda \geq 0}\left( \lambda \left( \sum_{j=1}^{N}\mathbb{E}_{%
\mathbb{Q}^{j}}\left[ X^{j}\right] \right) +\mathbb{E}_{\mathbb{P}}\left[
V\left( \lambda \frac{\mathrm{d}{\mathbb{Q}}}{\mathrm{d}{\mathbb{P}}}\right) %
\right] \right) \,.  \label{mSORTEeqminimaxappliedcor21Alinfty}
\end{equation}
\end{theorem}

\begin{proof}
To check \eqref{mSORTEeqminimaxappliedcor21Alinfty} we can apply the same
argument used in Step 1 of the proof of Theorem \ref{mSORTEthmadditionalreq}%
, by replacing Theorem \ref{mSORTEthmminimax} with Theorem \ref%
{mSORTEthmminimaxlinfty}. What is left to prove then is that for $\mathcal{C}%
=\mathcal{B}_{0}\cap (L^{\infty }({\mathbb{P}}))^{N}$, the set 
\begin{equation*}
\mathcal{N}:=(\mathcal{C}_{1}^{0})^{+}\cap \left\{ Z\in (L^{1}({\mathbb{P}}%
)_{+})^{N}\mid \mathbb{E}_{\mathbb{P}}\left[ V(\lambda Z)\right] <+\infty 
\text{ for some }\lambda >0\right\}
\end{equation*}%
is in fact $\mathcal{Q}_{\mathcal{B},V}$. To see this, observe that as consequence of
Lemma \ref{mSORTElemmaevfiniteisinkphi} we have $\mathcal{N}\subseteq K_\Phi$%
. From this, by closedness under truncation we have for any $Y\in\mathcal{B}%
_0\cap M^\Phi$ a sequence $(Y_n)_n\subseteq \mathcal{B}_0\cap (L^\infty({%
\mathbb{P}}))^N$ such that for each ${\mathbb{Q}}\in \mathcal{N}$, for each $%
j=1,\dots,N$ $Y^j_n\rightarrow_n Y^j$ ${\mathbb{Q}}^j$-a.s. and the
convergence is dominated: argue as in Lemma \ref{lemmafairnessgeneral}. Thus for any $Y\in\mathcal{B}_0\cap M^\Phi$ we have by Dominated Convergence Theorem that $%
\sum_{j=1}^{N}\mathbb{E}_{\mathbb{Q}^{j}}\left[ Y^{j}\right] \leq 0$. This
completes the proof that $\mathcal{N}=\mathcal{Q}_{\mathcal{B},V}$.
\end{proof}

\begin{corollary}
\label{mSORTEcorsuponlinftyeqmphi} Under Assumption \ref{mSORTEA1} we have 
\begin{equation}  \label{mSORTEeqsuplinftyequalsupmphi}
\sup_{Y\in\mathcal{B}_0\cap (L^\infty({\mathbb{P}}))^N}\mathbb{E}_\mathbb{P} %
\left[U(X+Y)\right]=\sup_{Y\in\mathcal{B}_0\cap M^\Phi}\mathbb{E}_\mathbb{P} %
\left[U(X+Y)\right]\,.
\end{equation}
\end{corollary}

\begin{proof}
By Theorem \ref{mSORTEthmadditionalreq} Item 1 and Theorem \ref%
{mSORTEthmlinfty}, both LHS and RHS of \eqref{mSORTEeqsuplinftyequalsupmphi}
are equal to the minimax expression

\begin{equation*}
\min_{\lambda\geq 0,\,{\mathbb{Q}}\in\mathcal{Q}_{\mathcal{B},V}}\left(\lambda\left(%
\sum_{j=1}^N\mathbb{E}_{\mathbb{Q}^j} \left[X^j\right]\right)+\mathbb{E}_%
\mathbb{P} \left[V\left(\lambda \frac{\mathrm{d}{\mathbb{Q}}}{\mathrm{d}{%
\mathbb{P}}}\right)\right]\right)\,.
\end{equation*}
\end{proof}

\subsection{General case: total wealth $A\in{\mathbb{R}}$}

\label{mSORTESecfrom0toA} In this section we extend previous results to
cover the case in which the total wealth $A$ might not be equal to $0$. For $%
A\in{\mathbb{R}}$ and ${\mathbb{Q}}\in\mathcal{Q}_{\mathcal{B},V}$ recall the definitions
of $\pi_A(X)$ in \eqref{mSORTEpiA} and introduce, coherently with %
\eqref{mSORTEeqqfixedsup1},

\begin{equation*}
\pi_A^{\mathbb{Q}}(X):=\sup\left\{\mathbb{E}_\mathbb{P} \left[U(X+Y)\right]%
\mid Y\in M^\Phi,\,\sum_{j=1}^N\mathbb{E}_{\mathbb{Q}^j} \left[Y^j\right]%
\leq A\right\}\,.
\end{equation*}
It is possible to reduce the maximization problem expressed by $\pi_A(X)$
(and similarly $\pi_A^{\mathbb{Q}}(X)$) to the problem related to $%
\pi_0(\cdot)$ (respectively, $\pi_0^{\mathbb{Q}}(\cdot)$).

Take any $a=[a^1,\dots,a^N]\in{\mathbb{R}}^N$ with $\sum_{j=1}^Na^j=A$. Then

\begin{equation*}
\pi_A(X)=\sup \left\{ \mathbb{E}_{\mathbb{P}}\left[ U\left( X+Y+a-a\right) %
\right] \mid \left( Y-a\right) \in \mathcal{B}\cap M^{\Phi
},\sum_{j=1}^{N}\left( Y^{j}-a^{j}\right) \leq 0\right\} =
\end{equation*}%
\begin{equation*}
=\sup \left\{ \mathbb{E}_{\mathbb{P}}\left[ U\left( X+Z+a\right) \right]
\mid Z\in \mathcal{B}_{0}\cap M^{\Phi }\right\} =\pi _{0}(X+a)
\end{equation*}%
where last line holds since ${\mathbb{R}}^{N}+\mathcal{B}=\mathcal{B}$ under
Standing Assumption II. We recognize then that $\pi_A(X)$ is just $\pi
_{0}(\cdot )$, with different initial point $(X+a)$ in place of $X$. %

The same technique adopted above can be exploited to show that for any $a\in{%
\mathbb{R}}^N$ with $\sum_{j=1}^Na^j=A$ 
\begin{equation*}
\begin{split}
\sup&\left\{\mathbb{E}_\mathbb{P} \left[U(X+Y)\right]\mid Y\in\mathcal{L}%
,\,\sum_{j=1}^N\mathbb{E}_{\mathbb{Q}^j} \left[Y^j\right]\leq A\,\forall{%
\mathbb{Q}}\in\mathcal{Q}_{\mathcal{B},V} \right\} \\
=&\sup\left\{\mathbb{E}_\mathbb{P} \left[U(X+a+Z)\right]\mid Z\in\mathcal{L}%
,\,\sum_{j=1}^N\mathbb{E}_{\mathbb{Q}^j} \left[Z^j\right]\leq 0\,\forall{%
\mathbb{Q}}\in\mathcal{Q}_{\mathcal{B},V}\right\}\,.
\end{split}%
\end{equation*}

The argument above shows how to generalize Theorem \ref%
{mSORTEthmadditionalreq}, Theorem \ref{mSORTEthmlinfty}, Corollary \ref%
{mSORTEcorsuponlinftyeqmphi} to cover the case $A\neq 0$, exploiting the
same results with $X+a$ in place of $X$.

Thus the statements of Theorem \ref{mSORTEthmadditionalreq}, Theorem \ref%
{mSORTEthmlinfty}, Corollary \ref{mSORTEcorsuponlinftyeqmphi} remain true
replacing $0,\,\mathcal{B}_{0}$ with $A,\,\mathcal{B}_{A}$ respectively, and
Equation \eqref{mSORTEeqqfixedsup12} (similarly for %
\eqref{mSORTEeqwithCminimax}, \eqref{mSORTEeqminimaxappliedcor3}, %
\eqref{mSORTEeqminimaxappliedcor21Alinfty}) with

\begin{equation}
\min_{\lambda\geq 0}\left(\lambda\left(\sum_{j=1}^N\mathbb{E}_{\mathbb{Q}^j} %
\left[X^j\right]+A\right)+\mathbb{E}_\mathbb{P} \left[V\left(\lambda\frac{%
\mathrm{d}{\mathbb{Q}}}{\mathrm{d}{\mathbb{P}}}\right)\right]\right)\,. 
\notag
\end{equation}

\section{On the assumptions and examples}

\label{mSORTEsecexamples} We here present a technique to produce a wide
variety of multivariate utility functions fulfilling our requirement. First,
we show how to produce multivariate utility functions satisfying %
\eqref{mSORTEcontrolwithphihat}.

\begin{proposition}
\label{mSORTEPropBLambda}Let $u_{1},\dots ,u_{N}$ be univariate utility
functions such that $u_{1}(0)=\dots =u_{N}(0)=0$ and that the Inada
conditions hold (see Remark \ref{mSORTEreminada}). Let $\Lambda :{\mathbb{R}}%
^{N}\rightarrow {\mathbb{R}}$ be concave and increasing, both not
necessarily strictly, and bounded from above. Set 
\begin{equation}
U(x):=\sum_{j=1}^{N}u_{j}(x)+\Lambda (x),\quad x\in {\mathbb{R}}%
^{N},\,\,\,\,\,\text{and}\,\,\,\,\,\widehat{\Phi }(x):=\sum_{j=1}^{N}%
\widehat{\Phi }_{j}(x^{j})\quad x\in {\mathbb{R}}_{+}^{N},
\label{mSORTEUaddLambda}
\end{equation}%
where $\widehat{\Phi }_{j}(z):=-u_{j}(-z)$ for $z\geq 0$. Then $\,\widehat{%
\Phi }$ is a multivariate Orlicz function and it ensures that $U$ is well
controlled (see Definition \ref{mSORTEwellcontrolled}). If additionally, for
every $j=1,\dots ,N,$ $u_{j}$ satisfies asymptotic elasticity at $-\infty $,
that is $u_{j}$ is differentiable on ${\mathbb{R}}$ and 
\begin{equation*}
\liminf_{x\rightarrow -\infty }\frac{xu_{j}^{\prime }(x)}{u_{j}(x)}>1\,,
\end{equation*}%
then also Assumption \ref{mSORTEA3} holds true.
\end{proposition}

\begin{proof}
First we motivate why $U$ in \eqref{mSORTEUaddLambda} is a multivariate
utility function: this follows from the fact that $U$ inherits strict
concavity and strict monotonicity from $u_1,\dots,u_N$ which are univariate
utility functions.

Now we show that the control from above in \eqref{mSORTEcontrolwithphihat}
holds true. To this end, notice that using the fact that $%
u_j(x)=u_j((x)^+)+u_j(-(x)^-)$ we have for any $z\geq 0$ 
\begin{align}
U(x)&=-\left(\sum_{j=1}^N-u_j(-(x^j)^-)\right)+\sum_{j=1}^Nu_j((x^j)^+)+%
\Lambda(x)  \notag \\
&\leq -\left(\sum_{j=1}^N-u_j(-(x^j)^-)\right)+\left(\max_{j=1}^N\frac{%
\mathrm{d}^+u_j}{\mathrm{d}x^j}(z)\right)\left(\sum_{j=1}^N(x^j)^+-z\right)+%
\sum_{j=1}^Nu_j(z)+\sup_{z\in{\mathbb{R}}^N}\Lambda(z)
\label{mSORTEpartialexample}
\end{align}

where $\frac{\mathrm{d}^+u_j}{\mathrm{d}x^j}(z)$ are the usual right
derivatives at $z$.

The Inada conditions imply that: for any $\varepsilon>0$ there exists $z\geq
0$ such that $\max_{j=1}^N\frac{\mathrm{d}^+u_j}{\mathrm{d}x^j}(z)\leq
\varepsilon$, and there exist $A>0,B\in{\mathbb{R}}$ such that $%
u_j(-(x^j)^-)\geq -A(x^j)^-+B$. From the latter we conclude that $\widehat{%
\Phi}(\cdot)$ is a multivariate Orlicz function according to Definition \ref%
{mSORTEdeforliczfunct}. From the former, continuing from %
\eqref{mSORTEpartialexample} once $\varepsilon>0$ had been fixed, 
\begin{equation*}
U(x)\leq -\widehat{\Phi}((x)^-)+\varepsilon \sum_{j=1}^N(x^j)^++\text{%
constant}(\varepsilon)\leq -\widehat{\Phi}((x)^-)+\varepsilon
\sum_{j=1}^N\left|x^j\right|+\text{constant}(\varepsilon)
\end{equation*}
which clearly implies the existence of $\widehat{\Phi},f$ satisfying %
\eqref{mSORTEcontrolwithphihat}.

To conclude, we prove the claim concerning Assumption \ref{mSORTEA3}.
Differentiability of $\widehat{\Phi}_1,\dots,\widehat{\Phi}_N$ on $%
(0,+\infty)$ follows from differentiability of $u_1,\dots,u_N$. It is
finally easy to verify that asymptotic elasticity at $-\infty$ implies $%
\liminf_{z\rightarrow +\infty }\frac{z\widehat{\Phi}_j ^{\prime }(z)}{%
\widehat{\Phi}_j (z)}>1\,\,\,\forall \,j=1,\dots ,N\,.$
\end{proof}

It remains now to elaborate on how to guarantee the integrability condition $%
L^{\widehat{\Phi}}=L^\Phi$ to obtain examples in which Standing Assumption I
holds. At this point we introduce a definition, inspired by Definition 2.2.1 
\cite{RaoRen}, which will serve precisely for the scope (see Proposition \ref%
{mSORTEPropB}).

\begin{definition}
\label{mSORTEdefuu}Let $u:\mathbb{R\rightarrow R}$ and $\widetilde{u}:%
\mathbb{R\rightarrow R}$. We say that $u\preceq \widetilde{u}$ if there
exist $k\in {\mathbb{R}}$, $c\in {\mathbb{R}}_{+},$ $C\in {\mathbb{R}}_{+}$
such that $\widetilde{u}(x)\geq Cu(cx)+k$ for each $x\leq 0$.
\end{definition}

Note that such control is required to hold only for negative values.

\begin{proposition}
\label{mSORTEPropB}Under the same assumptions and notation of Proposition %
\ref{mSORTEPropBLambda} suppose additionally that

\begin{equation}
\Lambda(x)=\sum_{k=1}^K \Lambda^k\left( \sum_{j=1}^{N}\beta^k _{j}x^{j}\right) ,\text{ with }\beta
^1_{1},\dots ,\beta^1 _{N},\dots, \beta
^K_{1},\dots ,\beta^K _{N}\geq 0\text{ and }\max_{j,k}\beta^k_j>0\, \notag
\end{equation}
for bounded from above, (not necessarily strictly) concave increasing $\Lambda^1,\dots,\Lambda^K:\R\rightarrow\R$.
Suppose that for every  $k=1,\dots,K$ it holds that $u_{j}\preceq \Lambda^k$, for each $j=1,\dots,N$. Then Standing Assumption I holds true.
\end{proposition}

\begin{proof}
By Proposition \ref{mSORTEPropBLambda} we only need to prove that $L^\Phi=L^{%
\widehat{\Phi}}$, which will conclude the proof of the fact that Standing
Assumption I holds. For the sake of notational simplicity, we take $K=1$. It will become clear at the end of the proof that the generalization to $K\geq 2$ is immediate. By the concavity of $\Lambda^1$ we have, for every $x\in {%
\mathbb{R}}^{N},$ 
\begin{equation}
\begin{split}
&\Lambda^1\left( \sum_{j=1}^{N}\beta _{j}x^{j}\right) =\Lambda^1\left( \sum_{j=1}^{N}\frac{%
\beta _{j}}{\sum_{n=1}^{N}\beta _{n}}\left( \sum_{n=1}^{N}\beta _{n}\right)
x^{j}\right) \\
&\geq \sum_{j=1}^{N}\frac{\beta _{j}}{\sum_{n=1}^{N}\beta _{n}}\Lambda^1\left(
\left( \sum_{n=1}^{N}\beta _{n}\right) x^{j}\right) .  \label{mSORTE111}
\end{split}%
\end{equation}%
By $u_{j}\preceq \Lambda^1,$ and boundedness from above of $\Lambda^1$ we have for each $%
x\in ((-\infty ,0])^{N}$ and from \eqref{mSORTE111} 
\begin{equation}
+\infty >\sup_{z\in {\mathbb{R}}^{N}}\Lambda^1 (z)\geq \Lambda^1(x)\geq \sum_{j=1}^{N}\frac{%
\beta _{j}}{\sum_{n=1}^{N}\beta _{n}}\left( C_{j}u_{j}\left( c_{j}\left(
\sum_{n=1}^{N}\beta _{n}\right) x^{j}\right) +k_{j}\right) \,.
\label{mSORTE333}
\end{equation}%
If $X\in L^{\widehat{\Phi}}$, then by definition there exists a $\lambda
_{0}>0$ such that the inequality $\mathbb{E}_{\mathbb{P}}\left[
u_{j}(\lambda (-\left\vert X^{j}\right\vert )\right] >-\infty $ holds for
every $\lambda \leq \lambda _{0}$ and $j=1,\dots ,N$. This and %
\eqref{mSORTE333} then imply the existence of some $\lambda_0>\lambda _{1}>0$ such
that $\mathbb{E}_{\mathbb{P}}\left[ \Lambda (-\lambda \left\vert X\right\vert )%
\right] >-\infty $ for every $\lambda \leq \lambda _{1},$ that is $X\in
L^{\Phi }.$
\end{proof}

Just to mention a few explicit examples, any of the following multivariate
utility functions satisfy the Standing Assumption I and Assumption \ref%
{mSORTEA3}: 
\begin{equation}
U(x):=\sum_{j=1}^{N}u_{j}(x^{j})+u\left( \sum_{j=1}^{N}\beta
_{j}x^{j}\right) ,\text{\quad with }\text{ }\beta _{j}\geq 0\text{, for all }%
j,  \label{mSORTEuuu}
\end{equation}%
where $u:\mathbb{R\rightarrow R}$, for some $p>1,$ is any one of the
following functions:%
\begin{equation*}
u(x):=1-\exp \left( -px\right) ;\text{\quad }u(x)=\,%
\begin{cases}
\,p\,\frac{x}{x+1} & \,\,\,x\geq 0 \\ 
\,1-\left\vert x-1\right\vert ^{p} & \,\,\,x<0%
\end{cases}%
;\text{\quad }u(x)=\,%
\begin{cases}
\,p\,\arctan (x) & \,\,\,x\geq 0 \\ 
\,1-\left\vert x-1\right\vert ^{p} & \,\,\,x<0%
\end{cases}
\end{equation*}%
and $u_{1},\dots ,u_{N}$ are exponential utility functions ($%
u_{j}(x^{j})=1-\exp {(-\alpha _{j}x^{j})},\,\alpha >0$) for any choice of $u$
as above.

The function $\Lambda $ in \eqref{mSORTEUaddLambda} could also be
constructed as follows. Let $G:\mathbb{R}^{N}\mathbb{\rightarrow R}$ be
convex, monotone decreasing and bounded from below, and $F:\mathbb{%
R\rightarrow R}$ be concave and monotone decreasing on $range(G)$. Then $%
\Lambda :\mathbb{R}^{N}\mathbb{\rightarrow R}$ defined by 
\begin{equation}
\Lambda (x)=F(G(x))  \notag
\end{equation}%
is concave, monotone increasing and bounded above by $F(\inf G).$ Notice
that we require differentiability only in few circumstances, mainly
concerning uniqueness. We here provide an example in which our assumptions
are met, covering the non differentiable case. Take $\gamma _{j}\geq
0,\,j=1,\dots ,N$, $G(x):=\sum_{j=1}^{N}\gamma _{j}(x^{j}-k^{j})^{-}$ and
take $F:\mathbb{R\rightarrow R}$ defined by $F(x):=-x^{\alpha }$ , $\alpha
\geq 1$,\ which is concave and monotone decreasing on $range(G)=[0,\infty ).$
Then 
\begin{equation}
\Lambda (x):=-\left( \sum_{j=1}^{N}\gamma _{j}(x^{j}-k^{j})^{-}\right)
^{\alpha }  \label{mSORTEExample1}
\end{equation}%
is concave, monotone increasing and bounded above by $0$, and $%
U(x):=\sum_{j=1}^{N}u_{j}(x^{j})+\Lambda (x)$, with $u_{1},\dots ,u_{N}$
exponential utility functions and $\Lambda $ assigned in %
\eqref{mSORTEExample1}, satisfies Standing Assumption I.

\bigskip

We conclude this Section providing a method to identify suitable candidates
for Standing Assumption I.

\begin{remark}
\label{mSORTEremuinfty} Let $U$ be a multivariate utility function, bounded
from above. Define for each $j=1,\dots,N$ 
\begin{equation*}
u_j^\infty(z):=\sup_{x^{[-j]}\in{\mathbb{R}}^{N-1}}U\left([x^{[-j]};z]%
\right)\,.
\end{equation*}
Then $u_1^\infty,\dots,u_N^\infty$ are real valued, concave and
nondecreasing, and satisfy 
\begin{equation*}
U(x)\leq \sum_{j=1}^Nu_j^\infty(-(x^j)^-)+\sup_{z\in{\mathbb{R}}^N}U(z)\,.
\end{equation*}

With the same notation for $\widehat{\Phi }$ and $\widehat{\Phi }_{j}$ in
Proposition \ref{mSORTEPropBLambda} (with $u_j^\infty$ in place of $u_j$),
we see that $U(x)\leq -\widehat{\Phi }((x)^{-})+\sup_{z\in {\mathbb{R}}%
^{N}}U(z)$. Thus, if one can guarantee, say by explicit computation, that $%
(a)$ $u_{1}^{\infty },\dots ,u_{N}^{\infty }$ are \emph{strictly} concave
and \emph{strictly} increasing, $(b)$ that $L^{\widehat{\Phi }}=L^{\Phi }$,
then Standing Assumption I is satisfied.
\end{remark}

\subsection{Comparison with (univariate) SORTE\label{mSORTEseccomparison}}

We suppose now that $U$ has the form %
\eqref{mSORTEuuu} for $u\equiv0$ (which is allowed, since $u$ needs not be
strictly concave nor strictly increasing by Proposition \ref%
{mSORTEPropBLambda} and \ref{mSORTEPropB}). 
Assumption \ref{mSORTEA1} and Assumption \ref{mSORTEA3} are left untouched,
and Standing Assumption I is satisfied automatically provided that $%
u_1,\dots,u_N$ satisfy the assumptions in Proposition \ref{mSORTEPropB}.
Theorems \ref{mSORTEthmmsorteexists} and \ref{mSORTEthmstrongunique} show
that both existence and uniqueness can be obtained assuming closedness under
truncation for far less strong assumptions than the ones in \cite{BDFFM}
(that is, without requesting the validity of Equation (22) in \cite{BDFFM}). In particular then the following holds true.

\begin{corollary}
\label{mSORTEcorcarteexists}Let let $u_{1},\dots ,u_{j}:{\mathbb{R}}%
\rightarrow {\mathbb{R}}$ be strictly increasing, strictly concave and
satisfying the Inada conditions. Then under either Assumption \ref{mSORTEA1}
or \ref{mSORTEA3} a SORTE, as defined in \cite{BDFFM}, exists.

%
%
%
\end{corollary}

\appendix

\section{Appendix}
\begin{lemma}
\label{lemmafairnessgeneral}
Suppose Assumption \ref{mSORTEA1} is satisfied. Then for any $Y\in \mathcal{L}\cap\mathcal{B}_A$ and any $\probq\in \mathcal{Q}_{\mathcal{B},V}$  it holds that
$$\sum_{j=1}^N\mathbb{E}_{\mathbb{Q}^j} \left[Y^j\right]\leq A\,. $$
\end{lemma}
\begin{proof}
Taking $Y_{m}$ as in
Definition \ref{mSORTEdefclosedundertrunc}, we have 
\begin{equation*}
\left\vert \sum_{j=1}^{N}Y_{m}^{j}\frac{\mathrm{d}{\mathbb{Q}}^{j}}{\mathrm{d%
}{\mathbb{P}}}\right\vert \leq \max \left( \left\vert \sum_{j=1}^{N}Y^{j}\frac{\mathrm{d}{\mathbb{Q}}^{j}}{\mathrm{d}{\mathbb{P}}}\right\vert
,\sum_{j=1}^{N}\left\vert c_{y}^{j}\right\vert\frac{\mathrm{d}{\mathbb{Q}}%
^{j}}{\mathrm{d}{\mathbb{P}}} \right) \in L^{1}({\mathbb{P}})
\end{equation*}%
and 
\begin{equation}  \label{mSORTEstep2beq}
\sum_{j=1}^{N}Y_{m}^{j}\frac{\mathrm{d}{\mathbb{Q}}^{j}}{\mathrm{d}{\mathbb{P%
}}}\rightarrow _{m}\sum_{j=1}^{N}{Y}^{j}\frac{\mathrm{d}{\mathbb{Q}}%
^{j}}{\mathrm{d}{\mathbb{P}}}\,\,{\mathbb{P}}-\text{ a.s.}
\end{equation}%
hence by Dominated Convergence Theorem 
\begin{equation}  \label{mSORTEDOM}
A\geq \mathbb{E}_{\mathbb{P}}\left[ \sum_{j=1}^{N}Y_{m}^{j}\frac{\mathrm{d}{%
\mathbb{Q}}^{j}}{\mathrm{d}{\mathbb{P}}}\right] \rightarrow _{m}\mathbb{E}_{%
\mathbb{P}}\left[ \sum_{j=1}^{N}{Y}^{j}\frac{\mathrm{d}{\mathbb{Q}}%
^{j}}{\mathrm{d}{\mathbb{P}}}\right]=\sum_{j=1}^{N}\mathbb{E}_{\mathbb{Q}^{j}}\left[ Y^{j}\right]
\end{equation}%
where the inequality for LHS comes from the fact that $Y_{m}\in \mathcal{B}%
_{A}\cap (L^{\infty })^{N}\subseteq \mathcal{B}_{A}\cap M^{\Phi }$ and ${%
\mathbb{Q}}\in \mathcal{Q}_{\mathcal{B},V}$, so that by Proposition \ref{mSORTEpropfairprob} $\mathbb{E}_{\mathbb{P}}%
\left[ \sum_{j=1}^{N}Y^{j}_m\frac{\mathrm{d}{\mathbb{Q}} ^{j}}{\mathrm{d}{%
\mathbb{P}}}\right] =\sum_{j=1}^{N}\mathbb{E}_{\mathbb{Q}^{j}}\left[ Y^{j}_m%
\right] \leq \sum_{j=1}^N Y_m^j\leq A.$
\end{proof}
\begin{lemma}
For every $\probq\in\mathcal{Q}_{\mathcal{B},V}$
\label{mSORTESequalsH}
\begin{equation}
\label{sqequalshq}
\begin{split}
S^{{{\mathbb{Q}}}}(A)&:=\sup\left\{
\sup \left\{ \mathbb{E}_{\mathbb{P}}\left[ U(X+Y)\right] \mid Y\in \mathcal{L%
},\,\mathbb{E}_{{\mathbb{Q}}^{j}}\left[ Y^{j}\right] \leq
a^{j},\,\forall \,j\mid  a\in {\mathbb{R}}^{N},\, \sum_{j=1}^{N}a_{j}=A\right\} \right\} 
\\
&=\sup \left\{ \sup_{\widetilde{Y}\in\mathcal{L}}\mathbb{E}_{\mathbb{P}}\left[ U\left( a+X+%
\widetilde{Y}-{\mathbb{E}_{{{\mathbb{Q}}}}[\widetilde{Y}]}\right) \right]
\mid a\in {\mathbb{R}}^{N},\,\sum_{j=1}^{N}a^{j}=A\right\}=:H^{{{\mathbb{Q}}}}(A) \,.
\end{split}
\end{equation}%
Furthermore, any optimum $(\widehat{Y},\widehat{a})$ for $S^{{{\mathbb{Q}}}}(A)$ produces an optimum $(\widetilde{Y},\widehat{a})$ for $H^{{{\mathbb{Q}}}}(A)$ by setting $\widetilde{Y}=\widehat{Y}-\widehat{a}$, and any optimum  $(\widetilde{Y},\widehat{a})$ for $H^{{{\mathbb{Q}}}}(A)$  produces an optimum $(\widehat{Y},\widehat{a})$ for $S^{{{\mathbb{Q}}}}(A)$ by setting $\widehat{Y}=\widetilde{Y}+\widehat{a}-\mathbb{E}_\probq[\widetilde{Y}]$.
\end{lemma}
\begin{proof}
We fist show $S^{{{\mathbb{Q}}}}(A)\geq H^{{{\mathbb{Q}}}}(A)$. Take $a\in\R^N$ with $\sum_{j=1}^Na^j=A$ and $\widetilde{Y}\in\mathcal{L}$. Define $Y:=a+\widetilde{Y}-\mathbb{E}_{{\probq}}[\widetilde{Y}]$. Then $Y\in\mathcal{L}$ and $\mathbb{E}_{{\probq}^j}[Y^j]=a^j$, therefore $S^{{{\mathbb{Q}}}}(A)\geq \Ep{U(X+Y)}=\Ep{U(X+a+\widetilde{Y}-\mathbb{E}_{{\probq}}[\widetilde{Y}])}$. Taking now a supremum over $\{\widetilde{Y}\in\mathcal{L},a\in\R^N\text{ s.t. }\sum_{j=1}^Na^j=A\}$ the claim follows. 

Now we show $S^{{{\mathbb{Q}}}}(A)\leq H^{{{\mathbb{Q}}}}(A)$. Observe first that since $U$ is strictly increasing the budget constraints are tight, so that all the inequalities in the definition of $S^{{{\mathbb{Q}}}}(A)$ can be turned into equalities without affecting the value of the supremum.

  Take $Y\in\mathcal{L},a\in\R^N$ with $\sum_{j=1}^Na^j=A$ and $\mathbb{E}_{{\probq}}[Y]=a^j,j=1,\dots,N$. Define $\widetilde{Y}:=Y-a$ and observe that $\widetilde{Y}\in\mathcal{L}$ and  $\mathbb{E}_{{\probq}}[\widetilde{Y}]=0$ . Then $H^{{{\mathbb{Q}}}}(A)\geq\Ep{U(X+a+\widetilde{Y}-\mathbb{E}_{{\probq}}[\widetilde{Y}])}= \Ep{U(X+Y)}$. The claim follows taking a supremum over $\{Y\in\mathcal{L},a\in\R^N \text{ s.t. }\mathbb{E}_{{\probq}}[Y]=a^j,\sum_{j=1}^Na^j=A\}$.  This argument also shows how optima for one of these problems produce optima for the other in the way descried in the statement (just observe that in this case all the inequalities above become equalities).
\end{proof}

\subsection{Additional properties of multivariate utility functions}

We work under Standing Assumption I and II without further mention. 

\begin{lemma}
\label{mSORTElemmakomlos} Let $(Z_n)_n$ be a sequence of random variables
taking values in ${\mathbb{R}}^N$ such that $\mathbb{E}_\mathbb{P} \left[%
U(Z_n)\right]\geq B$ for all $n$, for some $B\in{\mathbb{R}}$.

\begin{enumerate}
\item {\label{mSORTElemmabdd1} If $\sup_n\left|\sum_{j=1}^N\mathbb{E}_%
\mathbb{P} \left[Z^j_n\right]\right|<+\infty$ then $\sup_n\sum_{j=1}^N%
\mathbb{E}_\mathbb{P} \left[\left|Z^j_n\right|\right]<\infty$.}

\item {\label{mSORTElemmabdd2} If $Z_n\rightarrow Z$ ${\mathbb{P}}$-a.s. and 
$\sup_n\sum_{j=1}^N\mathbb{E}_\mathbb{P} \left[(Z^j_n)^+\right]<+\infty$
then $\mathbb{E}_\mathbb{P} \left[U(Z)\right]\geq B$.}
\end{enumerate}
\end{lemma}

\begin{proof}
$\,$

\textbf{Item \ref{mSORTElemmabdd1}}. Suppose that 
\begin{equation*}
\sup_{n}\left( \sum_{j=1}^{N}\mathbb{E}_{\mathbb{P}}\left[ \left\vert
Z_{n}^{j}\right\vert \right] \right) =\sup_{n}\left( \sum_{j=1}^{N}\mathbb{E}%
_{\mathbb{P}}\left[ (Z_{n}^{j})^{+}\right] +\sum_{j=1}^{N}\mathbb{E}_{%
\mathbb{P}}\left[ (Z_{n}^{j})^{-}\right] \right) =+\infty \,.
\end{equation*}%
From the boundedness of 
\begin{equation*}
\sum_{j=1}^{N}\mathbb{E}_{\mathbb{P}}\left[ Z_{n}^{j}\right] =\sum_{j=1}^{N}%
\mathbb{E}_{\mathbb{P}}\left[ (Z_{n}^{j})^{+}\right] -\sum_{j=1}^{N}\mathbb{E%
}_{\mathbb{P}}\left[ (Z_{n}^{j})^{-}\right]
\end{equation*}%
we conclude that $\sup_{n}\sum_{j=1}^{N}\mathbb{E}_{\mathbb{P}}\left[
(Z_{n}^{j})^{-}\right] =+\infty $. Select $a,b$ as in %
\eqref{mSORTElemmacontrolwithline} . Then we have 
\begin{equation*}
B\leq \mathbb{E}_{\mathbb{P}}\left[ U(Z_{n})\right] \leq a\sum_{j=1}^{N}%
\mathbb{E}_{\mathbb{P}}\left[ Z_{n}^{j}\right] -a\sum_{j=1}^{N}\mathbb{E}_{%
\mathbb{P}}\left[ (Z_{n}^{j})^{-}\right] +b
\end{equation*}%
which is clearly a contradiction.

\textbf{Item \ref{mSORTElemmabdd2}}. For $\varepsilon >0$ define the
function $\Gamma _{\varepsilon }$ as 
\begin{equation*}
\Gamma _{\varepsilon }(x):=2\varepsilon \left(
\sum_{j=1}^{N}(x^{j})^{+}\right) +b_{\varepsilon }-U(x)
\end{equation*}%
where the coefficient $b_{\varepsilon }$ is the one in %
\eqref{mSORTEcontrolwithepsilon}. Then $\Gamma _{\varepsilon }\geq 0$ and by
Fatou Lemma we have 
\begin{align*}
&2\varepsilon \left( \sum_{j=1}^{N}\mathbb{E}_{\mathbb{P}}\left[ (Z^{j})^{+}%
\right] \right) +b_{\varepsilon }-\mathbb{E}_{\mathbb{P}}\left[ U(Z)\right] =%
\mathbb{E}_{\mathbb{P}}\left[ \Gamma _{\varepsilon }(Z)\right] \leq
\liminf_{n}\mathbb{E}_{\mathbb{P}}\left[ \Gamma _{\varepsilon }(Z_{n})\right]
\\
&=\liminf_{n}\left( 2\varepsilon \left( \sum_{j=1}^{N}\mathbb{E}_{\mathbb{P}}%
\left[ (Z_{n}^{j})^{+}\right] \right) +b_{\varepsilon }-\mathbb{E}_{\mathbb{P%
}}\left[ U(Z_{n})\right] \right) \\
&\leq -B+b_{\varepsilon }+2\varepsilon \liminf_{n}\left( \sum_{j=1}^{N}%
\mathbb{E}_{\mathbb{P}}\left[ (Z_{n}^{j})^{+}\right] \right) \,.
\end{align*}%
As a consequence 
\begin{equation*}
\mathbb{E}_{\mathbb{P}}\left[ U(Z)\right] \geq B+2\varepsilon \left(
\sum_{j=1}^{N}\mathbb{E}_{\mathbb{P}}\left[ (Z^{j})^{+}\right]
-\sup_{n}\sum_{j=1}^{N}\mathbb{E}_{\mathbb{P}}\left[ (Z_{n}^{j})^{+}\right]
\right)\,.
\end{equation*}%
Since the term multiplying $\varepsilon $ is finite by hypothesis and the
inequality holds for all $\varepsilon >0$ we conclude that $\mathbb{E}_{%
\mathbb{P}}\left[ U(Z)\right] \geq B$.
\end{proof}
\subsection{Explicit computation of the convex conjugate in an exponential setup}

\begin{lemma}
\label{lemmacomputeV}
Take $\alpha_1,\dots,\alpha_N>0$ and consider for $x=[x^1,\dots,x^N]\in\R^N$ $$U(x^1,\dots,x^N):=\frac{1}{2}\sum_{j=1}^N \left(1-e^{-2\alpha_j x^j}\right)+\frac{1}{2}\sum_{\substack{i,j\in\{1,\dots,N\}\\i\neq j}}\left(1-e^{-(\alpha_ix^i+\alpha_jx^j)}\right)=\frac{N^2}{2}-\frac12\left(\sum_{j=1}^N e^{-\alpha_j x^j}\right)^2\,.$$
Define
$$\beta=\sum_{j=1}^N\frac{1}{\alpha_j}\,,\,\,\,\,\Gamma=\sum_{j=1}^N\frac{1}{\alpha_j}\log\left(\frac{1}{\alpha_j}\right)\,.$$
Then the conjugate $V(w)=\sup_{x\in\R^N}\left(U(x)-\sum_{j=1}^Nx^jw^j\right)$ satisfies: for any $w\in (0,+\infty)^N\cap \mathrm{dom}(V)$
\begin{align}
V(w)=\frac{N^2}{2}+\sum_{j=1}^N\left(\frac{w^j}{\alpha_j}+\frac{w^j}{\alpha^j}\log\left(\frac{w^j}{\alpha^j}\right)\right)-\frac12 \left[\sum_{j=1}^N\frac{w^j}{\alpha_j}+\left(\sum_{j=1}^N\frac{w^j}{\alpha_j}\right)\log\left(\sum_{j=1}^N\frac{w^j}{\alpha_j}\right)\right]\label{formulaexplv}
\end{align}
and for any $w=[z,\dots,z]\in (0,+\infty)^N\cap \mathrm{dom}(V)$
\begin{align}
\frac{\partial V}{\partial w^j}(z,\dots,z)&=\frac{1}{\alpha_j}+\frac{1}{\alpha_j}\log\left(\frac{1}{\alpha_j}\right)-\frac{1}{2\alpha_j}\log\left(\beta\right)+\frac{1}{2\alpha_j}\log(z)\,, \label{formulagradient}\\
\sum_{j=1}^Nz\frac{\partial V}{\partial w^j}(z,\dots,z)&=z\left[\beta+\Gamma-\frac{\beta\log(\beta)}{2}\right]+\frac{\beta}{2}z\log(z)\,.\label{sumgradient}
\end{align}
\end{lemma}

\begin{proof}
Take $w=[w^1,\dots,w^N],w^j>0\,\forall j=1,\dots,N$ with $w\in\mathrm{dom}(V)$. Then
\begin{align*}
V(w):=&\sup_{x\in\R^N}\left(U(x)-\sum_{j=1}^Nx^jw^j\right)=\frac{N^2}{2}+\sup_{x\in\R^N}\left(-\frac12\left(\sum_{j=1}^N e^{-\alpha_j x^j}\right)^2-\sum_{j=1}^Nx^jw^j\right)\\
=&\frac{N^2}{2}+\sup_{x\in\R^N}\left(-\sup_{h\geq 0}\left(\left(\sum_{j=1}^N e^{-\alpha_j x^j}\right)h-\frac12h^2\right)-\sum_{j=1}^Nx^jw^j\right)\\
=&\frac{N^2}{2}+\sup_{x\in\R^N}\inf_{h\geq 0}\left(\sum_{j=1}^N \left(-he^{-\alpha_j x^j}-x^jw^j\right)+\frac12h^2\right)\\
\stackrel{(\star)}{=}&\frac{N^2}{2}+\inf_{h\geq 0}\sup_{x\in\R^N}\left(\sum_{j=1}^N \left(-he^{-\alpha_j x^j}-x^jw^j\right)+\frac12h^2\right)
\end{align*}
where in $(\star)$ we exploited the minimax Theorem 2.4 in \cite{Sava}.
Now observe that since $w\neq 0$, if the infimum over $h$ were attained at $h=0$ we would get $V(w)=+\infty$. Hence, recalling that $w\in\mathrm{dom}(V)$,
\begin{align*}
V(w)=&\frac{N^2}{2}+\inf_{h> 0}\sup_{x\in\R^N}\left(\sum_{j=1}^N \left(-he^{-\alpha_j x^j}-x^jw^j\right)+\frac12h^2\right)\\
=&\frac{N^2}{2}+\inf_{h> 0}\sup_{x\in\R^N}\left(\sum_{j=1}^N \left(-e^{-\alpha_j \left(x^j-\frac{1}{\alpha_j}\log(h)\right)}-\left(x^j-\frac{1}{\alpha_j}\log(h)\right)w^j\right)-\sum_{j=1}^N\left(\frac{1}{\alpha_j}\log(h)\right)w^j+\frac12h^2\right)\\
=&\frac{N^2}{2}+\inf_{h>0}\left(\sum_{j=1}^N\sup_{x\in\R}\left(-e^{-\alpha_j x}-xw^j\right)-\sum_{j=1}^N\left(\frac{1}{\alpha_j}\log(h)\right)w^j+\frac12h^2\right)\\
=&\frac{N^2}{2}+\sum_{j=1}^N\sup_{x\in\R}\left(-e^{-\alpha_j x}-xw^j\right)-\frac12\sup_{h>0}\left(\left(\sum_{j=1}^N\frac{w^j}{\alpha_j}\right)2\log(h)-e^{2\log(h)}\right)\\
=&\frac{N^2}{2}+\sum_{j=1}^N\sup_{x\in\R}\left(-e^{-\alpha_j x}-xw^j\right)-\frac12\sup_{x\in\R}\left(-e^{-x}-\left(\sum_{j=1}^N\frac{w^j}{\alpha_j}\right)x\right)
\end{align*}
where in the last equality we set $2\log(h)=-x$. Now, recalling that $\sup_{x\in\R}\left(-e^{-\gamma x}-xw\right)=\frac{w}{\gamma}+\frac{w}{\gamma}\log\left(\frac{w}{\gamma}\right)$ we get from the equation above that \eqref{formulaexplv} holds. Expressions
 \eqref{formulagradient} and \eqref{sumgradient} can then be obtained by direct computation.
%
\end{proof}

\subsection{Results on multivariate Orlicz spaces}

\label{mSORTEsecmultiorliczproofs}

\begin{proof}[Proof \textit{of Proposition \protect\ref{mSORTEpropdual1}}]
We show that $K_{\Phi }$ is a subspace of the topological dual of $L^{\Phi }$
and is a subset of $(L^{1}({\mathbb{P}}))^{N}$.

For $Z\in K_{\Phi }$ consider the well defined linear map $\phi :L^{\Phi
}\rightarrow L^{1}({\mathbb{P}})$, $X\mapsto \sum_{j=1}^{N}X^{j}Z^{j}$.
Suppose $X_{n}\rightarrow X$ in $L^{\Phi }$ and $\phi (X_{n})\rightarrow W$,
then we can extract a subsequence $(X_{n_{k}})$ converging almost surely to $%
X$, since convergence in Luxemburg norm implies convergence in probability
(Lemma \ref{mSORTElemmasummary} Item 5). It is then clear that $\phi
(X_{n_{k}})=\sum_{j=1}^{N}X_{n_{k}}^{j}Z^{j}\rightarrow
_{k}\sum_{j=1}^{N}X^{j}Z^{j}=W\,{\mathbb{P}}$-a.s., thus the graph of $\phi $
is closed in $L^{\Phi }\times L^{1}({\mathbb{P}})$ (endowed with product
topology). By Closed Graph Theorem (\cite{Aliprantis} theorem 5.20) the map
is then continuous, thus any vector in $K_{\Phi }$ identifies a continuous
linear functional on $L^{\Phi }$. Finally since $[sign(Z^{j})]_{j=1}^{N}\in
L^{\infty }({\mathbb{P}})\subseteq M^{\Phi }\subseteq L^{\Phi }$, $%
\sum_{j=1}^{N}\left\vert Z^{j}\right\vert \in L^{1}({\mathbb{P}}) $ yielding 
$K_{\Phi }\subseteq L^{1}({\mathbb{P}}).$
\end{proof}

\begin{proof}[\textit{Proof of Proposition \protect\ref%
{mSORTEthmsummarykoethe} Item \protect\ref{mSORTElemmaluxenorm}}]
We show that for any extended real valued vector 
\begin{equation*}
Z\in L^{0}\left( (\Omega ,\mathcal{F},{\mathbb{P}});\allowbreak[-\infty
,+\infty ]^{N}\right)
\end{equation*}
we have 
\begin{equation}
\sup_{X\in L^{\Phi },\left\Vert X\right\Vert _{\Phi }\leq 1}\mathbb{E}_{%
\mathbb{P}}\left[ \sum_{j=1}^{N}\left\vert X^{j}Z^{j}\right\vert \right]
=\sup_{X\in M^{\Phi },\left\Vert X\right\Vert _{\Phi }\leq 1}\mathbb{E}_{%
\mathbb{P}}\left[ \sum_{j=1}^{N}\left\vert X^{j}Z^{j}\right\vert \right] \,.
\label{mSORTEeqputabs}
\end{equation}%
and that, moreover 
\begin{equation*}
K_{\Phi }=\left\{ Z\in L^{0}\left( (\Omega ,\mathcal{F},{\mathbb{P}}%
);[-\infty ,+\infty ]^{N}\right) \mid \sum_{j=1}^{N}X^{j}Z^{j}\in L^{1}({%
\mathbb{P}}),\,\forall \,X\in M^{\Phi }\right\} \,.
\end{equation*}%
Argue as in Proposition 2.2.8 of \cite{Edgar}: take any $X\in L^\Phi$ and $%
Z\in (L^0({\mathbb{P}}))^N$ and assume w.l.o.g. both are componentwise
nonnegative (multiplying by signum functions will not affect Luxemburg norms
by definition). Take sequences of simple functions $(Y_{n}^{j})_{n}$, $%
j=1,\dots ,n$ each converging to $X^{j}$ monotonically from below. Clearly $%
\left\Vert Y_{n}\right\Vert _{\Phi }\leq \left\Vert X\right\Vert _{\Phi }$
for each $n$ and by Monotone Convergence Theorem 
\begin{equation*}
\mathbb{E}_{\mathbb{P}}\left[ \sum_{j=1}^{N}\left\vert X^{j}Z^{j}\right\vert %
\right] =\lim_{n}\mathbb{E}_{\mathbb{P}}\left[ \sum_{j=1}^{N}\left\vert
Y_{n}^{j}Z^{j}\right\vert \right] \,.
\end{equation*}%
This implies that 
\begin{equation*}
\sup_{X\in L^{\Phi },\left\Vert X\right\Vert _{\Phi }\leq 1}\mathbb{E}_{%
\mathbb{P}}\left[ \sum_{j=1}^{N}\left\vert X^{j}Z^{j}\right\vert \right]
\end{equation*}%
\begin{equation*}
\leq \sup_{X\in L^{\infty },\left\Vert X\right\Vert _{\Phi }\leq 1}\mathbb{E}%
_{\mathbb{P}}\left[ \sum_{j=1}^{N}\left\vert X^{j}Z^{j}\right\vert \right]
\leq \sup_{X\in M^{\Phi },\left\Vert X\right\Vert _{\Phi }\leq 1}\mathbb{E}_{%
\mathbb{P}}\left[ \sum_{j=1}^{N}\left\vert X^{j}Z^{j}\right\vert \right]
\end{equation*}%
since $L^{\infty }\subseteq M^{\Phi }$. The converse inequality is evident,
so that \eqref{mSORTEeqputabs} follows. Now suppose 
\begin{equation*}
Z\in \left\{ Z\in L^{0}\left( (\Omega ,\mathcal{F},{\mathbb{P}});[-\infty
,+\infty ]^{N}\right) \mid \sum_{j=1}^{N}X^{j}Z^{j}\in L^{1}({\mathbb{P}}%
),\,\forall \,X\in M^{\Phi }\right\} \,.
\end{equation*}%
Observe (by using $\left\vert X^{j}\right\vert sgn(Z^{j})$ in place of $%
X^{j} $ in RHS below) that 
\begin{equation*}
\sup_{X\in M^{\Phi },\left\Vert X\right\Vert _{\Phi }\leq 1}\mathbb{E}_{%
\mathbb{P}}\left[ \sum_{j=1}^{N}\left\vert X^{j}Z^{j}\right\vert \right]
=\sup_{X\in M^{\Phi },\left\Vert X\right\Vert _{\Phi }\leq 1}\mathbb{E}_{%
\mathbb{P}}\left[ \sum_{j=1}^{N}X^{j}Z^{j}\right] <+\infty \,.
\end{equation*}

where we used a Closed Graph Theorem argument similar to the one in the
proof of Proposition \ref{mSORTEpropdual1}, with $M^\Phi$ in place of $%
L^\Phi $, to show finiteness of RHS: since $X\mapsto \sum_{j=1}^NX^jZ^j$ is
well defined and continuous on $M^\Phi$ it must have finite operator norm,
i.e. RHS. Now it follows that 
\begin{equation*}
\sup_{X\in L^{\Phi },\left\Vert X\right\Vert _{\Phi }\leq 1}\mathbb{E}_{%
\mathbb{P}}\left[ \sum_{j=1}^{N}\left\vert X^{j}Z^{j}\right\vert \right] 
\overset{\eqref{mSORTEeqputabs}}{=}\sup_{X\in M^{\Phi },\left\Vert
X\right\Vert _{\Phi }\leq 1}\mathbb{E}_{\mathbb{P}}\left[ \sum_{j=1}^{N}%
\left\vert X^{j}Z^{j}\right\vert \right] <+\infty
\end{equation*}%
which in turns provides $Z\in K_{\Phi }$.
\end{proof}

\medskip

\begin{proof}[\textit{Proof of Proposition \protect\ref%
{mSORTEthmsummarykoethe} Item \protect\ref{mSORTEthmdualhearth}}]
We prove that the topological dual of $(M^{\Phi },\left\Vert \cdot
\right\Vert _{\Phi })$ is $(K_{\Phi },\left\Vert \cdot \right\Vert _{\Phi
}^{\ast })$. By order continuity, for a given linear functional $\phi $ in
the topological dual of $M^{\Phi }$ we have that $A\mapsto \phi ([0,\dots
,0,1_{A},0,\dots ,0])$ defines a (finite) absolutely continuous measure with
respect to ${\mathbb{P}}$. This gives by Radon-Nikodym Theorem a vector $%
Z\in (L^{1})^{N} $ satisfying: for every vector of simple functions $s\in
(L^{\infty }({\mathbb{P}}))^{N}$ $\phi (s)=\sum_{j=1}^{N}\mathbb{E}_{\mathbb{%
P}}\left[ s^{j}Z^{j}\right] $ We now prove that $Z$ belongs to $K_{\Phi }$:
take $X\geq 0$ and a sequence $(Y_{n})_{n}$ of non negative simple functions
(vectors of simple functions more precisely) converging to $X$ from below.

By order continuity of the topology on $M^{\Phi }$ we have 
\begin{equation*}
\sum_{j=1}^{N}\mathbb{E}_{\mathbb{P}}\left[ sgn(Z^{j})Y_{n}^{j}Z^{j}\right]
=\phi \left( \left[ sign(Z^{j})Y_{n}^{j}\right] _{j=1}^{N}\right) %
\xrightarrow[n]{\norm{\cdot}_\Phi}\phi \left( \left[ sign(Z^{j})X^{j}\right]
_{j=1}^{N}\right) <+\infty\,.
\end{equation*}%
Thus by Monotone Convergence Theorem 
\begin{equation*}
+\infty >\lim_{n}\sum_{j=1}^{N}\mathbb{E}_{\mathbb{P}}\left[
sgn(Z^{j})Y_{n}^{j}Z^{j}\right] =\lim_{n}\sum_{j=1}^{N}\mathbb{E}_{\mathbb{P}%
}\left[ Y_{n}^{j}\left\vert Z^{j}\right\vert \right] =\sum_{j=1}^{N}\mathbb{E%
}_{\mathbb{P}}\left[ X^{j}\left\vert Z^{j}\right\vert \right] \,.
\end{equation*}%
This proves that $Z\in K_{\phi }$, since the argument above can be applied
to any $0\leq X\in M^{\Phi }$ and subsequently to any $X\in M^{\Phi }$.
Finally, the norm we use on $K_{\Phi }$ is exactly the usual one for
continuous linear functionals, so $(K_{\Phi },\left\Vert \cdot \right\Vert
_{\Phi }^{\ast })$ is isometric to the topological dual of $(M^{\Phi
},\left\Vert \cdot \right\Vert _{\Phi })$.
\end{proof}

\begin{proof}[\textit{Proof of Proposition \protect\ref%
{mSORTEthmsummarykoethe} Item \protect\ref{mSORTElemmakoetheok}}]
We show that if we suppose 
\begin{equation}
L^{\Phi }=L^{\Phi _{1}}\times \dots \times L^{\Phi _{N}}\,,
\label{mSORTEeqlisprod}
\end{equation}%
then we have that $K_{\Phi }=L^{\Phi _{1}^{\ast }}\times \dots \times
L^{\Phi _{N}^{\ast }}$. To see this, observe that%
\begin{equation*}
K_{\Phi }:=\left\{ Z\in L^{0}\left( (\Omega ,\mathcal{F},{\mathbb{P}}%
);[-\infty ,+\infty ]^{N}\right) \mid \sum_{j=1}^{N}X^{j}Z^{j}\in L^{1}({%
\mathbb{P}}),\,\forall \,X\in L^{\Phi }\right\}
\end{equation*}%
\begin{equation*}
\overset{\eqref{mSORTEeqlisprod}}{=}\left\{ Z\in L^{0}\left( {\mathbb{P}}%
;[-\infty ,+\infty ]\right)^{N} \mid \sum_{j=1}^{N}X^{j}Z^{j}\in L^{1}({%
\mathbb{P}}),\,\forall \,X\in L^{\Phi _{1}}\times \dots \times L^{\Phi
_{N}}\right\}
\end{equation*}%
\begin{equation*}
=\left\{ Z\in L^{0}\left( {\mathbb{P}};[-\infty ,+\infty ]\right)^{N} \mid
X^{j}Z^{j}\in L^{1}({\mathbb{P}}),\,\forall \,X^{j}\in L^{\Phi
_{j}},\,\forall j=1\,\dots ,N\right\} \,.
\end{equation*}%
Now apply Corollary 2.2.10 in \cite{Edgar} componentwise.
\end{proof}

\begin{proof}[\textit{Proof of Remark \protect\ref{mSORTERemMPHI}}]
To prove the claims, observe that $M^{\Phi }\subseteq M^{\Phi _{1}}\times
\dots \times M^{\Phi _{N}}$ follows from the fact that $\mathbb{E}_{\mathbb{P%
}}\left[ \Phi _{j}(\lambda \left\vert X^{j}\right\vert )\right] \leq \mathbb{%
E}_{\mathbb{P}}\left[ \Phi (\lambda \left\vert X\right\vert )\right] $,
while the converse $(\supseteq )$ is trivial.

We now prove inequalities \eqref{mSORTEnormequiv}. First observe that for $%
X\in M^{\Phi }$ and for every $j=1,\dots ,N$ the functions $\gamma \mapsto 
\mathbb{E}_{\mathbb{P}}\left[ \Phi (\frac{1}{\gamma }\left\vert X\right\vert
)\right] $ and $\gamma \mapsto \mathbb{E}_{\mathbb{P}}\left[ \Phi _{j}(\frac{%
1}{\gamma }\left\vert X^{j}\right\vert )\right] $ are continuous by
Dominated Convergence Theorem, hence for $\left\Vert X\right\Vert _{\Phi
}\neq 0$ and every $j=1,\dots ,N$

\begin{equation*}
\mathbb{E}_{\mathbb{P}}\left[ \Phi _{j}\left( \frac{1}{\left\Vert
X^{j}\right\Vert _{\Phi _{j}}}\left\vert X^{j}\right\vert \right) \right]%
\leq \mathbb{E}_{\mathbb{P}}\left[ \Phi \left( \frac{1}{\left\Vert
X\right\Vert _{\Phi }}\left\vert X\right\vert \right) \right] =1\,.
\end{equation*}%
Since also for $\left\Vert X\right\Vert _{\Phi }=0$ we have $X=0$ and as a
consequence $\left\Vert X^{j}\right\Vert _{\Phi _{j}}=0,\,j=1,\dots ,N$, we
have 
\begin{equation}
\left\Vert X^{j}\right\Vert _{\Phi _{j}}\leq \left\Vert X\right\Vert _{\Phi
}\,j=1,\dots ,N\,.  \label{mSORTEeqequiv1}
\end{equation}%
Moreover for $X\neq 0$ set $\lambda :=\max_{j}\left( \left\Vert
X^{j}\right\Vert _{\Phi _{j}}\right) $. Then 
\begin{equation*}
\mathbb{E}_{\mathbb{P}}\left[ \Phi _{j}\left( \frac{1}{N\lambda }\left\vert
X^{j}\right\vert \right) \right] \leq \frac{1}{N}\mathbb{E}_{\mathbb{P}}%
\left[ \Phi _{j}\left( \frac{1}{\lambda }\left\vert X^{j}\right\vert \right) %
\right] \leq \frac{1}{N}\,.
\end{equation*}%
Hence for $X\neq 0$ 
\begin{equation*}
\left\Vert X\right\Vert _{\Phi }\leq N\max_{j}\left( \left\Vert
X^{j}\right\Vert _{\Phi _{j}}\right)
\end{equation*}%
and the same trivially holds for $X=0$. In general then 
\begin{equation}
\left\Vert X\right\Vert _{\Phi }\leq N\max_{j}\left( \left\Vert
X^{j}\right\Vert _{\Phi _{j}}\right) \leq N\sum_{j=1}^{N}\left\Vert
X^{j}\right\Vert _{\Phi _{j}}\,.  \label{mSORTEeqequiv2}
\end{equation}%
Now inequalities \eqref{mSORTEnormequiv} follow from inequalities %
\eqref{mSORTEeqequiv1} and \eqref{mSORTEeqequiv2} and the claims are proved.
\end{proof}

\begin{lemma}
\label{mSORTElemmaevfiniteisinkphi} Let $Z\in (L^1({\mathbb{P}}))^N$ be such
that for some $\lambda>0$ $\mathbb{E}_\mathbb{P} \left[V(\lambda Z)\right]%
<+\infty $. Then $Z\in K_\Phi$.
\end{lemma}

\begin{proof}
By definition of $V$ we have for any $x,z\in {\mathbb{R}}^{N}$ $%
-\left\langle x,z\right\rangle \leq V(z)-U(X)$. Take $Z$ with $\mathbb{E}_{%
\mathbb{P}}\left[ V(\lambda Z)\right] <+\infty $ for some $\lambda >0$. For
any $X\in M^{\Phi }$ consider $\widehat{X}$ defined as 
\begin{equation*}
\widehat{X}^{j}:=-sgn(X^{j})sgn(Z^{j})X^{j},\,j=1,\dots ,N
\end{equation*}%
and observe that $\widehat{X}\in M^{\Phi }$. Moreover we have $\lambda
\left\langle \left\vert X\right\vert ,\left\vert Z\right\vert \right\rangle
=-\left\langle \widehat{X},\lambda Z\right\rangle \leq V(\lambda Z)-U(%
\widehat{X})$. If $\widehat{X}\in M^{\Phi }$ then, by %
\eqref{mSORTEassocorlicz}, $\mathbb{E}_\mathbb{P} \left[U(\widehat{X})\right]%
>-\infty .$ Since $V(\lambda Z)\in L^{1}({\mathbb{P}})$ by hypothesis, we
conclude that $\left\langle X,Z\right\rangle \in L^{1}({\mathbb{P}})$ for
every $X\in M^{\Phi }$, which in turns yields $Z\in K_{\Phi }$ by
Proposition \ref{mSORTEthmsummarykoethe} Item 1.
\end{proof}

\subsection{On Koml\'{o}s Theorem}

We now recall the original Koml\'{o}s Theorem:

\begin{theorem}[Koml\'{o}s]
\label{mSORTEkomlosoriginal} Let $(f_{n})_{n}\subseteq L^{1}((\Omega ,%
\mathcal{F},{\mathbb{P}},);{\mathbb{R}})$ be a sequence with bounded $L^{1}$
norms. Then there exists a subsequence $(f_{n_{k}})_{k}$ and a $g$ again in $%
L^{1}$ such that for any further subsequence the C\'{e}saro means satisfy: 
\begin{equation*}
\frac{1}{N}\sum_{i\leq N}f_{n_{k_{i}}}\rightarrow g\,\,\,{\mathbb{P}}-\text{
a.s. as }N\rightarrow +\infty\,.
\end{equation*}
\end{theorem}

\begin{proof}
See \cite{Komlos} Theorem 1a.
\end{proof}

\begin{corollary}
\label{mSORTEcorkomplosmultidim} Let a sequence $(Y_n)_n$ be given in $L^1({%
\mathbb{P}}_1)\times\dots\times L^1({\mathbb{P}}_N)$ such that for
probabilities ${\mathbb{P}}_1,\dots,{\mathbb{P}}_N\ll{\mathbb{P}}$ 
\begin{equation*}
\sup_n\sum_{j=1}^N\mathbb{E}_\mathbb{P} \left[\left|Y_n^j\right|\frac{%
\mathrm{d}{\mathbb{P}}_j}{\mathrm{d}{\mathbb{P}}}\right]<\infty\,.
\end{equation*}
Then there exists a subsequence $(Y_{n_h})_h$ and an $\widehat{Y}\in L^1({%
\mathbb{P}}_1)\times\dots\times L^1({\mathbb{P}}_N)$ such that every further
subsequence $(Y_{n_{h_k}})_k$ satisfies 
\begin{equation*}
\frac1K\sum_{k=1}^K Y_{n_{h_k}}^j\rightarrow\widehat{Y}^j\,\,\,\,{{\mathbb{P}%
}_j-\text{a.s.}} \,\,\,\,\forall\,\,j=1,\dots N\,\text{ as }%
K\rightarrow+\infty\,.
\end{equation*}
\end{corollary}

\begin{proof}
We suppose $N=2$, the argument can be iterated. The result follows from a
diagonal argument: take the first component, we have a subsequence and an $%
\widehat{Y}^1$ s.t. each further subsequence has ${\mathbb{P}}_1$-a.s.
converging C\'{e}saro means as in Theorem \ref{mSORTEkomlosoriginal}. Now
take this sequence in place of the one we began with, and do the same for
the second component. Notice that in the end we get a subsequence for the
second component too, and the corresponding indices yield a subsequence of
the one we extracted for the first component. The claim follows.
\end{proof}

{\small 
\bibliographystyle{abbrv}
\bibliography{BibAll}
}

\end{document}